%% file: higher-order-abstract-gsos-jfp.tex
\PassOptionsToPackage{nosumlimits,nonamelimits}{amsmath}

\documentclass[final]{jfp-epi}
\include{preamble}

\usepackage[capitalize,noabbrev,nameinlink]{cleveref}
\crefname{rem}{Remark}{Remarks} 
\crefname{assumptions}{Assumptions}{Assumptions} 

\hypersetup{colorlinks,linkcolor={blue},citecolor={blue},urlcolor={red},breaklinks=true,final}

\makeatletter
\renewcommand{\qedhere}{\ensuremath{\qed}}
\renewcommand{\noqed}{\def\qed{}}
\makeatother

\renewcommand{\xto}[1]{\mathrel{\raisebox{-.75pt}{$\xrightarrow{\;\smash{\raisebox{-1.75pt}{{\scriptsize $#1$}}\;}}$}}}

\begin{document}\allowdisplaybreaks

\title{Higher-order bialgebraic semantics}

\author{Sergey Goncharov}
	\orcid{0000-0001-6924-8766}
	\affiliation{%
		\institution{School of Computer Science, University of Birmingham}
		\city{Birmingham}
		\country{UK}
		\authoremail{s.goncharov@bham.co.uk}
	}

\author{Stefan Milius}
	\orcid{0000-0002-2021-1644}
	\affiliation{%
		\institution{Friedrich-Alexander-Universität Erlangen-Nürnberg}
		\city{Erlangen}
		\country{Germany}
		\authoremail{stefan.milius@fau.de}
	}

\author{Lutz Schr{\"o}der}
	\orcid{0000-0002-3146-5906}
	\affiliation{%
		\institution{Friedrich-Alexander-Universität Erlangen-Nürnberg}
		\city{Erlangen}
		\country{Germany}
		\authoremail{lutz.schroeder@fau.de}
	}

\author{Stelios Tsampas}
	\orcid{0000-0001-8981-2328}
	\affiliation{%
		\institution{Centre for Formal Methods and Future Computing, Syddansk Universitet}
		\city{Odense}
		\country{Denmark}
		\authoremail{stelios@imada.sdu.dk}
	}

\author{Henning Urbat}
	\orcid{0000-0002-3265-7168}
	\affiliation{%
		\institution{Friedrich-Alexander-Universität Erlangen-Nürnberg}
		\city{Erlangen}
		\country{Germany}
		\authoremail{henning.urbat@fau.de}
	}

\jfpVolume{36}
\jfpArticle{1}
\jfpDOI{10.46298/jfp.17738}
\jfpYear{2026}
\received[Submitted]{April 2024}
\received[accepted]{March 2026}

\begin{abstract}
Compositionality proofs in higher-order languages are notoriously involved, and general semantic frameworks guaranteeing compositionality are hard to come by. In particular, Turi and Plotkin's bialgebraic abstract GSOS framework, which provides off-the-shelf compositionality results for first-order languages, so far does not apply to higher-order languages. In the present work,  we develop a theory of abstract GSOS specifications for higher-order languages, in effect transferring the core principles of Turi and Plotkin's framework to a higher-order setting. In our theory, the operational semantics of higher-order languages is represented by certain dinatural transformations that we term \emph{(pointed) higher-order GSOS laws}. We give a general compositionality result that applies to all systems specified in this way and discuss how compositionality of combinatory logics and the $\lambda$-calculus w.r.t.\ a strong variant of Abramsky's applicative bisimilarity are obtained as instances.
\end{abstract}

\maketitle

\section{Introduction}
\label{sec:intro}

The framework  of \emph{Mathematical Operational
  Semantics}, introduced by~\cite{DBLP:conf/lics/TuriP97}, elucidates the operational semantics of programming languages, and guarantees 
compositionality of programming language semantics in all cases that it covers.
In this framework, operational semantics are presented as distributive laws,
varying in complexity, of a monad over a comonad in a suitable category. An
important example is that of \emph{GSOS laws}, i.e. natural transformations of type
\begin{equation*}
  \rho_X \colon \Sigma (X \product BX) \to B\Sigmas X, 
\end{equation*}
with endofunctors $\Sigma, B \c \gcat \to \gcat$ respectively specifying the
\emph{syntax} and \emph{behaviour} of the system at hand. The idea is that a GSOS law  represents a set of inductive transition rules that specify how programs are run. For instance, the choice of $\gcat = \Set$
and $B = (\mypowfin)^{L}$, where~$\mypowfin$ is the finite powerset
functor and $L$ a set of transition labels, leads to the
well-known GSOS rule format by \cite{DBLP:journals/jacm/BloomIM95} for specifying labeled transition
systems. For that reason, Turi and Plotkin's framework is
often referred to as \emph{abstract GSOS}.

The semantic interpretation of GSOS laws is conveniently presented in a bialgebraic setting (cf.\ \Cref{sec:abstract-gsos}). Every GSOS law $\rho$  \eqref{eq:rho} canonically induces a bialgebra
\[ \Sigma(\mS)\xto{\ini} \mS \xto{\gamma} B(\mS) \]
on the object $\mS$ of programs freely generated by the syntax functor $\Sigma$,
where the algebra structure~$\ini$ inductively constructs programs and the
coalgebra structure $\gamma$ describes their one-step behaviour according to the
given law $\rho$. The above bialgebra is thus the \emph{operational model} of
$\rho$.
Dually, its \emph{denotational model} is a bialgebra
\[ \Sigma(\nu B)\xto{\alpha} \nu B \xto{\tau} B(\nu B),\]
which extends the final coalgebra $\nu B$ (to be
thought of as the domain of abstract program behaviours) of the behaviour functor $B$. Both the operational
and the denotational model are characterized by universal properties, namely
as the \emph{initial $\rho$-bialgebra} and the \emph{final $\rho$-bialgebra}, respectively.
This immediately entails a key feature of abstract GSOS: The semantics is
automatically \emph{compositional}, in that
behavioural equivalence (e.g.\ bisimilarity) of programs is a congruence with respect to the operations of the
language. The bialgebraic framework has been used widely to establish
further correspondences and obtain compositionality
results, see e.g.\ the work of~\cite{56f40c248cb44359beb3c28c3263838e}, 
  \cite{DBLP:conf/fossacs/KlinS08}, \cite{DBLP:conf/lics/FioreS06}, 
  \cite{DBLP:journals/tcs/MiculanP16}, and \cite{DBLP:conf/fscd/0001MS0U22}.

As a first step towards extending the abstract GSOS framework to languages with
\emph{variable binding}, such as the $\pi$-calculus by \cite{DBLP:journals/iandc/MilnerPW92a} and the
$\lambda$-calculus, \cite{DBLP:conf/lics/FiorePT99} use the theory of
\emph{presheaves} to establish an abstract categorical foundation of
syntax with variable binding, and develop a theory of capture-avoiding
substitution in this abstract setting. Based on these foundations,
the semantics of \emph{first-order} languages with variable binding,
more precisely that of the $\pi$-calculus and value-passing
CCS, see \cite{DBLP:books/daglib/0067019}, is formulated in terms of
GSOS laws on categories of presheaves by \cite{DBLP:conf/lics/FioreT01}. We also
introduce higher-order bialgebras and construct the initial such bialgebra.

However, the question of the mathematical operational semantics of the
$\lambda$-calculus, or generally that of higher-order languages, still
remains a well-known issue in the literature (see e.g.~the
introductory paragraph by
\cite{DBLP:journals/lmcs/HirschowitzL22}). Indeed, in order to give
the semantics of a higher-order language in terms of some sort of a
distributive law of a syntax functor over some choice of a behaviour
functor, one needs to overcome a number of fundamental problems. For
instance, for a generic set $X$ of programs, the most obvious set of
``higher-order behaviours over $X$'' would be $X^{X}$, the set of
functions that expect an input program in~$X$ and
produce a new program in $X$. Of course, the assignment $X \mapsto X^{X}$ is
not functorial in $X$ but bifunctorial; more precisely, it yields a bifunctor
\[B(X,Y) = Y^{X} \c \Set^{\opp} \product \Set \to \Set\] of mixed
variance.  Working with mixed variance bifunctors as a basis for
higher-order behaviour makes the situation substantially more complex
in comparison to Turi and Plotkin's original setting. In particular,
natural transformations alone will no longer suffice as the technical
basis of a framework involving mixed variance functors, and it is not
a priori clear what the right notion of coalgebra for a mixed variance
functor should be. In this paper, we address these issues, with a view
to obtaining a general congruence result.

\paragraph*{Contributions}~ We develop a theory of abstract GSOS for
higher-order languages, extending Turi and Plotkin's original
first-order framework. We model such languages abstractly in terms of
syntax endofunctors of the form $\Sigma=V+\Sigma'\c \C\to \C$ (for an endofunctor $\Sigma'\colon \C\to \C$ representing the constructors of the language and a choice of an object $V \in \gcat$ to be thought of as an object of variables), and behaviour
bifunctors $B \c \gcat^{\opp} \product \gcat \to \gcat$. The key concept introduced in our paper is that of a \emph{$\Pt$-pointed higher-order GSOS law}: a family of morphisms
\[
  \rho_{X,Y} \c \Sigma(jX \product B(jX,Y)) \to B(jX,\Sigmas(jX + Y))
\]
\emph{dinatural} in $X \in \Pt/\gcat$ and \emph{natural} in
$Y \in \gcat$, with $j \c \Pt/\gcat \to \gcat$ denoting
the forgetful functor from the coslice category $\Pt/\gcat$ to $\gcat$.
Similar to the first-order case, we think of a higher-order GSOS law as encoding the set of inductive small-step operational rules of a higher-order language. We show that each $\Pt$-pointed higher-order GSOS law inductively
determines an operational model given by a \emph{higher-order bialgebra}
\[ \Sigma(\mS)\xto{\ini} \mS \xto{\gamma} B(\mS,\mS) \]
on the initial algebra $\mS$ of program terms. In analogy to the
first-order case, the operational model is the initial higher-order bialgebra for the given higher-order GSOS law.

From a coalgebraic standpoint, the morphism
$\gamma \c \mS \to B(\mS,\mS)$ is a coalgebra for the restricted
\emph{endofunctor} $B(\mS, \argument) \c \gcat \to \gcat$. Our
semantic domain of choice is the final $B(\mS, \argument)$-coalgebra
$(Z,\zeta)$, the object of abstract behaviours determined by the functor
$B(\mS, \argument)$. We obtain a morphism
$\coiter \gamma \c \mS \to Z$ by coinductively extending $\gamma$;
that is, we take the unique coalgebra morphism into the final coalgebra:
\begin{equation*}
  \begin{tikzcd}
    \mS
    \arrow[r, "\gamma"]
    \ar[dashed]{d}[swap]{\coiter \gamma}
    &
    B(\mS, \mS)
    \arrow[d, "B(\mS{,} \coiter \gamma)"]
    \\
    \finalc
    \arrow[r, "\zeta"]
    & B(\mS, \finalc)
  \end{tikzcd}
\end{equation*}
Importantly, in contrast to first-order abstract GSOS, the final coalgebra
$(Z,\zeta)$ generally does \emph{not} extend to a final higher-order bialgebra; 
in fact, a final bialgebra does not usually seem to exist (\Cref{ex:final-bialgebra}). 
As a consequence, proving compositionality in our higher-order setting is substantially more
challenging than in the first-order case, requiring entirely new techniques and 
additional assumptions on the base category $\C$ and the functors $\Sigma$ and $B$.
Specifically, we investigate higher-order GSOS laws in a \emph{regular} category~$\C$. 

As our main compositionality result, we show that the kernel pair of
$\coiter (\gamma) \c \mS \to \finalc$ (which under mild
conditions coincides with the coalgebraic bisimilarity relation) is a congruence. We
demonstrate the expressiveness of higher-order abstract GSOS by modeling
two important examples of higher-order systems. We draw our first
example, the \emph{SKI combinator calculus} by \cite{10.2307/2370619},
from the world of combinatory logic, which we represent using a higher-order GSOS law on the category of sets. For our second
example we move on to a category of presheaves, on which we model the
call-by-name and the call-by-value $\lambda$-calculus.  In all of these
examples, we demonstrate that the induced semantics corresponds to
{strong} variants of \emph{applicative
  bisimilarity}, originally introduced by \cite{Abramsky:lazylambda}.

\paragraph*{Organization}~
In~\Cref{sec:prelim} we provide a brief introduction to the core
categorical concepts that are used throughout this paper. Moving on,
in~\Cref{sec:sets} we discuss examples from combinatory logic and
present a basic rule format for higher-order languages that
illustrates the principles behind our approach.  \Cref{sec:hogsos} is
where we define our key notion of pointed higher-order GSOS law and
prove our main compositionality result
(\Cref{th:main}). In~\Cref{sec:lam} we show how to implement the call-by-name and
call-by-value $\lambda$-calculus in our abstract framework.  We
conclude the paper with a
discussion of further developments and potential avenues for future work in \Cref{sec:concl}.

\paragraph*{Related work}~
Formal reasoning on higher-order languages is a long-standing research
topic
(e.g.~\cite{3720}, \cite{Abramsky:lazylambda}, \cite{DBLP:journals/iandc/AbramskyO93}).
The series of workshops on \emph{Higher Order Operational Techniques
  in Semantics}, see \cite{10.5555/309656}, played an important role in
establishing the so-called \emph{operational methods} for higher-order
reasoning. The two most important such methods are \emph{logical
  relations} (see \cite{tait1967intensional}, \cite{DBLP:journals/iandc/Statman85},
  \cite{DBLP:journals/iandc/OHearnR95}, \cite{DBLP:journals/corr/abs-1103-0510}) and
\emph{Howe's method} (see \cite{DBLP:conf/lics/Howe89,
  DBLP:journals/iandc/Howe96}), both of which remain in use to date.
Other significant contributions towards reasoning on higher-order
languages were made by
\mbox{\citet{DBLP:journals/iandc/Sangiorgi94, DBLP:journals/iandc/Sangiorgi96}} and
\citet{DBLP:conf/lics/Lassen05}.
 While GSOS-style
frameworks ensure compositionality for free by mere adherence to given
rule formats, both logical relations and Howe's method instead have
the character of robust but inherently complex methods whose
instantiation requires considerable effort.

Recently, notable progress has been made towards generalizing Howe's
method by
\cite{DBLP:journals/lmcs/HirschowitzL22} and \cite{DBLP:conf/lics/BorthelleHL20},
based on previous work on \emph{familial monads} and operational
semantics by \cite{DBLP:journals/pacmpl/Hirschowitz19}. According to the
authors, their approach departs from Turi and Plotkin's bialgebraic
framework exactly because that framework did not cover
higher-order languages at the time. \cite{DBLP:conf/lics/LagoGL17} give a general account of
congruence proofs, and specifically Howe's method, for applicative
bisimilarity for $\lambda$-calculi with algebraic effects, based on
the theory of relators. \cite{DBLP:journals/entcs/HermidaRR14} present a foundational
account of logical relations as \emph{structure-preserving} relations
in a reflexive graph category.

Rule formats like the GSOS format by \cite{DBLP:journals/jacm/BloomIM95}
have been very useful for guaranteeing congruence of bisimilarity at a
high level of generality. However, rule formats for higher-order
languages have been scarce.  An important example is that of
the \emph{promoted tyft/tyxt} rule
format, see~\cite{DBLP:conf/lics/Bernstein98} and \cite{DBLP:journals/entcs/MousaviR07},
which has similarities to our presentation of combinatory logic
in~\Cref{sec:sets}, but it is unclear whether or not it is an instance of a
general, categorical format.  The rule format of
\citet{DBLP:journals/iandc/Howe96} was presented in the context of
Howe's method. A variant of Howe's format was recently developed by
\citet{DBLP:journals/lmcs/HirschowitzL22}.

The present paper is an extended and fully revised version of our contribution to POPL 2023, see~\cite{gmstu23}. In comparison to the latter, we include additional details and explanations that had to be omitted for lack of space, in particular full proofs, which we hope will further illuminate our results. 

Subsequent work following our POPL 2023 paper is discussed in \Cref{sec:concl}.

\section{Preliminaries}\label{sec:prelim}

\subsection{Category theory}\label{sec:categories}
We assume familiarity with basic notions from category theory
such as limits and colimits, functors, natural transformations, and monads; see e.g.~\cite{mac2013categories}. For the convenience of the reader, we review some terminology and notation used in the paper.

\paragraph*{Products and coproducts}~ Given
 objects $X_1, X_2$ in a category~$\C$, we write $X_1\times X_2$ for their product, 
 $\fst\colon X_1 \times X_2 \to X_1$ and $\snd\colon X_1\times X_2\to X_2$ for the projections, and $\brks{f_1,f_2}\colon Y \to X_1 \times X_2$ for the pairing of
morphisms $f_i\colon Y \to X_i$, $i = 1,2$. Dually, we write $X_1+X_2$ for the coproduct, $\inl\c X_1\to X_1+X_2$ and $\inr\c X_2\to X_1+X_2$ for the injections, and $[g_1,g_2]\c X_1+X_2\to Y$ for the copairing of morphisms $g_i\colon X_i\to Y$, $i=1,2$. Moreover, we let $\nabla=[\id_X,\id_X]\colon X+X\to X$ denote the codiagonal.

\paragraph*{Dinatural transformations}~
Given bifunctors $F,G\colon \C^\opp\times \C\to \D$ of mixed variance, a \emph{dinatural transformation} from $F$ to $G$ is a family of morphisms \[\sigma_X\colon F(X,X)\to G(X,X)\qquad (X\in \C)\] such that for every morphism $f\colon X\to Y$ of~$\C$, the hexagon below commutes.
\[
\begin{tikzcd}
  & F(X,X) \ar{r}{\sigma_X} & G(X,X) \ar{dr}{G(\id,f)} & \\
  F(Y,X) \ar{ur}{F(f,\id)} \ar{dr}[swap]{F(\id,f)} & & & G(X,Y) \\
  & F(Y,Y) \ar{r}{\sigma_Y} & G(Y,Y) \ar{ur}[swap]{G(f,\id)} & 
\end{tikzcd}
\]

\paragraph*{Regular categories}~
A \emph{regular epimorphism} is an epimorphism that is the coequaliser of some
pair of parallel morphisms.
A category is \emph{regular} if (1)~it has finite limits, (2)~for
every morphism $f\c A\to B$, the kernel pair $p_1,p_2\colon E\to A$ of
$f$ has a coequalizer, and (3)~regular epimorphisms are stable under
pullback.  In a regular category, every morphism $f\colon A\to B$
admits a factorization $\begin{tikzcd}A \ar[two heads]{r}{e} & C \ar[tail]{r}{m} & B\end{tikzcd}$ into a regular
epimorphism $e$ followed by a monomorphism $m$; specifically, $e$ is the
coequalizer of the kernel pair of $f$\footnote{The kernel pair of $f\colon A\to B$ is the pullback $p_1,p_2\colon K\to A$ of $f$ along itself.}, and $m$ is the unique factorizing
morphism. Indeed, the main purpose of regular categories is to provide
a notion of image factorization of morphisms that relates to kernels
of morphisms in a similar way as in set theory and universal algebra. Examples of regular
categories include the category $\Set$ of sets and functions, the
category $\Set^\C$ of (covariant) presheaves on a small category $\C$ and natural transformations, more generally every elementary topos, and every category of algebras over a
signature (see below). In all these cases, regular epimorphisms and
monomorphisms are the (componentwise) surjective and
injective morphisms, resp.

\paragraph*{Algebras}~ Given an endofunctor $\Sigma$ on a category $\gcat$,
a \emph{$\Sigma$-algebra} is a pair $(A,a)$ of an object~$A$
(the \emph{carrier} of the algebra) and a morphism $a\colon \Sigma A\to A$
(its \emph{structure}). A \emph{morphism} from
$(A,a)$ to a $\Sigma$-algebra $(B,b)$ is a morphism $h\colon A\to B$
of~$\gcat$ such that $h\comp a = b\comp \Sigma h$. (We denote composition of morphisms by `$\cdot$'.) Algebras for $\Sigma$ and
their morphisms form a category $\Alg(\Sigma)$, and an \emph{initial}
$\Sigma$-algebra is simply an initial object in that category.  We denote
the initial $\Sigma$-algebra's carrier $\mu \Sigma$ if it exists, and its structure by
$\ini\colon \Sigma(\mu \Sigma) \to \mu \Sigma$. Moreover, we write
\[\iter(a)\colon (\mu \Sigma,\ini) \to (A,a)\] for the unique morphism
from~$\mu \Sigma$ to the algebra~$(A,a)$; we often drop parentheses and write $\iter a$ for $\iter(a)$.

More generally, a \emph{free $\Sigma$-algebra} on an object $X$ of $\C$ is a
$\Sigma$-algebra $(\Sigma^{\star}X,\iota_X)$ together with a morphism
$\eta_X\c X\to \Sigma^{\star}X$ of~$\C$ such that for every algebra $(A,a)$
and every morphism $h\colon X\to A$ in $\C$, there exists a unique
$\Sigma$-algebra morphism $h^\star\colon (\Sigma^{\star}X,\iota_X)\to (A,a)$
such that $h=h^\star\comp \eta_X$. If free algebras
exist on every object, then their formation induces a monad
$\Sigma^{\star}\colon \C\to \C$, the \emph{free monad} on $\Sigma$,
see \cite{Barr70}. Every $\Sigma$-algebra $(A,a)$ induces an
Eilenberg--Moore algebra $\hat a \colon \Sigma^{\star} A \to A$, viz.\ $\hat a =\id_A^\star$ for the identity morphism $\id_A\c A\to A$.

\begin{example}[Algebras over a signature]
The most familiar example of functor algebras is that of algebras over a
signature.  An \emph{algebraic signature} consists of a set $\Sigma$
of operation symbols together with a map $\ar\colon \Sigma\to \Nat$
associating to every operation symbol $\f$ its \emph{arity}
$\ar(\f)$. Symbols of arity $0$ are called \emph{constants}. Every
signature~$\Sigma$ induces the polynomial set functor
$X\mapsto \coprod_{\f\in\Sigma} X^{\ar(\f)}$, which we
denote by the same letter $\Sigma$. We identify an element $(x_1,\ldots,x_{\ar(\f)})$ of the summand $X^{\ar(\f)}$ with the flat $\Sigma$-term $\f(x_1,\ldots,x_{\ar(\f)})$. An algebra for the functor
$\Sigma$ then is equivalently an algebra for the signature $\Sigma$,
i.e.~a set $A$ equipped with an operation $\f^A\colon A^n \to A$ for
every $n$-ary operation symbol $\f\in \Sigma$. Morphisms between
$\Sigma$-algebras are maps respecting the algebraic structure.

Given a set $X$ of variables, the free algebra $\Sigmas X$ is the
$\Sigma$-algebra of terms generated by $\Sigma$ with variables from
$X$. In
particular, the free algebra on the empty set is the initial algebra~$\mu \Sigma$, 
and it is formed by all \emph{closed terms} of the
signature.
For every $\Sigma$-algebra $(A,a)$, the induced
Eilenberg--Moore algebra $\hat a\colon \Sigmas A \to A$ is given by the map evaluating terms over $A$ in the algebra.
\end{example}

A relation ${\sim}\seq A\times A$ on a $\Sigma$-algebra
$A$ is called a \emph{congruence} if, for every $n$-ary operation
symbol $\f\in \Sigma$ and all elements $a_i,b_i\in A$
($i=1,\ldots,n$),
\[ a_i \sim b_i \quad(i=1,\ldots,n) \qquad\text{implies}\qquad
  \f^A(a_1,\ldots,a_n)\sim \f^A(b_1,\ldots,b_n). \] 
Note that unlike other authors, we do not require congruences to be equivalence relations. An equivalence relation $\sim$ is a congruence if and only if it is the \emph{kernel} of some morphism, i.e.\ there
exists a morphism $h\colon A\to B$ to some $\Sigma$-algebra $B$ such that
\[ a\sim b \qquad\text{iff}\qquad h(a)=h(b). \]

\paragraph*{Coalgebras}~Dually to the notion of algebra, a \emph{coalgebra} for an
endofunctor $B$ on $\gcat$ is a pair $(C,c)$ of an object $C$ (the
\emph{state space}) and a morphism $c\colon C\to BC$ (its
\emph{structure}). A \emph{morphism} from
$(C,c)$ to a $B$-coalgebra $(D,d)$ is a morphism
$h\colon C\to D$ of $\C$ such that $Bh\comp c = d\comp h$.
Coalgebras for $B$ and their morphisms form a category $\Coalg(B)$, and a
\emph{final} $B$-coalgebra is a final object in that category. If
it exists, we denote the corresponding state space and the structure as $\nu B$ and 
$\ter\colon \nu B \to B(\nu B)$. Moreover, for each coalgebra $(C,c)$ we denote the unique coalgebra morphism by
\[\coiter (c)\colon (C,c)\to (\nu B, \ter).\]
Again, we often write $\coiter c$ for $\coiter(c)$. Informally, a coalgebra is a categorical abstraction of a state-based system. The final coalgebra $\nu B$ is the domain of all possible abstract behaviours that $B$-coalgebras may expose, and $\coiter(c)$ sends every state of a coalgebra to its abstract behaviour. Accordingly, given coalgebras $(C,c)$ and $(C',c')$ for a functor $B\colon \Set\to\Set$, two states $x\in C$ and $x'\in C'$ are \emph{behaviourally equivalent}, denoted $x\equiv x'$, if they are identified in the final coalgebra:
\[ x\equiv x' \iff \coiter(c)(x)=\coiter(c')(x)\]
Behavioural equivalence is closely related to the notion of {bisimilarity}. A \emph{bisimulation} between coalgebras $(C,c)$ and $(C',c')$ is a relation $R\seq C\times C'$ that can be equipped with a coalgebra structure $r\colon R\to BR$ such that both projection maps $p\colon (R,r)\to (C,c)$ and $p'\colon (R,r)\to (C',c')$ are coalgebra morphisms. Two states $x\in C$ and $x'\in C'$ are \emph{bisimilar}, denoted $x\sim x'$, if they are contained in some bisimulation. For functors $B$ preserving weak pullbacks, bisimilarity coincides with behavioural equivalence,
see \cite{DBLP:journals/tcs/Rutten00}:
\[ x\sim x' \iff x\equiv x'. \]\par  
\begin{example}\label{ex:lts}
A coalgebra $c\colon C\to (\Pow_\omega C)^L$ for the functor $BX=(\Pow_\omega X)^L$ on $\Set$, where~$\Pow_\omega$ is the finite power set functor and $L$ is a fixed set of labels, is precisely a labeled transition system (LTS). Given LTS $(C,c)$ and $(C',c')$, a relation $R\seq C\times C'$ is a bisimulation if it preserves and reflects transitions, i.e.\ for each pair $R(x,x')$ and each label~$a\in L$,
\begin{itemize}
\item if $x\xto{a} y$ then there exists $y'\in C'$ such that $x'\xto{a} y'$ and $R(y,y')$;
\item if $x'\xto{a} y'$ then there exists $y\in C$ such that $x\xto{a} y$ and $R(y,y')$.
\end{itemize}
Since the functor $B$ preserves weak pullbacks, bisimilarity is behavioural equivalence.
\end{example}

\subsection{Abstract GSOS}\label{sec:abstract-gsos}
In the following we briefly review the categorical \emph{abstract GSOS} framework by \cite{DBLP:conf/lics/TuriP97} for modeling the operational semantics of (first-order) languages. We will refer to it as \emph{first-order abstract GSOS} for distinction with the higher-order extension developed in \Cref{sec:hogsos}. The framework is parametric in two endofunctors $\Sigma, B\colon \C\to \C$ on a category $\C$ with binary products, where $\Sigma$ is assumed to have an initial algebra $\mS$ and to generate a free monad~$\Sigmas$.
The functors $\Sigma$ and $B$ represent the \emph{syntax} and \emph{behaviour} of a language; in particular,~$\mS$ is regarded as the algebra of closed program terms. In abstract GSOS, the (small-step) operational semantics of a language is specified by a \emph{GSOS law of $\Sigma$ over $B$}, viz.\ a natural transformation 
\begin{equation}\label{eq:rho} \rho_X\colon \Sigma(X\times BX)\to B\Sigmas X \quad (X\in \C). 
\end{equation}
Informally, $\rho$ encodes the operational rules of the underlying language: for every constructor~$\f$ of the language, it specifies the one-step-behaviour of programs $\f(-,\cdots,-)$, i.e.\ the $\Sigma$-terms they transition into next, depending on the one-step behaviours of all the operands.
\begin{example}
Consider a process algebra with a parallel composition operator $(p,q)\mapsto p\parallel q$ specified by the operational rules
\begin{equation}\label{eq:rules-par} \frac{p\xto{a} p'}{p \parallel q\xto{a} p' \parallel q} \qquad\qquad \frac{q\xto{a} q'}{p \parallel q\xto{a} p \parallel q'} 
\end{equation}
where $a$ ranges over a fixed set $L$ of action labels. To model this specification in abstract GSOS, one takes the polynomial functor $\Sigma X=X\times X$ corresponding to the binary operator~$\parallel$, and the behaviour functor $BX=(\Pow_\omega X)^L$ representing labeled transition systems. The rules \eqref{eq:rules-par} induce a GSOS law of $\Sigma$ over $B$, i.e.\ a natural transformation 
\[ \rho_X\colon X\times (\Pow_\omega X)^L \times X \times (\Pow_\omega X)^L \to (\Pow_\omega(\Sigmas X))^L \qquad (X\in \Set), \]
whose components simply encode the two rules into a function:
\[ \rho_X(p,f,q,g) = \lambda a.\, \{ p'\parallel q : p'\in f(a) \} \cup \{ p\parallel q' : q'\in g(a) \}.    \]
\cite{DBLP:conf/lics/TuriP97} observed more generally that GSOS laws of polynomial functors~$\Sigma$ over the functor $BX=(\Pow_\omega X)^L$ correspond to specifications in the \emph{GSOS rule format} by \cite{DBLP:journals/jacm/BloomIM95}. The rules \eqref{eq:rules-par} for parallel composition are instances of GSOS rules.
\end{example} 
Every GSOS law $\rho$ canonically induces an operational and a denotational model, which both form \emph{bialgebras} for the given law $\rho$. Formally, a \emph{$\rho$-bialgebra} $(X,a,c)$ consists of an object $X\in \C$, a $\Sigma$-algebra $a\c \Sigma X\to X$ and a $B$-coalgebra $c\c X\to BX$ such that the left-hand diagram below commutes. A \emph{morphism} from $(X,a,c)$ to a $\rho$-bialgebra $(X',a',c')$ is a $\C$-morphism $h\c X\to X'$ that is both a $\Sigma$-algebra morphism and a $B$-coalgebra morphism, i.e.\ the right-hand diagram commutes:
\[
\begin{tikzcd}
\Sigma X \ar{r}{a} \ar{d}[swap]{\Sigma\langle \id,\,c\rangle} & X \ar{r}{c} & BX \\
\Sigma(X\times BX) \ar{rr}{\rho_X} & & B\Sigmas X \ar{u}[swap]{B\hat a}  
\end{tikzcd}
\qquad\qquad
\begin{tikzcd}
\Sigma X \ar{r}{a} \ar{d}[swap]{\Sigma h} & X \ar{d}{h} \ar{r}{c} & BX \ar{d}{Bh} \\
\Sigma X' \ar{r}{a'} & X' \ar{r}{c'} & BX'  
\end{tikzcd}
\]
We think of a $\rho$-bialgebra $(X,a,c)$ as a \emph{model} of the law $\rho$: the algebra $a\colon \Sigma X\to X$ interprets the operations of the language, and the coalgebra $c\colon X\to BX$ is a transition system whose transitions are given by the operational rules specified by $\rho$. 

The universal property of the initial algebra $(\mS,\iota)$ entails that there
exists a unique $B$-coalgebra structure $\gamma\c \mS\to B(\mS)$ such that
$(\mS,\iota,\gamma)$ is a $\rho$-bialgebra. This is the initial
$\rho$-bialgebra, i.e.\ the initial object in the category of bialgebras and
their morphisms. It is called the \emph{operational model} of $\rho$, as it is thought of as the transition system on terms specified by the rules corresponding to $\rho$. 

 Dually, if $B$ has a final coalgebra $(\nu B,\tau)$, it
uniquely extends to a $\rho$-bialgebra $(\nu B,\alpha,\tau)$. This is the
final $\rho$-bialgebra, and it is called the \emph{denotational model} of $\rho$. We think of $\nu B$ as the object of abstract behaviours of systems of type $B$, and of the denotational model as interpreting the operations of the language at the level of abstract behaviours.

It follows that, both by initiality and finality, there exists a unique bialgebra morphism $\beh_\rho$ from the operational model
$(\mS,\iota,\gamma)$ to the denotational model $(\nu B,\alpha,\tau)$:
\[
\begin{tikzcd}
\Sigma (\mS) \ar{r}{\ini} \ar{d}[swap]{\Sigma \beh_\rho} & \mS \ar[dashed]{d}{\beh_\rho} \ar{r}{\gamma} & B(\mS) \ar{d}{B\beh_\rho} \\
\Sigma(\nu B) \ar{r}{\alpha} & \nu B \ar{r}{\tau} & B(\nu B)  
\end{tikzcd}
\]
The map $\beh_\rho$ gives the denotational semantics of the language: it assigns to each program in $\mS$ its abstract behaviour in $\nu B$. 

 For $\C=\Set$ and a polynomial functor $\Sigma$, the fact that $\beh_\rho$ is a
 $\Sigma$-algebra morphism immediately implies that \emph{behavioural equivalence},
 namely the relation $\equiv$ on $\mS$ given by \[p\equiv q \quad\iff\quad
 \beh_\rho(p)=\beh_\rho(q),\]
is a congruence on the initial algebra $\mS$. This means that the behaviour of a program term $\f(p_1,\ldots,p_n)$ is completely determined by the behaviour of its subterms $p_1,\ldots,p_n$; in other words, the denotational semantics
 induced by $\rho$ is \emph{compositional}.

\section{Combinatory logic}
\label{sec:sets}
Our goal in the sequel is to extend Turi and Plotkin's abstract GSOS in a way that \emph{higher-order} languages, such as the $\lambda$-calculus, can be modeled and reasoned about in the abstract categorical framework. We ease the reader into our theory by first considering a combinatory logic, the \emph{SKI
calculus}, originally introduced by~\cite{10.2307/2370619}. It forms a computationally complete fragment of the untyped $\lambda$-calculus that does not feature variables, which for now allows us to bypass the technical intricacies arising from binding and substitution; these are treated later in \Cref{sec:lam}. Specifically, we
work with a variant of the SKI calculus featuring auxiliary
operators, first introduced by~\cite{GianantonioRPO}. We refer to this variant as \emph{extended} combinatory logic, or $\xCL$ for short. It
is as expressive as the standard calculus but allows for a simpler coalgebraic
presentation.

\subsection{Extended combinatory logic}\label{sec:unary-ski}
 The set $\skitermu$ of $\xCL$-terms
is generated by the grammar
\[
  \skitermu \Coloneqq S \mid K \mid I \mid \appp(\skitermu, \skitermu) \mid
  S'(\skitermu) \mid K'(\skitermu) \mid S''(\skitermu,\skitermu).
\]
The binary operation $\appp$ corresponds to function application; we usually write $s\,t$ for $\appp(s,t)$. The standard combinators (constants) $S$, $K$ and $I$ represent the $\lambda$-terms \[S=\lambda x.\,\lambda y.\,\lambda z.\, (x\app z)\app (y\app z),\qquad K=\lambda x.\, \lambda y.\, x, \qquad I=\lambda x.\, x.\] The unary operators
$S'$ and $K'$ capture application of~$S$ and~$K$, respectively, to one
 argument: $S'(t)$ behaves like $S\app t$, and $K'(t)$
behaves like $K\app t$.
Finally, the binary operator $S''$ is meant to
capture application of $S$ to two arguments: $S''(s,t)$ behaves like
$(S\app s)\app t$. In this way, the behaviour of each combinator can
be described in terms of \emph{unary} higher-order functions; for example, the
behaviour of $S$ is that of a function taking a term $t$ to $S'(t)$.

The small-step operational semantics of $\xCL$ is given by the rules displayed in 
\Cref{fig:skirules}, where $p,p',q,t$ range over terms in $\skitermu$. The rules determine a labeled transition system 
\begin{equation}\label{eq:xcl-lts}\to~
\subseteq \skitermu \product (\skitermu + \{\_\}) \product \skitermu
\end{equation} by
induction on the structure of terms in $\skitermu$, with $\{\_\}$ denoting the lack
of a transition label. In this instance the set $\skitermu$ of labels coincides
with the state space of the transition system. Note that every $t\in\skitermu$ either admits a single
unlabeled transition $t\to t'$ or a family of labeled transitions
${(t\xto{s} t_s)_{s\in\skitermu}}$; thus, the transition system is deterministic. The intention is that unlabeled transitions
correspond to \emph{reductions} (i.e.\ computation steps) and that labeled transitions represent
\emph{higher-order behaviour}: a transition $t\xto{s} t_s$ means that the program $t$ acts as a function that outputs $t_s$ on input $s$, where $s$ is itself a program term.
An important observation towards specifying an abstract format in
\Cref{sec:ho} is that labeled transitions are uniformly defined for every
input $s \in \skitermu$, in that operators do not inspect the structure
of $s$.

\begin{remark}
  The incorporation of the auxiliary operators $S'$, $S''$ and $K'$ leads to a
  more consistent semantics, where all combinators accept exactly one argument.
  In addition, this allows us to bypass the issue of a program equivalence
  potentially telling combinators apart based on the number of their input arguments.
  Interestingly, \cite{GianantonioRPO} also introduce the auxiliary operators $S'$,
  $S''$ and $K'$ to avoid precisely this issue (see
  \textsection 5.2 in \emph{op.~cit.}). The added operators do not alter the functional
behaviour of programs compared to the standard $\SKI$ calculus, except for
adding more unlabeled transitions. For example, the conventional rule $S\, t\, s\,
e\to (t\,e)\,(s\,e)$ for the $S$-combinator (see e.g.~\cite{hindley2008lambda})
is rendered as the following sequence of transitions in $\xCL$:
\[
S\,t\,s\,e\to S'(t)\,s\,e\to S''(t,s)\, e\to (t\,e)\,(s\,e).
\]
The first transition exists because (1) $S\xto{t} S'(t)$ by the rule for $S$; (2) therefore $S\, t \to S'(t)$ by the second rule of application; (3) therefore $S\,t\,s\,e\to S'(t)\,s\,e$ by the first rule of application. Similarly for the second and third transition.
\end{remark}

\begin{figure*}[t]
\begin{gather*}
\frac{}{S\xto{t}S'(t)}
\qquad
\frac{}{S'(p)\xto{t}S''(p,t)}
\qquad
\frac{}{S''(p,q)\xto{t}(p\app t)\app (q\app t)}
\\[1ex]
\frac{}{K\xto{t}K'(t)}
\qquad
\frac{}{K'(p)\xto{t}p}
\qquad
\frac{}{I\xto{t}t}
\qquad
\frac{p\to p'}{p \app q\to p' \app q}
\qquad
\frac{p\xto{q} p'}{p \app q\to p'}
\end{gather*}
\caption{Operational semantics of the $\xCL$ calculus.}
\label{fig:skirules}
\end{figure*}

Like every labeled transition system, $\skitermu$ comes with a notion of \emph{(strong) bisimilarity}; see \Cref{ex:lts}. Explicitly, a relation
$R \subseteq \skitermu \product \skitermu$ is a \emph{bisimulation} if, for all $p \mathbin{R} q$,
\begin{enumerate}
\item $p\to p' \implies \exists q'.\, q\to q' \wedge p' \mathbin{R} q'$;
\item $q\to q' \implies \exists p'.\, p\to p' \wedge p' \mathbin{R} q'$;
\item $\forall t\in\skitermu,~p\xto{t} p' \implies \exists q'.\, q\xto{t} q' \wedge p'
  \mathbin{R} q'$;
\item $\forall t\in\skitermu,~q\xto{t} q' \implies \exists p'.\, p\xto{t} p' \wedge p'
  \mathbin{R} q'$.
\end{enumerate}
We write $\sim$ for the greatest bisimulation, viz.\ the union of all
bisimulations (which is itself a bisimulation), and call two programs $p$ and
$q$ \emph{bisimilar} if $p\sim q$. The relation $\sim$ identifies programs that are indistinguishable in terms of computation steps and functional behaviour, in that they
produce bisimilar terms given the same input.

\begin{example}
  \label{ex:ski}
The terms $(S\app K) \app I$ and $(S \app K) \app K$ transition as follows:
\[
\arraycolsep=1.4pt
 \begin{array}{ccccccccccc}
    (S \app K) \app I & \to & S'(K)\app I&\to & S''(K,I)&\xto{t}& (K \app t) \app (I \app
    t)&\to& K'(t) \app (I \app t)&\to& t,\\
    (S \app K) \app K&\to& S'(K)\app K&\to& S''(K,K)&\xto{t}& (K \app t) \app (K \app
    t)&\to& K'(t) \app (K \app t)&\to& t.
  \end{array}
\]
  It follows that $(S\app K)\app I \sim (S\app K) \app K$.
\end{example}

The set $\skitermu$ of $\xCL$-terms forms the initial algebra for
the algebraic signature
 \[ \Sigma=\{\,S/0,K/0,I/0,S'/1,K'/1,S''/2,\appp/2\,\}, \]
with arities as indicated, and bisimilarity is respected by all operations of the language:
\begin{proposition}[Compositionality of $\xCL$]
  \label{prop:skicong1}
  The bisimilarity relation $\sim$ is a congruence.
\end{proposition}
\begin{proof}
  In the following, by a \emph{context} we mean a term
  $C\in \Sigmas\{\cdot\}$ in which the variable~`$\cdot$' (the `hole' of the context) appears at most
  once. We denote a context by $C[\cdot]$, and we write $C[p]=C[p/\cdot]\in \skitermu$ for the closed term
  obtained by substituting $p\in \skitermu$ for the hole in $C[\cdot]$. An
  equivalence relation $R\subseteq\skitermu\times\skitermu$ is a
  congruence if and only if the following relation is contained in $R$: 
  \begin{align*}
    \hat R = \{(C[p],C[q])\in\skitermu\times\skitermu\mid \text{$C[\cdot]$ is
      a context and $p \mathbin{R} q$}\}.
  \end{align*}
  Thus, our task is to prove ${\hat\sim}\subseteq{\sim}$. To this end,
  it suffices to prove that $\hat\sim$ is a bisimulation up to
  transitive closure. This
means that for every context~$C[\cdot]$ and
  $p,q\in\skitermu$ such that $p \sim q$,
\begin{itemize}
\item either there exist $p',q'\in \skitermu$ such that $C[p]\to p'$, $C[q]\to q'$ and $p' \hat\sim^\star q'$,
\item or for every $t\in\skitermu$, there exist $p',q'\in \skitermu$ such that $C[p]\xto{t} p'$, $C[q]\xto{t} q'$
    and~$p' \hat\sim^\star q'$,
\end{itemize}
where $\hat\sim^\star$ denotes the transitive closure of
$\hat\sim$. This then implies that $\hat\sim^\star$ is a bisimulation.~\footnote{In other words, in $\xCL$, \emph{bisimulation up to
      transitive closure} is \emph{sound} for bisimilarity. The proof is simple,
  and we omit details. For more information on {up-to}
  techniques, see e.g. \cite{SANGIORGI_1998,DBLP:conf/sofsem/RotBR13}.}
Consequently we obtain ${\hat\sim}\seq {\hat\sim^\star}\seq {\sim}$ because $\sim$ is
the greatest bisimulation.

We proceed by structural induction on $C[\cdot]$. The cases where~$C[\cdot]$ is
not an application term are straightforward. For instance, for
$C[\cdot] = S''(r,C'[\cdot])$, the property in question can be read from the
diagram
\begin{equation*}
\begin{tikzcd}[column sep=normal, row sep=normal]
S''(r,C'[p])
  \rar[phantom,description,"\hat\sim"]
  \dar["t"'] &
S''(r,C'[q])
  \dar["t"]\\
(r \app t)\app (C'[p] \app t)
  \rar[phantom,description,"\hat\sim^\star"] &
(r \app t)\app (C'[q] \app t)
\end{tikzcd}
\end{equation*}
In fact, the two terms on the bottom are even related by $\hat\sim$. For
application terms, we distinguish cases as follows.
\begin{itemize}
\item If $C[\cdot]=C'[\cdot] \app r$ and the transition of $C[p]$ comes from an
  unlabeled transition ${C'[p]\to p'}$, then the transition of $C[p]$
  is $C'[p]\app r\to p' \app r$, and by induction, we have~$q'$ such that
  $C'[q]\to q'$ and $p' \mathbin{\hat\sim^\star} q'$. This implies $C[q]=C'[q]\app r\to
  q' \app r$ and moreover $p'\app r \mathbin{\hat\sim^\star} q' \app r$; for the latter we use that the relation $\hat\sim^\star$ is a congruence, which follows from the fact that $\hat\sim$ is a congruence by definition and transitive closures of congruences are congruences.
\item If $C[\cdot]=C'[\cdot]\app r$ and the transition of $C[p]$ comes from a labeled
  transition $C'[p]\xto{r} p'$, then the transition of $C[p]$ is
  $C'[p]\app r\to p'$. By induction, we have~$q'$ such that
  $C'[q]\xto{r} q'$ and $p' \mathbin{\hat\sim^\star} q'$, and then $C'[q]\app r\to
  q'$.
\item If $C[\cdot]=r\app C'[\cdot]$ and the transition of $C[p]$ comes from an unlabeled
  transition $r\to r'$, then we have
  $r \app C'[p]\to r'\app C'[p]$ and $r\app C'[q] \to r'\app C'[q]$ which completes the case since $r'\app
  C'[p] \mathbin{\hat\sim} r'\app C'[q]$, so $r'\app
  C'[p] \mathbin{\hat\sim^\star} r'\app C'[q]$.
\item Finally, suppose that $C[\cdot]=r\app C'[\cdot]$ and the transition of $C[p]$
  comes from a labeled transition of~$r$. According to the rules in
  \Cref{fig:skirules}, for every $t\in\skitermu$ we have
  $r\xto{t} r'[t/x]$ for some term $r'$ with one free variable $x$
  such that $r'$ depends only on $r$ but not on $t$. Hence,
  $r\xto{C'[p]} r' [C'[p]/x]$, $r\app C'[p]\to r' [C'[p]/x]$ and
  similarly $r\app C'[q]\to r' [C'[q]/x]$. Since
  $C'[p]\mathbin{\hat\sim} C'[q]$, we conclude that
  $r'[C'[p]/x] \mathbin{\hat\sim^\star} r' [C'[q]/x]$. (Note that the
  variable $x$ can appear multiple times in $r'$, so we generally do
  not have $r'[C'[p]/x] \mathbin{\hat\sim} r' [C'[q]/x]$. This explains the need for the transitive closure
  ${\hat\sim^\star}$.)  \hfill$\qed$
\end{itemize}\def\qed{}
\end{proof}
The proof of \Cref{prop:skicong1} is laborious, as it requires tedious
case distinctions and a carefully chosen up-to technique, although the latter
could have been avoided by working with multi-hole contexts. The present proof is 
also tailored to a specific language, one among many systems exhibiting
higher-order behaviour. In the sequel, we describe an abstract,
categorical representation of such higher-order systems that
guarantees the compositionality of the semantics. In particular, we
shall demonstrate that \Cref{prop:skicong1} emerges as an instance of a
general compositionality result (\Cref{th:main}).

\subsection{A simple higher-order rule format}
\label{sec:ho}
From a coalgebraic perspective, the deterministic labeled
transition system \eqref{eq:xcl-lts} on $\xCL$-terms forms a coalgebra
\begin{equation}
  \label{eq:skiutr}
  \gamma \c \skitermu \to \skitermu +
  \skitermu^{\skitermu}
\end{equation}
for the set functor $Y\mapsto Y+Y^{\skitermu}$, where the two summands of the
codomain represent unlabeled and labeled transitions, respectively. Note that the
coalgebra \eqref{eq:skiutr} can be regarded as an instance of an \emph{applicative
  transition system} (see \cite{Abramsky:lazylambda}) with $\beta$-reduction.
We can abstract away from the set~$\skitermu$ of labels  and
consider $\gamma$ as a system of the form
\begin{equation*}
  Y \to Y + Y^{X}.
\end{equation*}
For higher-order systems such as $\gamma \c \skitermu \to \skitermu +
\skitermu^{\skitermu}$, we expect $X = Y$, underlining the fact that inputs
come from the state space of the system. Note that the assignment
\begin{equation}
  \label{eq:simplehigher}
  B(X,Y) = Y + Y^{X} \c \Set^{\opp} \product \Set \to \Set
\end{equation}
gives rise to a \emph{bifunctor} that is contravariant in $X$ and
covariant in $Y$. On the side of syntax, the signature of
$\xCL$ yields the polynomial endofunctor $\Sigma \c \Set \to \Set$ given by
\begin{equation*}
  \Sigma X = \coprod\nolimits_{\f\in\{S,K,I,S',K',S'',\appp\}} X^{\ar(\f)},
\end{equation*}
As a first step towards our higher-order abstract GSOS framework, we next introduce a simple concrete
rule format for higher-order combinatory calculi that generalizes $\xCL$.

\begin{definition}[$\HO$ rule format]  \label{def:hoformat}
Fix the countably infinite set
 \[\V = \{\,x\,\} \;\cup\; \{\,x_{i}, y_i, y_i^z \;:\; i\in \{1,2,3,\ldots\} \text{ and } z\in \{x, x_1,x_2,x_3,\ldots\}\,\}.\]
 of metavariables and an algebraic signature $\Sigma$. 
\begin{enumerate}
\item An \emph{$\HO$ rule} for an operation symbol
$\f\in \Sigma$ is an expression of the form
\begin{equation}\label{eq:rule-red}
  \inference{(x_j\to y_j)_{j\in
      W}\qquad(\goesv{x_{i}}{y^{z}_{i}}{z})_{i\in
      \ol{W},\,z \in \{x_1,\ldots,x_n\}}}{\f(x_1,\ldots,x_{n})\to
    t}
\end{equation}
or
\begin{equation}\label{eq:rule-nonred}
  \inference{(x_j\to y_j)_{j\in W}\qquad(\goesv{x_{i}}{y^{z}_{i}}{z})_{i\in
      \ol{W},\,z \in \{x,x_1,\ldots,x_n\}}}{\goesv{\f(x_1,\ldots,x_{n})}{t}{x}}
\end{equation}
where $x,x_i,y_i,y_i^{x_j}\in\V$, $n=\ar(\f)$, $W\seq \{1,\ldots, n\}$, $\ol{W}=\{1,\ldots,n\}\smin W$, and $t\in \Sigmas \V$ is a term depending only on the variables occurring in the premise; that is, in the rule \eqref{eq:rule-red} the term $t$ can depend on the variables $x_i$ ($i=1,\ldots,n$), $y_j$ ($j\in W$), and $y_i^{x_j}$ ($i\in \ol{W}$, $j=1,\ldots,n$), and in \eqref{eq:rule-nonred} it can additionally depend on $x$ and $y_i^x$ ($i\in \ol{W}$).
\item  An \emph{$\HO$ specification} for $\Sigma$ is a set of $\HO$
  rules such that for
  each $n$-ary $\f\in \Sigma$ and each $W \seq
  \{1,\ldots,n\}$, there is exactly one rule of the form \eqref{eq:rule-red} or \eqref{eq:rule-nonred} in the set.
\end{enumerate}
\end{definition}

Intuitively, for every given rule the set $W\seq \{1,\ldots.,n\}$ determines which of
the operands of $\f(x_1,\ldots,x_n)$ perform a reduction and which exhibit
higher-order behaviour, i.e.~behave like functions. For
$i\in \ol{W}$, the format dictates that said
functions can be applied to a left-side variable $x_{j}$ or the input
label~$x$, and then the output $x_{i}(x_{j}) = y_{i}^{x_j}$ or
$x_i(x)=y_i^x$ can be used in the conclusion term $t$. The uniformity is apparent:
rules cannot make any assumptions on the input label $x$ or on other
left-side variables that are used as arguments on the premises.

\begin{example}\label{ex:ski-to-ho}
  The rules of $\xCL$ in \Cref{fig:skirules} form an $\HO$ specification modulo suitable renaming of variables and adding dummy premises. For illustration, let us consider the second rule for application. Using the variables $x_1, x_2, y_1^{x_2}$ instead of $p,q,p'$, this rule can be rewritten as
\[ \inference[\texttt{app2}]{x_1\xto{x_2} y_1^{x_2}}{x_1 \app x_2\to y_1^{x_2}}. \]
This is not yet an $\HO$ rule, since the latter require a complete list of premises. However, by filling in the missing premises in every possible way, \texttt{app2} is equivalent to the following two $\HO$ rules. These rules correspond to $y_2$ being a reducing term or to $y_2$ computing a function, respectively, that is, to the choices $W=\emptyset$ and $W=\{2\}$ in \eqref{eq:rule-red}:
  \begin{gather*}
    \inference[\texttt{app2-a}]{\goesv{x_1}{y_1^{x_1}}{x_1} & \goesv{x_1}{y_1^{x_2}}{x_2} &
      \goesv{x_2}{y_2^{x_1}}{x_1} & \goesv{x_2}{y_2^{x_2}}{x_2}} {\goes{x_1 \app x_2}{y_1^{x_2}}} \\
  \inference[\texttt{app2-b}]{\goesv{x_1}{y_1^{x_1}}{x_1} & \goesv{x_1}{y_1^{x_2}}{x_2} &
     \goes{x_2}{y_2}}{\goes{x_1 \app x_2}{y_1^{x_2}}}
  \end{gather*}
  Similarly, the combinator rule
  \[ \frac{}{S'(p)\xto{t}S''(p,t)} \]
  is turned into the two rules
  \[ \frac{x_1\to y_1}{S'(x_1)\xto{x} S''(x_1,x)}\qquad \frac{x_1\xto{x_1} y_1^{x_1}}{S'(x_1)\xto{x} S''(x_1,x)}  \]  
 by using the variables $x_1, x$ instead of $p,t$ and adding the required dummy premises for $x_1$.
\end{example} 
Generalizing the case of $\xCL$, every $\HO$ specification induces a $B(\mS,-)$-coalgebra \begin{equation}\label{eq:can-model-ho}\gamma\colon \mS\to  \mS+\mS^\mS\end{equation} carried by the initial algebra $\mS$ of closed $\Sigma$-terms that runs program terms according to the given $\HO$ rules. Again, its bisimilarity relation is compatible with all language operations:
\begin{proposition}[Compositionality of $\HO$]\label{prop:cong-ho} For every $\HO$ specification, the bisimilarity relation $\sim$ of the induced coalgebra $\gamma\colon \mS\to  B(\mS,\mS)$ is a congruence on $\mS$. 
\end{proposition}
The proof is similar to that of \Cref{prop:skicong1}; we omit it because it is subsumed by our general compositionality result for higher-order abstract GSOS (\Cref{th:main}).

We are now prepared to make the key observation leading to our notion of \emph{higher-order GSOS law} developed in \Cref{sec:hogsos}. Recall from \Cref{sec:abstract-gsos} that GSOS specifications are in bijective correspondence with GSOS laws of a polynomial functor $\Sigma$ over $BX=(\Pow_\omega X)^L$, i.e.\ natural
transformations of the form
\[
  \rho_X\colon \Sigma(X \product (\mypowfin X)^{L}) \to (\mypowfin(\Sigmas X))^{L} \qquad (X\in \Set).
\]
Similarly, $\HO$
specifications correspond to certain (di)natural transformations that distribute a polynomial functor $\Sigma$ over the behaviour bifunctor $B(X,Y)=Y+Y^X$:

\begin{theorem}
  \label{prop:yon1}
  For every algebraic signature $\Sigma$, there is a bijective correspondence between $\HO$ specifications for $\Sigma$ and families of maps
\begin{equation}
  \label{eq:simpleho}
  \rho_{X,Y} \c \Sigma(X \times B(X,Y))\to
  B(X,\Sigmas(X+Y))\qquad (X,Y\in \Set)
\end{equation}
dinatural in $X$ and natural in $Y$.
\end{theorem}

\begin{rem}
Before turning to the proof, let us elaborate on the statement of the theorem and discuss the underlying intuitions.
\begin{enumerate}
\item The (di)naturality conditions on the family $\rho = (\rho_{X,Y})_{X,Y\in\Set}$ assert that the hexagon
\begin{equation}\label{diag:dinat-rho}
\begin{tikzcd}[column sep=5, row sep=20]
  & [-15pt]\Sigma(X\times B(X,Y)) \ar{rr}{\rho_{X,Y}} && [12pt] B(X,\Sigmas (X+Y)) \ar{dr}{B(\id,\Sigmas(f+\id))} &[-10pt] \\
  \Sigma(X\times B(X',Y) \ar{ur}{\Sigma(\id\times B(f,\id))} \ar{dr}[swap]{\Sigma(f\times B(\id,\id))} & && & B(X,\Sigmas(X'+Y)) \\
  & \Sigma(X'\times B(X',Y)) \ar{rr}{\rho_{X',Y}} && [12pt] B(X', \Sigmas(X'+Y)) \ar{ur}[swap]{ B(f,\Sigmas(\id+\id))} & 
\end{tikzcd}
\end{equation}
commutes for all sets $X,X',Y$ and functions $f\colon X\to X'$, and moreover that the rectangle
\begin{equation}\label{diag:nat-rho}
\begin{tikzcd}
\Sigma(X\times B(X,Y)) \ar{r}{\rho_{X,Y}} \ar{d}[swap]{\Sigma(\id\times B(\id,g))} & B(X,\Sigmas(X+Y)) \ar{d}{B(\id,\Sigmas(\id+g))} \\
\Sigma(X\times B(X,Y')) \ar{r}{\rho_{X,Y'}} & B(X,\Sigmas(X+Y'))
\end{tikzcd}
\end{equation}
commutes for all sets $X,Y,Y'$ and functions $g\colon Y\to Y'$.

\item  The need for \emph{di}naturality comes from
  the mixed variance of the behaviour bifunctor $B$, which in turn is caused
  by the fact that variables are used both as states (covariantly) and
  as labels (contravariantly). The role of dinaturality is then the
  same as otherwise played by naturality: It ensures on an abstract
  level that the rules are parametrically polymorphic, that is, they do not
  inspect the structure of their arguments.

In more technical terms, (di)naturality enables the use of the Yoneda lemma to establish the bijective correspondence of \Cref{prop:yon1}. Explicitly, the bijection maps an $\HO$ specification~$\mathcal{R}$ to the family \eqref{eq:simpleho} defined as follows. Given $X,Y\in \Set$ and \[w=\f((u_1,v_1),\ldots, (u_n,v_n))\in \Sigma(X\times B(X,Y)),\]
consider the unique rule in $\mathcal{R}$ matching $\f$ and $W=\{ j\in \{1,\ldots,n\} : v_j\in Y  \}$. If that rule is of the form \eqref{eq:rule-red}, then 
\[\rho_{X,Y}(w)\in \Sigmas(X+Y)\seq B(X,\Sigmas(X+Y)) \]
is the term obtained by taking the term $t$ in the conclusion of \eqref{eq:rule-red} and applying the substitutions
\[ x_i\mapsto u_i~(i\in \{1,\ldots,n\}),\quad y_j\mapsto v_j~(j\in W), \quad y_i^{x_j}\mapsto v_i(u_j)~(i\in \ol{W}, j\in \{1,\ldots,n\} ).  \]
If the rule is of the form \eqref{eq:rule-nonred}, then 
\[\rho_{X,Y}(w)\in \Sigmas(X+Y)^X\seq B(X,\Sigmas(X+Y)) \]
is the map $u\mapsto t_u$, where the term $t_u$ is obtained by taking the term $t$ in the conclusion of \eqref{eq:rule-nonred} and applying the above substitutions along with
\[ x\mapsto u\qquad\text{and}\qquad y_i^{x}\mapsto v_i(u)\quad (i\in \ol{W}). \]
\item Thus, the intuition behind \Cref{prop:yon1} is that $\rho$ encodes a set of $\HO$-rules into a (di-)natural, i.e.\ parametrically polymorphic, family of functions. The use of two sets $X$ and~$Y$ in $\rho_{X,Y}$ reflects that rules may have premises $x\xto{x'} y$ with two types of variables, namely variables $x,x'\in X$ that can appear both as inputs (covariantly) and labels (contravariantly), and variables $y\in Y$ that appear (covariantly) as outputs. The coproduct in $\Sigmas(X+Y)$ ensures that both types of variables can appear in output terms of rules.
\end{enumerate}
\end{rem}

\begin{example}\label{ex:ski-to-rho}
Let $\rho$ be the family  corresponding to the $\HO$ specification of $\xCL$, see \Cref{ex:ski-to-ho}. Given \[w = (u_1,v_1)\app (u_2,v_2) = \appp((u_1,v_1),(u_2,v_2))\in \Sigma(X\times B(X,Y))\] where $v_1\in Y^X$,  one has $\rho_{X,Y}(w)=v_1(u_2)$, according to the rule $\texttt{app2}$.
\end{example}

\begin{proof}[Proof of~\Cref{prop:yon1}]
\begin{enumerate}
\item Removing the syntactic sugar from \Cref{def:hoformat}, we see
  that $\HO$ specifications for a signature~$\Sigma$ correspond
  bijectively to elements of the set
\begin{equation}
  \label{eq:hoformatns}
  \prod\limits_{\substack{\f\in \Sigma\\ W\seq \ar(\f)}}\Bigl(\Sigma^\star \bigl(\ar(\f)+W+\ar(\f)\times \ol{W}\bigr) + \Sigma^\star \bigl(\ar(\f)+1+W+(\ar(\f)+1)\times
    \ol{W}\bigr)\Bigr).
\end{equation}
Here, we identify the natural number $\ar(\f)$ with the set $\{1,\ldots, \ar(\f)\}$. Recall that a choice of $W\seq \{1,\ldots.\ar(\f)\}$ determines which of
the operands of $\f$ perform a reduction (see \Cref{def:hoformat} and the explanations afterwards), and we let $\ol{W}=\ar(\f)\smin W$ denote the complement. Thus the
summands under $\Sigmas$ spell out which variables may be used in the conclusion of the
respective rule. For instance, the rule \texttt{app2-b} of \Cref{ex:ski-to-ho} corresponds to the element
\[ (2,1)\in \ar(\appp)\times \ol{W}\seq \Sigma^\star \bigl(\ar(\appp)+W+\ar(\appp)\times \ol{W}\bigr) \]
where the variable $y_1^{x_2}$ is identified with $(2,1)$, and $\ar(\appp)=2$ is the arity of the  application operator, and $W=\{2\}$.

 We are thus left to prove that elements of the set~\eqref{eq:hoformatns} are in a bijective correspondence with (di)natural transformations of type~\eqref{eq:simpleho}.
\item 
  For functors $F,G\colon \Set^{\opp}\times \Set\times \Set\to \Set$ we let
  $\DiNat_{X,Y}(F(X,X,Y),G(X,X,Y))$ denote the collection of all families of
  maps
  $\rho_{X,Y}\colon F(X,X,Y)\to G(X,X,Y)$
dinatural in $X\in \Set$ and natural in $Y\in \Set$. Then we have the following
chain of bijections:
\begin{align*}
& \DiNat_{X,Y}\bigl(\Sigma(X\times B(X,Y)),B(X,\Sigmas(X+Y))\bigr) \\
=\; & \DiNat_{X,Y}\bigl(\coprod_{\f\in\Sigma} (X\times B(X,Y))^{\ar(\f)},B(X,\Sigmas(X+Y))\bigr) \\
\cong\; & \prod_{\f\in \Sigma}\DiNat_{X,Y}\bigl((X\times B(X,Y))^{\ar(\f)},B(X,\Sigmas(X+Y))\bigr) \\
=\; & \prod_{\f\in \Sigma}\DiNat_{X,Y}\bigl((X\times (Y + Y^{X}))^{\ar(\f)}, B(X,\Sigmas(X+Y))\bigr) \\
\cong\; & \prod_{\f\in \Sigma}\DiNat_{X,Y}\bigl(X^{\ar(\f)}\times (Y + Y^{X})^{\ar(\f)}, B(X,\Sigmas(X+Y))\bigr) \\
\cong\; & \prod_{\f\in \Sigma}\DiNat_{X,Y}\bigl(X^{\ar(\f)}\times \coprod_{W\seq \ar(\f)} Y^W\times (Y^X)^{\ol{W}}, B(X,\Sigmas(X+Y))\bigr) & (*)\\
\cong\; & \prod_{\f\in \Sigma}\DiNat_{X,Y}\bigl(X^{\ar(\f)}\times \coprod_{W\seq \ar(\f)} Y^W\times Y^{X\times\ol{W}}, B(X,\Sigmas(X+Y))\bigr) & (**) \\
\cong\; & \prod_{\f\in \Sigma}\DiNat_{X,Y}\bigl(\coprod_{W\seq \ar(\f)} X^{\ar(\f)}\times Y^W \times Y^{X\times \ol{W}} , B(X,\Sigmas(X+Y))\bigr) & (*) \\
\cong\; & \prod_{\f\in \Sigma} \prod_{W\seq \ar(\f)}\DiNat_{X,Y}\bigl(X^{\ar(\f)}\times Y^{W + X\times \ol{W}} , B(X,\Sigmas(X+Y))\bigr). 
\end{align*}
In the two steps marked ($*$) we use that products distribute over coproducts in $\Set$, and in  the step marked ($**$) we use that $(A^B)^C \cong A^{B\times C}$ for all $A,B,C\in \Set$. We claim that each factor of the last product satisfies
\begin{equation}\label{eq:dinat-sum}
 \begin{aligned}
    \DiNat_{X,Y}\bigl( X^{\ar(\f)}\times &Y^{W + X\times \ol{W}} , B(X,\Sigmas(X+Y))\bigr)\\ 
    \cong\;& 
    \DiNat_{X,Y}\bigl( X^{\ar(\f)}\times Y^{W + X\times \ol{W}} , \Sigmas(X+Y)\bigr) \;+\;\\ 
    \;&\DiNat_{X,Y}\bigl( X^{\ar(\f)}\times Y^{W + X\times \ol{W}} ,(\Sigmas(X+Y))^{X}\bigr).
  \end{aligned}
\end{equation}
To see this, let $\rho$ be a family of maps in $\DiNat_{X,Y}\bigl(X^{\ar(\f)}\times Y^{W + X\times \ol{W}} , B(X,\Sigmas(X+Y))\bigr)$. Consider the diagram below, where $!_X\c X\to 1$ and $!_Y\c Y\to 1$ are the unique maps into the singleton set $1$. The upper part of the diagram commutes by naturality in $Y$ and the lower part by dinaturality in $X$.
\[
\begin{tikzcd}
X^{\mathsf{ar}(\f)} \times Y^{W+ X\times\ol{W}} \ar{r}{\rho_{X,Y}} \ar{d}[swap]{X^{\mathsf{ar}(\f)}\times (!_Y)^{\id+X\times\id}} & \Sigmas(X+Y)+(\Sigmas(X+Y))^X \ar{d}{\Sigmas(\id+!_Y)+ (\Sigmas(\id+!_Y))^X}  \\
 X^{\mathsf{ar}(\f)}\times 1^{W+X\times \ol{W}}  \ar{r}{\rho_{X,1}} & \Sigmas(X+1)+(\Sigmas(X+1))^X  \ar{d}{ \Sigmas(!_X+\id) + (\Sigmas(!_X+\id))^X }  \\
X^{\mathsf{ar}(\f)}\times 1^{W+1\times\ol{W}} \ar{u}{\id\times 1^{\id+!_X\times\id}  }[swap]{\cong} \ar{d}[swap]{ (!_X)^{\mathsf{ar}(\f)}\times \id} & \Sigmas(1+1)+(\Sigmas(1+1))^X \\
 1^{\mathsf{ar}(\f)}\times 1^{W+1\times\ol{W}} \ar{r}{\rho_{1,1}} & \Sigmas(1+1)+(\Sigmas(1+1))^1 \ar{u}[swap]{\id+(\Sigmas(1+1))^{!_X}  }    
\end{tikzcd}
\] 
Note that every map of type $1\to A+B$ into a coproduct (disjoint union) simply chooses an element of one the summands $A$ or $B$. In particular, this applies to the map $\rho_{1,1}$: It has domain $1^{\mathsf{ar}(\f)}\times 1^{W+1\times\ol{W}}\cong 1$, so it chooses an element of $\Sigmas(1+1)$ or $(\Sigmas(1+1))^1$. Note that this element is independent of the given objects $X$ and $Y$ since the component $\rho_{1,1}$ is independent from $X$ and $Y$. It follows that the image of the upper leg
\[  (\Sigmas(!_X+\id) + (\Sigmas(!_X+\id)) \cdot  (\Sigmas(\id+!_Y)+ (\Sigmas(\id+!_Y))^X) \cdot \rho_{X,Y}  \]
of the above commutative diagram
is either a single element of $\Sigmas(1+1)$ or a single element of $(\Sigmas(X+1))^X$. This implies that the image of $\rho_{X,Y}$ must be contained in one of the summands of its codomain: It is either a subset of $\Sigmas(X+Y)$ for every $X,Y$, or a subset of $(\Sigmas(X+Y))^X$ for every $X,Y$.
This proves the isomorphism \eqref{eq:dinat-sum}.
\item It remains to show that the two summands in \eqref{eq:dinat-sum} are isomorphic to the
corresponding summands in~\eqref{eq:hoformatns}. For the first one, let $\NT_X(F(X),G(X))$ denote the collection of natural transformations between functors $F,G\c \Set\to \Set$. Then 
\begin{align*}
  & \DiNat_{X,Y}\bigl(X^{\ar(\f)}\times Y^{W + X\times \ol{W}} , \Sigmas(X+Y)\bigr)\\
  \cong\;& \NT_X\bigl(X^{\ar(\f)}, \NT_Y(Y^{W + X\times \ol{W}}, \Sigmas(X+Y))\bigr) \\
  \cong\; &  \NT_Y\bigl(Y^{W + \ar(\f)\times \ol{W}}, \Sigmas(\ar(\f)+Y)\bigr) \\
  \cong\; &  \Sigmas(\ar(\f)+W+\ar(\f)\times \ol{W}).
\end{align*}
The last two isomorphisms use the Yoneda lemma, and the first one is given by currying:
\[ \rho \quad\mapsto \quad (\,\lambda x\in X^{\ar(\f)}. \,(\rho_{X,Y}(x,-))_Y\,)_X.\]
Note that for every $x\in X^{\ar(\f)}$ the family $(\rho_{X,Y}(x,-)\c Y^{W+X\times \ol{W}}\to \Sigmas(X+Y))_Y$ is indeed natural in $Y$; the naturality squares are equivalent to the ones witnessing naturality of $\rho_{X,Y}$ in $Y$. Similarly, the family $(\,\lambda x\in X^{\ar(\f)}. \,(\rho_{X,Y}(x,-))_Y\,)_X$ is natural in $X$; the naturality squares are equivalent to the commutative hexagons witnessing dinaturality of $\rho_{X,Y}$ in $X$. 

Much analogously, we have
\begin{align*}
  & \DiNat_{X,Y}\bigl(X^{\ar(\f)}\times Y^{W + X\times \ol{W}} , (\Sigmas(X+Y))^X\bigr) \\
  \cong\;& \NT_X\bigl(X^{\ar(\f)}\times X, \NT_Y(Y^{W + X\times \ol{W}}, \Sigmas(X+Y))\bigr) \\
  \cong\;& \NT_X\bigl(X^{\ar(\f)+1}, \NT_Y(Y^{W + X\times \ol{W}}, \Sigmas(X+Y))\bigr) \\
  \cong\; &  \NT_Y\bigl(Y^{W + (\ar(\f)+1)\times \ol{W}}, \Sigmas(\ar(\f)+1+Y)\bigr) \\
  \cong\; &  \Sigmas(\ar(\f)+1+W+(\ar(\f)+1)\times \ol{W}).
\end{align*}
This concludes the proof.\hfill$\qed$
\end{enumerate}
\def\qed{}
\end{proof}

\subsection{Nondeterministic $\xCL$}\label{sec:nd-ski}
Just as the $\lambda$-calculus, combinatory logic can be enriched with other features,
such as nondeterminism, and the theory of applicative bisimulations can be readily
developed for such extensions. For the $\lambda$-calculus this has been pioneered by
\citet{DBLP:journals/iandc/Sangiorgi94}. For example, consider an extension
of $\xCL$ with a binary operator $\oplus$ representing nondeterministic choice.
The grammar of the extended language $\xCL^{\oplus}$ is given by
\[
  \skitermu^\oplus
  \Coloneqq
  S \mid K \mid I \mid
  \appp(\skitermu^\oplus, \skitermu^\oplus)
  \mid
  S'(\skitermu^\oplus)
  \mid
  K'(\skitermu^\oplus)
  \mid
  S''(\skitermu^\oplus,\skitermu^\oplus)
  \mid
  \skitermu^\oplus \oplus \skitermu^\oplus.
\]
On the side of the operational semantics, $\xCL^{\oplus}$ has the 
same rules as $\xCL$ (see \Cref{fig:skirules}), plus the following two rules for resolving nondeterminism:
\begin{align*}
\inference{}{p\oplus q\to p}
&&
\inference{}{p\oplus q\to q}
\end{align*}
This semantics calls for the modification of the behaviour bifunctor $B(X,Y)=Y+Y^X$ to
\begin{equation}
  \label{eq:higher-nd}
  B^{\oplus}(X,Y) = \mypowfin(Y+Y^{X}) \c \Set^{\opp} \product \Set \to \Set,
\end{equation}
where $\mypowfin$ is the finite powerset functor. Sets of nondeterministic
transition rules such as those of $\xCL^{\oplus}$ then correspond to families of functions
\begin{equation*}
  \rho_{X,Y} \c \Sigma(X \times B^\oplus(X,Y))\to B^\oplus(X,\Sigmas(X+Y))\qquad (X,Y\in \Set)
\end{equation*}
dinatural in $X$ and natural in $Y$.
In analogy to \Cref{prop:skicong1}, we have the following compositionality result:
\begin{proposition}\label{prop:skicong2}
  Bisimilarity for $\xCL^{\oplus}$ is a congruence.
\end{proposition}
Rather than giving another proof by induction on the syntax, we will derive
this proposition from our abstract congruence result (\Cref{th:main}). This
highlights the advantage of the genericity achieved by working in a
category-theoretic framework.

\section{Higher-order abstract GSOS}
\label{sec:hogsos}
We present the main contribution of our paper, a theory of abstract GSOS for
higher-order languages that retains the key feature of Turi and Plotkin's first-order framework, namely that (under mild conditions) compositionality of specifications comes for free.

\subsection{Higher-order GSOS laws}\label{sec:ho-gsos-laws}
The results of the previous section, most notably \Cref{prop:yon1}, suggest a path towards modeling higher-order languages in an abstract, purely categorical fashion: Present their small-step operational semantics in terms of families
of morphisms
\begin{align}\label{eq:law-specialcase}
\rho_{X,Y}\c\Sigma(X\times B(X,Y))\to B(X,\Sigma^\star (X+Y)),
\end{align}
dinatural in $X\in \C$ and natural in $Y\in \C$, parametric
in a base category $\gcat$ and two functors $\Sigma \c \gcat
\to \gcat$ and $B\colon \C^\opp\times \C\to \C$ representing the syntax and
behaviour of the language. The initial $\Sigma$-algebra $\mu\Sigma$
should thus correspond to the terms of the language and their dynamics should be modelled
after a $B(\mu\Sigma,-)$-coalgebra structure
$\gamma \c \mu\Sigma \to B(\mu\Sigma,\mu\Sigma)$.

With the developments in \Cref{sec:lam} in mind, we will
actually work with a slightly more general format where $X$ is
required to be a \emph{pointed object}:
\begin{notation}
Given a fixed
object $\Pt$ of a category~$\gcat$, we let~$\Pt/\gcat$ denote the coslice
category of \emph{$\Pt$-pointed objects}. Its objects are pairs $(X,p_X)$ of an object $X\in\gcat$ and a morphism
$p_X\colon\Pt\to X$ of $\gcat$. A \emph{morphism} from~$(X,p_X)$ to $(Y,p_Y)$
is a morphism $h\colon X\to Y$ of $\C$ such that $h\comp p_X =
p_Y$.
We write \[j\c \Pt/\gcat \to \gcat\] for the forgetful functor given by
$(X, p_X) \mapsto X$ and $h\mapsto h$.
\end{notation}
We think of $\Pt$ as a set of variables, and of a $V$-pointed object $(X,p_X)$ as a
set $X$ of program terms in free variables from $\Pt$ with an embedding
$p_X\c\Pt\to X$ of the variables.
\begin{assumptions}
  \label{mainassumptions}
From now on, we fix the following data:
\begin{enumerate}
\item a category $\C$ with finite limits and colimits;
\item an object $V\in \C$ of variables;
\item two functors $\Sigma\colon \C\to\C$ and $B\colon \C^\opp\times \C\to \C$,
  where $\Sigma$ is of the form $\Sigma = \Pt + \Sigma'$ for some $\Sigma'\colon
  \C\to \C$ and has free algebras on every object, hence generates a free monad
  $\Sigmas$. In addition, we require the functor $B(\mu\Sigma, -)$ to admit a
  final coalgebra, where $\mu\Sigma = \Sigmas 0$.
\end{enumerate}
\end{assumptions}
\begin{definition}
  \label{def:hog}
  A \emph{$\Pt$-pointed higher-order GSOS law} of $\Sigma$ over
  $B$ is a family of
  morphisms
\begin{align}\label{eq:law}
\rho_{X,Y} \c \Sigma (jX \times B(jX,Y))\to B(jX, \Sigma^\star (jX+Y))
\end{align}
dinatural in $X\in\Pt/\gcat$ and natural in $Y\in \gcat$.
\end{definition}

\begin{remark}
More explicitly, dinaturality in $X\in\Pt/\gcat$ means that the hexagon \eqref{diag:dinat-rho} commutes for all objects $(X,p_X),(X',p_{X'})\in \Pt/\gcat$ and $Y\in \C$ and all morphisms $f\colon (X,p_X)\to (X',p_{X'}) $ in $\Pt/\gcat$. Similarly, naturality in $Y\in \C$ means that the rectangle \eqref{diag:nat-rho} commutes for all objects $(X,p_X)\in \Pt/\gcat$ and $Y,Y'\in \C$ and all morphisms $g\colon Y\to Y'$ in $\C$.
\end{remark}

Laws of the form \eqref{eq:law-specialcase} emerge from~\eqref{eq:law} 
by choosing $\Pt=0$, the initial object of $\gcat$. When running the semantics, 
both $X$ and $Y$ will be instantiated to the free algebra~$\mS$. Abstracting
from this choice ensures that program terms are used in a parametrically
polymorphic, uniform way.

Every object $X \in \gcat$ induces an endofunctor
$B(X,\argument) \c \gcat \to \gcat$. For instance, the transition
system $\gamma\c\skitermu \to \skitermu + \skitermu^{\skitermu}$
from \eqref{eq:skiutr} is a $B(\mS, \argument)$-coalgebra. The state
space~$\skitermu$ is the initial $\Sigma$-algebra for the
corresponding polynomial set functor $\Sigma$; the codomain is
$B(\mS,\mS)$. The definition of the map~$\gamma$ is inductive on
the structure of terms. It turns out to be an instance of a definition
by structural induction  in which we assign to a
$\Pt$-pointed higher-order GSOS law its canonical operational~model.
For technical reasons, we formulate the abstract definition of
$\gamma$ yet more generally, by parametrizing it with a
$\Sigma$-algebra $(A,a)$ --- the motivating instance is
obtained by instantiating $A$ with the initial algebra $\mS$.

\begin{rem}\label{rem:alg-pointed}
  For every $\Sigma $-algebra $(A,a)$, we regard $A$ as a $\Pt$-pointed
  object with point
  \[p_A = \bigl(\Pt\xra{\inl} \Pt+\Sigma' A = \Sigma  A \xra{a} A\bigr).\] Note that
  if $h\colon (A,a)\to (B,b)$ is a morphism of $\Sigma $-algebras,
  then $h$ is also a morphism of the corresponding $\Pt$-pointed
  objects.
\end{rem}

\begin{lemma}\label{lem:clubs}
Given a $\Pt$-pointed higher-order GSOS law $\rho$, every $\Sigma $-algebra $(A,a)$ induces a
{unique} morphism $a^\clubsuit\c\mS\to B(A,A)$ in $\C$ such that the following
diagram commutes.
\begin{equation}\label{diag:can-model}
\begin{tikzcd}[column sep=8ex, row sep=normal]
\Sigma(\mS) 
  \dar["\Sigma \brks{\iter a,\,a^\clubsuit}"']
  \ar[rrr,"\iota"] 
  &[-4ex] &[1ex] &[-2ex]
\mS 
  \dar["a^\clubsuit"]
  \\
\Sigma (A\times {B(A,A)})
  \rar["{\rho_{A,A}}"] 
&
B(A,\Sigma^\star(A+A))
  \rar["{B(\id,\Sigmas\nabla)}"] &
B(A,\Sigma^\star A)
  \rar["{B(\id, \hat a)}"] &
B(A,A)
\end{tikzcd}
\end{equation}
Here $(\mS,\ini)$ is the initial $\Sigma $-algebra, $\hat a\colon \Sigmas  A \to
A$ is the Eilenberg--Moore algebra corresponding to the $\Sigma$-algebra $(A,a)$, and we regard $A$ as $V$-pointed as in \Cref{rem:alg-pointed}.
\end{lemma}

\begin{proof}
\emph{Existence of $a^\clubsuit$.} By initiality of $\mS$, there exists a unique morphism $\langle w,a^\clubsuit\rangle$ in $\C$ making the diagram below commute:
\begin{equation}\label{diag:w-a-clubs}
\begin{tikzcd}[column sep=10ex, row sep=normal]
\Sigma(\mS) 
  \dar["\Sigma\brks{w,\,a^\clubsuit}"']
  \ar[rr,"\iota"] 
  & &[1ex]
\mS 
  \dar["\brks{w,\,a^\clubsuit}"]
  \\
\Sigma (A\times {B(A,A)})
  \rar["\brks{a\cdot\Sigma\fst,\,\rho_{A,A}}"] 
&
A\times B(A,\Sigma^\star(A+A))
  \rar["{\id\times B(\id,\hat{a}\comp \Sigmas\nabla)}"] &
A\times B(A,A)
\end{tikzcd}
\end{equation}
Postcomposing this diagram with $\fst\colon A\times B(A,A)\to A$ shows that $w\c \mS\to (A,a)$ is a $\Sigma$-algebra morphism; hence $w=\iter a$ and the diagram \eqref{diag:w-a-clubs} can be rewritten as
\begin{equation}\label{diag:w-a-clubs-2}
\begin{tikzcd}[column sep=10ex, row sep=normal]
\Sigma(\mS) 
  \dar["\Sigma\brks{\iter a,\,a^\clubsuit}"']
  \ar[rr,"\iota"] 
  & &
\mS 
  \dar["\brks{\iter a,\,a^\clubsuit}"]
  \\
\Sigma (A\times {B(A,A)})
  \rar["\brks{a\cdot\Sigma\fst,\,\rho_{A,A}}"] 
&
A\times B(A,\Sigma^\star(A+A))
  \rar["{\id\times B(\id,\hat{a}\comp \Sigmas\nabla)}"] &
A\times B(A,A)
\end{tikzcd}
\end{equation}
By postcomposing with $\snd\colon A\times B(A,A)\to B(A,A)$ we see that \eqref{diag:can-model} commutes.\\[.01ex] 

\noindent\emph{Uniqueness of $a^\clubsuit$.} Suppose that $a^\clubsuit$ is a morphism such that \eqref{diag:can-model} commutes. Then \eqref{diag:w-a-clubs} 
commutes with $w=\iter a$, so uniqueness of $a^\clubsuit$ is entailed by 
the uniqueness of $\brks{w,a^\clubsuit}$.
\end{proof}

\begin{definition}\label{def:operational-model}
The \emph{operational model} of a $V$-pointed higher-order GSOS law $\rho$ is given by the $B(\mS,-)$-coalgebra
\[\iota^\clubsuit\c\mu \Sigma \to B(\mS ,\mS ).\]
\end{definition}

\begin{example}\label{ex:a-clubsuit-ski}
\begin{enumerate}
\item Consider the higher-order GSOS law $\rho$ corresponding to $\xCL$ (see \Cref{ex:ski-to-ho,ex:ski-to-rho}). Then the operational model $\iota^\clubsuit$ is precisely that transition system $\gamma\colon\skitermu\to \skitermu + \skitermu^{\skitermu}$ of \eqref{eq:skiutr} induced by the rules in \Cref{fig:skirules}. Given a $\Sigma$-algebra $(A,a)$, the morphism $a^\clubsuit$ is obtained by interpreting all those transitions in the algebra $A$. For instance, since there is a transition $K\xto{t}K'(t)$, we have $a^\clubsuit(K)\in A^A$ given by $a^\clubsuit(K)(u)=(K')^A(u)$, where $(K')^A\c A\to A$ is the interpretation of the unary operation symbol $K'\in \Sigma$ in $A$.
\item More generally, the operational model $\gamma\colon \mS\to\mS+\mS^\mS$ of an $\HO$ specification, which runs programs according to the given inductive $\HO$ rules, coincides with the operational model $\ini^\clubsuit$ of its corresponding higher-order GSOS law (\Cref{prop:yon1}).
\end{enumerate}
\end{example}

\begin{rem}[First-order vs.\ higher-order abstract GSOS]\label{rem:gsos-to-ho-gsos}
Higher-order abstract GSOS is a conservative extension of first-order abstract GSOS (\Cref{sec:abstract-gsos}). In more detail, given endofunctors $\Sigma, B_0\colon \C\to \C$, every first-order GSOS law \[\rho^0_Y\c \Sigma(Y\times B_0Y)\to B_0\Sigmas Y
  \qquad(Y\in \C)\] of $\Sigma$ over $B_0$ can be turned into a $0$-pointed
  higher-order GSOS law
  \[\rho_{X,Y}\c \Sigma(X\times B(X,Y))\to B(X,\Sigmas(X+Y))\qquad
  (X,Y\in \C)\]
of $\Sigma$ over $B\colon \C^\opp \times \C\to \C$, where $B(X,Y)=B_0Y$, whose components are given by
  \[
   \rho_{X,Y} \;=\; (\begin{tikzcd}[column sep=2.5em]\Sigma(X\times B_0Y)  \ar{rr}{\Sigma(\inl\times B_0\inr)} && \Sigma((X+Y)\times B_0(X+Y))  \ar{r}{\rho^0_{X+Y}} &  B_0\Sigmas(X+Y)\end{tikzcd}).
\]
Let us check that $\rho$ is indeed a higher-order GSOS law. Naturality of $\rho_{X,-}$ is shown by the following
commutative diagram for $f\c Y\to Y'$. The left-hand part obviously
commutes, and the right-hand one commutes by naturality of $\rho^0$.
\[
\begin{tikzcd}[column sep=.5em, row sep=3em, scale cd=.85]
\Sigma(X\times B(X,Y)) = \Sigma(X\times B_0Y) \ar{d}[swap]{\Sigma(\id\times B_0f)} \ar{r}[yshift=0.5em]{\Sigma(\inl\times B_0\inr)} & \Sigma((X+Y)\times B_0(X+Y)) \ar{d}{\Sigma((\id+f)\times B_0(\id+f))} \ar{r}[yshift=.5em]{\rho^0_{X+Y}} &  B_0\Sigmas(X+Y)=B(X,\Sigmas(X+Y)) \ar{d}{B_0\Sigmas(\id+f)} \\
\Sigma(X\times B(X,Y')) = \Sigma(X\times B_0Y') \ar{r}[yshift=0.5em]{\Sigma(\inl\times B_0\inr)} & \Sigma((X+Y')\times B_0(X+Y')) \ar{r}[yshift=.5em]{\rho^0_{X+Y'}} &  B_0\Sigmas(X+Y')=B(X,\Sigmas(X+Y'))
\end{tikzcd}
\]
Since $B(X,Y)=B_0Y$ does not depend on its contravariant component, dinaturality
 of $\rho_{-,Y}$ is equivalent to naturality and is shown by the commutative diagram below for  $g\c X\to X'$:
\[
\begin{tikzcd}[column sep=.5em, row sep=3em, scale cd=.85]
\Sigma(X\times B(X,Y)) = \Sigma(X\times B_0Y) \ar{d}[swap]{\Sigma(g\times \id)} \ar{r}[yshift=0.5em]{\Sigma(\inl\times B_0\inr)} & \Sigma((X+Y)\times B_0(X+Y)) \ar{d}{\Sigma((g+\id)\times B_0(g+\id))} \ar{r}[yshift=.5em]{\rho^0_{X+Y}} &  B_0\Sigmas(X+Y)=B(X,\Sigmas(X+Y)) \ar{d}{B_0\Sigmas(g+\id)} \\
\Sigma(X'\times B(X',Y)) = \Sigma(X'\times B_0Y) \ar{r}[yshift=0.5em]{\Sigma(\inl\times B_0\inr)} & \Sigma((X'+Y)\times B_0(X'+Y)) \ar{r}[yshift=.5em]{\rho^0_{X'+Y}} &  B_0\Sigmas(X'+Y)=B(X',\Sigmas(X'+Y))
\end{tikzcd}
\]
The two laws $\rho^0$ and $\rho$ are semantically equivalent in the sense that their operational models
\[ \gamma\c \mS\to B_0(\mS)\qquad\text{and}\qquad \iota^\clubsuit\c \mS\to B_0(\mS)=B(\mS,\mS)  \]
coincide. To see this, recall from \Cref{sec:abstract-gsos} that the coalgebra structure $\gamma$ is uniquely determined by the following commutative diagram:
\begin{equation*}
\begin{tikzcd}[column sep=8ex, row sep=normal]
\Sigma(\mS) 
  \dar["\Sigma \brks{\id,\,\gamma}"']
  \ar[rrr,"\iota"] 
  &[-2ex] &[1ex] &[2ex]
\mS 
  \dar["\gamma"]
  \\
\Sigma (\mS\times B_0(\mS))
  \ar{rr}{\rho^0_{\mS}} 
&
&
B_0\Sigmas(\mS) \ar{r}{B_0\hat\ini}
&
B_0(\mS)
\end{tikzcd}
\end{equation*} 
Thus, we only need to show that $\iota^\clubsuit$ is such a $\gamma$, which follows from the commutative diagram below. The upper cell commutes by definition of $\iota^\clubsuit$, the cell involving $\rho^0_{\mS}$ commutes by naturality of $\rho^0$, and the remaining cells commute either trivially or by definition.
\[ 
\begin{tikzcd}[scale cd=.65, row sep=3em, column sep=1em]
\Sigma(\mS) \ar{rrrrr}{\iota} \ar{ddd}[swap]{\Sigma\langle \id,\,\iota^\clubsuit\rangle} \ar{dr}{\langle \iter \iota,\, \iota^\clubsuit\rangle} & & & & & \mS \ar{ddd}{\iota^\clubsuit} \ar{dl}[swap]{\iota^\clubsuit} \\
& \Sigma(\mS\times B(\mS,\mS)) \ar{d}{\Sigma(\inl\times B_0\inr)} \ar{r}{\rho_{\mS,\mS}} \ar[equals, bend right=2em]{ddl} & B(\mS,\Sigmas(\mS+\mS)) \ar{r}[yshift=.5em]{B(\id,\Sigmas\nabla)} \ar{dd}{B_0\Sigmas\nabla} & B(\mS,\Sigmas(\mS)) \ar{r}[yshift=.5em]{B(\id,\hat\ini)} & B(\mS,\mS) \ar[equals]{ddr} & \\
& \Sigma((\mS+\mS)\times B_0(\mS+\mS)) \ar{ur}[swap]{\rho^0_{\mS+\mS}} \ar{dl}{\Sigma(\nabla\times B_0\nabla)} & & & & \\
\Sigma(\mS\times B_0(\mS)) \ar{rr}{\rho^0_{\mS}} & & B_0\Sigmas(\mS) \ar{rrr}{B_0\hat\ini} & & & B_0(\mS) 
\end{tikzcd}
\]
\end{rem}

\begin{rem}
The construction of the operational model (\Cref{def:operational-model}) only uses the component $\rho_{\mS,\mS}$ of the given higher-order GSOS law, followed by a codiagonal that simply forgets which of the two copies of $\mS$ the individual variables of the output term come from. It would thus be tempting to generalize higher-order GSOS laws \eqref{eq:law} to dinatural transformations of type 
\begin{equation}\label{eq:ho-gsos-law-single-object}
\rho'_X\colon \Sigma(X\times B(X,X))\to B(X,\Sigmas X).
\end{equation}
Note that they are indeed more general: every higher-order GSOS law \eqref{eq:law} induces a dinatural transformation \eqref{eq:ho-gsos-law-single-object} by putting 
\[ \rho'_X = ( \begin{tikzcd}
\Sigma(X\times B(X,X)) \ar{r}{\rho_{X,X}} & B(X,\Sigmas(X+X)) \ar{r}{B(\id,\nabla)} & B(X,\Sigmas X)
\end{tikzcd} ). \] 
Unfortunately, this format appears to be too permissive to guarantee congruence. Specifically, the proof of \Cref{lem:law_comm} below, which is the key to our congruence result, crucially rests on the `two-variable' form $\rho_{X,Y}$ of higher-order GSOS laws and does not carry over to the generalized format. 
\end{rem}

\subsection{Compositionality}\label{sec:ho-gsos-compositional}
We now investigate when a higher-order GSOS law gives rise to a compositional semantics.
Recall from \Cref{sec:abstract-gsos} that in first-order abstract GSOS, compositionality comes for free and is an immediate consequence of $\mS$ extending to an initial bialgebra and $\nu B$ extending to a final bialgebra for a given GSOS law. We shall see in \Cref{sec:bialg} that the latter does not carry over to the higher-order setting. Therefore, we take a different route to compositionality, working in a framework of regular categories (\Cref{sec:categories}).

Assuming that the final $B(\mS ,\argument)$-coalgebra
\[
 \zeta \c
  \finalc
  \to
  B(\mS , \finalc)
\]
exists, we think of the unique coalgebra morphism
$\coiter(\iota^\clubsuit) \c \mS  \to Z$ from the operational model $\iota^\clubsuit\colon \mS\to B(\mS,\mS)$ (\Cref{def:operational-model}) to $(Z,\zeta)$ as the map assigning to each
program in $\mS $ its abstract behaviour. The ensuing notion
of \emph{behavioural equivalence} is then expressed categorically by
the kernel pair of $\coiter(\iota^\clubsuit)$, i.e.\ the pullback
\begin{equation}
  \label{eq:kernel}
\begin{tikzcd}[column sep=3em, row sep=normal]
  E
  \rar["p_1"]
  \dar["p_2"']
  \ar[dr, phantom , very near start, color=black]
  \pullbackangle{-45}
  &
\mS 
  \dar["\coiter \iota^\clubsuit"]\\
\mS  \rar["\coiter \iota^\clubsuit"] & \finalc
\end{tikzcd}
\end{equation}

\begin{definition}
The kernel pair $E$ is a \emph{congruence} if it forms a subalgebra of the product algebra $\mS\times \mS$; that is, $E$ can be equipped with a (necessarily unique) $\Sigma$-algebra structure $(E,e)$ such that $p_1,p_2\colon (E,e)\to (\mS,\ini)$ are $\Sigma$-algebra homomorphisms.
\end{definition}

\begin{rem}
  For $\C=\Set$ and $\Sigma$ a polynomial functor, the kernel $E$ is
  the equivalence relation on~$\mS $ defined by
  \[
    E = \{\, (s,t)\in \mS \times \mS \c (\coiter\iota^\clubsuit)(s)=(\coiter\iota^\clubsuit)(t)\, \}
  \]
  with the two projection maps $p_1(s,t)=s$ and $p_2(s,t)=t$, and $E$
  forms a congruence in the above categorical sense if and only if it
  forms a congruence in the usual algebraic sense as recalled in \Cref{sec:categories}, i.e.~it is compatible with the $\Sigma$-algebra
  structure of $\mS$.
\end{rem}
\noindent Our main compositionality result asserts that for a regular base category~$\C$ and under mild conditions on the functors $\Sigma$ and $B$, behavioural equivalence is always a congruence.
\begin{rem}\label{rem:reflexive}
  Recall that a parallel pair $f,g\c X \parto Y$ is \emph{reflexive}
  if there exists a common splitting, viz.~a morphism $s\c Y \to X$ such that
  $f\comp s = \id_Y = g \comp s$. A \emph{reflexive coequalizer} is a
  coequalizer of a reflexive pair. Preservation of reflexive
  coequalizers is a relatively mild condition for set functors. For instance, every polynomial set functor~$\Sigma$ and, more generally,
  every finitary set functor preserves reflexive coequalizers~\cite[Cor.~6.30]{AdamekEA11}.
\end{rem}
\begin{theorem}[Compositionality]
  \label{th:main}
  Suppose that \Cref{mainassumptions} hold, and
  let $\rho$ be a $\Pt$-pointed higher-order GSOS law of
  $\Sigma \c \gcat \to \gcat$ over
  $B\c \gcat^{\opp} \product \gcat \to \gcat$. Suppose that the
  category $\C$ is regular, that~$\Sigma $ preserves reflexive
  coequalizers, and that $B$ preserves monomorphisms. Then the kernel
  pair of $\coiter(\iota^\clubsuit) \c \mS \to \finalc$ is a
  congruence.
\end{theorem}

\noindent (Note that a morphism $(f,g)$ in
$\gcat^{\opp} \product \gcat$ is monic iff $g$ is monic in~$\gcat$ and
$f$ is epic in $\gcat$.) 

\begin{remark}\label{rem:proof-strategy}
We will show that the coequalizer $q\c \mS \epito Q$ of $p_1,p_2$ in $\C$ can be equipped with a
$\Sigma $-algebra structure $a\colon \Sigma  Q \to Q$ such that $q$
is a $\Sigma $-algebra morphism from~$\mS $ to $(Q,a)$, that is,
$q=\iter(a)$. This immediately implies the theorem: since $p_1,p_2$ is the kernel pair of $q$ in $\C$ and the forgetful functor from $\Alg(\Sigma)$ to $\C$ creates limits, in particular kernel pairs, there exists a unique $\Sigma$-algebra structure $e$ on $E$ such that $p_1,p_2\colon (E,e)\to (\mS,\ini)$ are $\Sigma$-algebra homomorphisms (and this makes $p_1,p_2$ the kernel pair of $q$ in $\Alg(\Sigma)$).
\end{remark}

We make use of the following two lemmas,
which are of independent interest, en route to proving
\Cref{th:main}. First, we establish a crucial connection between
$\iota^\clubsuit$ and $a^\clubsuit$ (for general $a\c\Sigma A\to A$),
showing that they can be unified by running the unique morphism
$\iter(a)\c\mS\to A$ at the covariant and the contravariant positions
of $B$ correspondingly. This critically relies on (di)naturality
of the given law $\rho$.
\begin{lemma}\label{lem:law_comm}
Let $(A,a)$ be a $\Sigma $-algebra. Then the following diagram commutes:
\begin{equation*}
\begin{tikzcd}[column sep=10em, row sep=normal]
\mS 
  \rar["\iota^\clubsuit"]
  \dar["a^\clubsuit"'] &
B(\mS ,\mS )
  \dar["{B(\id, \iter a)}"] \\
B(A,A)
  \rar["{B(\iter a,\id)}"] &
{B(\mS , A)}
\end{tikzcd}
\end{equation*}
\end{lemma}

\begin{example}[$\xCL$, cf.\ \Cref{ex:a-clubsuit-ski}]
The two legs of the diagram send the term $K\in \mS$ to the function $f\in A^\mS$ given by $f(t)=(K')^A((\iter(a))(t))$, equivalently $f(t)=(\iter(a))(K'(t))$ since $\iter(a)$ is a $\Sigma$-algebra morphism.
\end{example}

\begin{proof}[Proof of~\Cref{lem:law_comm}]
We strengthen the claim of the lemma a bit and show that the outside of the following diagram commutes:
\begin{equation*}
\begin{tikzcd}[column sep=-1em, row sep=3ex]
	\mS  & &[35ex] &[-4ex] \mS \times B(\mS,\mS )\\[3ex]
	\mS \times B(A,A) &   & & \mS \times B(\mS, A)  
	\arrow[from=2-1, to=2-4, "{\id \times B(\iter a, \id)}"]
	\arrow[from=1-1, to=2-1, "\brks{\id,\,a^\clubsuit}"']
	\arrow[bend left=6pt, from=1-1, to=2-4, "\iter b" {yshift=-5pt, xshift=22pt}]
    \arrow[from=1-4, to=2-4, "{\id\times B(\id ,\iter a)}"]
	\arrow[from=1-1, to=1-4, "\brks{\id,\,\iota^\clubsuit}"]
\end{tikzcd}
\end{equation*}
From this, the goal easily follows by applying the right projection $\snd$.
To prove commutativity of the diagram, we consider the $\Sigma$-algebra structure $b$ on $\mS\times B(\mS,A)$ given by the following composite:
\begin{equation*}
\begin{tikzcd}[column sep = 3.2em, row sep=4em]
  \Sigma (\mS\times B(\mS, A))
  \ar[r,"{\brks{\iota\comp\Sigma \fst,\,\rho_{\mS ,A}}}"]
& 
\mS\times B(\mS, \Sigmas (\mS+A))
  \ar[dl,"{\id\times B(\id, \Sigmas(\iter a + \id))}"']
\\
\mS \times B(\mS ,\Sigmas (A+A))  
  \ar[r,"{\id\times B(\id,\Sigmas\nabla)}"]
&
\mS\times B(\mS ,\Sigmas  A)
  \ar[r,"{\id\times B(\id,\hat a)}"]
&
\mS \times B(\mS ,A).
\end{tikzcd}
\end{equation*}
We will show that the composite morphisms
\begin{align}
\mS \xto{\brks{\id ,\,\iota^\clubsuit}}   &\;\mS \times B(\mS,\mS)  \xto{\id\times B(\id, \iter a)} \mS \times B(\mS ,A)\label{eq:law_comm2} \\
\mS \xto{\brks{\id ,\,a^\clubsuit}}       &\;\mS \times B(A,A)      \xto{\id\times B(\iter a, \id)} \mS \times B(\mS ,A),\label{eq:law_comm1}
\end{align}
are both $\Sigma $-algebra morphisms from $\mS$ to $(\mS\times B(\mS,A),b)$;
hence they are equal to $\iter b$ by initiality of $\mS$.

The morphism \eqref{eq:law_comm2} is a composition of $\Sigma $-algebra morphisms:
$\brks{\id,\,\iota^\clubsuit}$ is so by definition, and $\id\times B(\id,\iter a)$
is so by commutativity of the following diagram:
\begin{equation*}
\begin{tikzcd}[column sep = 8em]
\Sigma (\mS \times B(\mS ,\mS ))
  \ar[dd,swap,"{\brks{\iota\comp\Sigma\fst,\, \rho_{\mS ,\mS }}}"]
  \ar[r,"{\Sigma (\id\times B(\id,\iter a))}"]
&
\Sigma (\mS\times B(\mS, A))
  \ar[d,"{\brks{\iota\comp\Sigma \fst,\,\rho_{\mS ,A}}}"]
\\
&
\mS\times B(\mS, \Sigmas (\mS+A))
  \ar[d,"{\id\times B(\id, \Sigmas(\iter a + \id))}"]
\\
\mS\times B(\mS ,\Sigmas  (\mS +\mS ))
  \ar[d,swap,"{\id\times B(\id,\Sigmas\nabla)}"]
  \ar[r,"{\id\times B(\id, \Sigmas(\iter a+\iter a))}"]
  \ar[ur,near end,bend left=10pt,"{\id\times B(\id,\Sigmas(\id+\iter a))}"]
&
\mS \times B(\mS ,\Sigmas (A+A))  
  \ar[d,"{\id\times B(\id,\Sigmas\nabla)}"]
\\
\mS\times B(\mS,\Sigmas (\mS) )
  \ar[d,swap,"{\id\times B(\id,\hat \iota)}"]
  \ar[r,"{\id\times B(\id,\Sigmas(\iter a))}"]
&
\mS\times B(\mS ,\Sigmas  A)
  \ar[d,"{\id\times B(\id,\hat a)}"]
\\
\mS\times B(\mS,\mS )
  \ar[r,"{\id\times B(\id,\iter a)}"]
&
\mS \times B(\mS ,A)
\end{tikzcd}
\end{equation*}
The three upper cells commute by naturality of $\rho$ and by functoriality of $B$
in the second argument; the bottom cell commutes because $\iter (a)\c {\mS\to A}$ 
is a $\Sigma $-algebra morphism. 

That the morphism~\eqref{eq:law_comm1} is a $\Sigma $-algebra morphism
is shown from the following diagram:
\begin{equation*}
\begin{tikzcd}[column sep=2em, row sep=normal]
  \Sigma(\mS)
  \ar[dddd,"\iota"'] 
  \rar["\Sigma\brks{\id,\, a^\clubsuit}"]
  & 
  \Sigma (\mS \times{B(A,A)})
  \rar["\Sigma (\id\times{B(\iter a,\,\id)})"]
  \dar["\brks{\iota\comp\Sigma\fst,\,\Sigma (\iter a\times\id)}"']
  &[3em]
  \Sigma (\mS\times {\BmS A })
  \dar["\brks{\iota\comp\Sigma\fst,\,\rho_{\mS ,A}}"]
  \\
  & \mS\times \Sigma (A\times B(A,A))
  \dar["\id\times\rho_{A,A}"'] &
  \mS\times\BmS{\Sigma^\star (\mS +A)}
  \dar["{\id\times\BmSf{\Sigma^\star(\iter a+\id)}}"] \\
  &
  \mS \times B(A,\Sigma^\star (A+A))
  \dar["{\id\times B(\id,\Sigma^\star\nabla)}"']
  \rar["{\id\times B(\iter a,\,\id)}"] &
  \mS \times\BmS{\Sigma^\star (A+A)}
  \dar["{\id\times\BmSf{\Sigma^\star\nabla}}"] \\
  &
  \mS\times B(A,\Sigma^\star A)
  \rar["{\id\times B(\iter a,\,\id)}"]
  \dar["{\id\times B(\id,\hat a)}"'] &
  \mS \times \BmS{\Sigma^\star A}
  \dar["{\id\times\BmSf{\hat a}}"] \\
  \mS
  \rar["\brks{\id,\,a^\clubsuit}"]&
  \mS \times B(A,A)
  \rar["{\id\times B(\iter a,\,\id)}"] &
  \mS \times \BmS A
\end{tikzcd}
\end{equation*}
The left cell commutes by definition of~$a^\clubsuit$. The two lower right cells 
commute by functoriality of $B$, and the
right upper cell commutes by an instance of dinaturality for~$\rho$:
\begin{equation*}
\begin{tikzcd}[column sep=-2em,row sep=4ex, baseline = (B.south)]
  &%
  \Sigma(\mS\times B(\mS ,A))
  \rar["\rho_{\mS ,A}"]
  &[4em]
  B(\mS, \Sigma^\star (\mS +A))
  \ar{dr}{{B(\id,\Sigma^\star(\iter a+\id))}}
  \\
  \Sigma (\mS\times{B(A,A)})
  \ar[ru,"{\Sigma(\id\times{B(\iter a,\id)})}"]
  \ar[rd,"{\Sigma(\iter a\times\id)}"']
  & & &[-1.5ex]
  B(\mS ,\Sigma^\star(A+A))
  \\
  &
  |[alias = B]|
  \Sigma (A\times{B(A,A)})
  \rar["\rho_{A,A}"]
  &%
  B(A,\Sigma^\star (A+A))
  \ar[ur,"{B(\iter a,\id)}"'] \tag*{\qedhere}
\end{tikzcd}     
\end{equation*}
\noqed
\end{proof}
\takeout{
\begin{lemma}%
For every $\Sigma $-algebra $(A,a)$ the following diagram commutes:
\begin{equation*}
\begin{tikzcd}[column sep=10em, row sep=normal]
\mS 
  \rar["\iter\brks{\id,\,\iota^\clubsuit}"]
  \dar["\iter \brks{\iter a,\,a^\clubsuit}"'] &
\mS \times B(\mS ,\mS )
  \dar["{\iter a\times B(\mS ,\iter a)}"] \\
A\times B(A,A)
  \rar["{A\times B(\iter a,A)}"] &
A\times {B(\mS ,A)}
\end{tikzcd}
\end{equation*}
\end{lemma}
\begin{proof}
  First, recall that $A\times {B(A,A)}$ and
  $\mS \times {B(\mS ,\mS )}$ are $\Sigma $-algebras with the
  structures~$\brks{\iter a,\,a^\clubsuit}$ and $\brks{\id,\,\iota^\clubsuit}$ obtained as
  in~\eqref{eq:iter}. In addition, we define $\Sigma $-algebra
  structures $\alpha$ and $\alpha'$ on $\mS \times B(A,A)$ and
  $\mS \times B(\mS ,A)$, respectively, as follows:
  \[
    \begin{tikzcd}[column sep = 60]
      \Sigma (\mS \times B(A,A))
      \ar{d}{\brks{\iota\comp \Sigma \fst,\, \Sigma (\iter a\times B(A,A))}}
      \ar[shiftarr = {xshift=-60}]{dddd}[swap]{\alpha}
      &
      \Sigma (\mS \times B(\mS ,A))
      \ar{d}[swap]{ \brks{ \iota\comp \Sigma \fst,\,\rho_{\mS ,A} }  }
      \ar[shiftarr = {xshift=70}]{dddd}{\alpha'}
      \\
      \mS  \times \Sigma (A\times B(A,A))
      \ar{d}{\mS \times \rho_{A,A}}
      &
      \mS  \times B(\mS ,\Sigmas (\mS +A))
      \ar{d}[swap]{\mS \times B(\mS ,\Sigmas (\iter a + A))}
      \\
      \mS  \times B(A,\Sigmas  (A+A))
      \ar{d}{\mS  \times B(A,\Sigmas \nabla) }
      &
      \mS  \times B(\mS ,\Sigmas (A+A))
      \ar{d}[swap]{\mS  \times B(\mS ,\Sigmas \nabla) }
      \\
      \mS  \times B(A,\Sigmas  A)
      \ar{d}{\mS  \times B(A,\hat a)}
      &
      \mS  \times B(\mS ,\Sigmas  A)
      \ar{d}[swap]{\mS  \times B(\mS ,\hat a)}
      \\
      \mS  \times B(A,A)
      &
      \mS  \times B(\mS ,A)
    \end{tikzcd}
  \]

  It is our task to prove that the outside of the diagram below commutes: 
  \[
    \begin{tikzcd}[column sep=50]
      \mS 
      \ar[shiftarr = {xshift=-60}]{dd}[swap]{\iter \brks{\iter a,\,a^\clubsuit}}
      \ar{r}{\iter\brks{\id,\,\iota^\clubsuit}}
      \ar{d}[swap]{\iter \alpha}
      \ar{rd}[near end]{\iter \alpha'}
      & 
      \mS \times B(\mS ,\mS )
      \ar[shiftarr = {xshift=70}]{dd}{\iter a\times B(\mS ,\iter a)}
      \ar[dashed]{d}{\mS  \times B(\mS ,\iter a)}
      \\
      \mS \times B(A,A)
      \ar[dashed]{d}[swap]{\iter a \times B(A,A)}
      \ar[dashed]{r}[swap]{\mS  \times B(\iter a, A)}
      &
      \mS \times B(\mS ,A)
      \ar{d}{\iter a \times B(\mS , A)}
      \\
      A\times B(A,A)
      \ar{r}{A\times B(\iter a,A)}
      & 
      A\times B(\mS ,A)
    \end{tikzcd}
  \]
  The lower rectangle and the right-hand part obviously commute. To
  see that the remaining three parts commutes, we will prove that the
  three dashed arrows are appropriate $\Sigma $-algebra morphisms; the
  desired parts then commute by the initiality of $\mS $.
  \begin{enumerate}
  \item To show that
    $\iter(a)\times B(A,A)\c \mS \times B(A,A)\to A\times B(A,A)$ is a
    morphism, we prove that the following diagram commutes:
    \[
      \begin{tikzcd}[column sep=8em]
        \Sigma (\mS \times B(A,A))
        \ar{d}{\brks{\iota\comp \Sigma \fst,\. \Sigma (\iter a\times B(A,A))}}
        \ar{r}{\Sigma (\iter a \times B(A,A))}
        \ar[shiftarr = {xshift=-70}]{dddd}[swap]{\alpha}
        &
        \Sigma (A\times B(A,A))
        \ar{dd}{ \brks{ a\comp \Sigma \fst,\,\rho_{A,A} }  }
        \ar[shiftarr = {xshift=60}]{dddd}{\brks{\iter a,\,a^\clubsuit}}
        \\
        \mS  \times \Sigma (A\times B(A,A))
        \ar{d}[swap]{\mS \times \rho_{A,A}}
        \\
        \mS  \times B(A,\Sigmas  (A+A))
        \ar{d}[swap]{\mS  \times  B(A,\Sigmas \nabla) }
        \ar{r}{\iter a\times B(A,\Sigmas (A+A))}
        &
        A \times B(A,\Sigmas (A+A))
        \ar{d}{A \times  B(A,\Sigmas \nabla) }
        \\
        \mS  \times B(A,\Sigmas  A)
        \ar{d}[swap]{\mS  \times B(A,\hat a)}
        \ar{r}{\iter a\times B(A,\Sigmas  A)}
        &
        A \times B(A,\Sigmas  A) \ar{d}{A \times B(A,\hat a)}
        \\
        \mS  \times B(A,A)
        \ar{r}{\iter a\times B(A,A)}
        &
        A \times B(A,A)
      \end{tikzcd}
    \]
    Indeed, the upper part commutes because $\iter(a)$ is a morphism
    of $\Sigma $-algebras, and the other parts trivially commute.

  \item Next, we show that
    $\mS \times B(\mS ,\iter(a))\c \mS \times B(\mS ,\mS ) \to
    \mS \times B(\mS ,A)$ is a $\Sigma$-algebra morphism, which amounts to
    commutativity of the following diagram:
    \[
      \begin{tikzcd}[column sep = 80]
        \Sigma (\mS \times B(\mS ,\mS ))
        \ar{dd}[swap]{\brks{\iota\comp \Sigma \fst,\, \rho_{\mS ,\mS } }}
        \ar{r}{ \Sigma (\mS \times B(\mS ,\iter a)) }
        \ar[shiftarr = {xshift=-70}]{dddd}[swap]{\brks{\id,\,\iota^\clubsuit}}
        &
        \Sigma (\mS \times B(\mS ,A))
        \ar{d}{ \brks{ \iota\comp \Sigma \fst,\,\rho_{\mS ,A} }  }
        \ar[shiftarr = {xshift=90}]{dddd}{\alpha'}
        \\
        &
        \mS  \times B(\mS ,\Sigmas (\mS +A))
        \ar{d}{\mS \times B(\mS , \Sigmas (\iter a + A))}
        \\
        \mS  \times B(\mS ,\Sigmas  (\mS +\mS ))
        \ar{d}[swap]{\mS  \times B(\mS ,\Sigmas \nabla) }
        \ar{r}[swap]{\mS \times B(\mS ,\Sigmas (\iter a+\iter a))}
        \ar{ur}[near end]{\mS \times B(\mS , \Sigmas (\mS +\iter a)) }
        &
        \mS  \times B(\mS ,\Sigmas (A+A))  \ar{d}{\mS  \times B(\mS ,\Sigmas \nabla) }
        \\
        \mS  \times B(\mS ,\Sigmas  \mS )
        \ar{d}[swap]{\mS  \times B(\mS ,\hat \iota)}
        \ar{r}{ \mS \times B(\mS ,\Sigmas  \iter a) }
        &
        \mS  \times B(\mS ,\Sigmas  A)
        \ar{d}{\mS  \times B(\mS ,\hat a)}
        \\
        \mS  \times B(\mS ,\mS )
        \ar{r}{ \mS \times B(\mS , \iter a)  }
        &
        \mS  \times B(\mS ,A)
      \end{tikzcd}
    \]
    The upper part commutes by naturality of $\rho_{\mS ,-}$ and the
    lower part because $\iter(a)\c \mS  \to A$ is a $\Sigma $-algebra
    morphism. The two remaining parts commute trivially.

  \item Finally, we prove that 
    $\mS \times B(\iter(a),A)\c \mS \times{B(A,A)}\to \mS \times
    B(\mS ,A)$ is a $\Sigma$-algebra morphism. The relevant diagram is
    \[
      \begin{tikzcd}[column sep=8em, row sep=normal]
        \Sigma (\mS \times{B(A,A)})
        \rar["\Sigma (\mS \times{B(\iter a,\,A)})"]
        \ar{d}{\brks{\iota\comp\Sigma \fst,\,\Sigma (\iter a\times  B(A,A))}}
        \ar[shiftarr = {xshift=-70}]{dddd}[swap]{\alpha}
        &
        \Sigma (\mS \times {B(\mS ,\,A)})
        \dar["\brks{\iota\comp\Sigma \fst,\,\rho_{\mS ,A}}"]
        \ar[shiftarr = {xshift=90}]{dddd}{\alpha'}
        \\
        \mS \times \Sigma (A\times B(A,A))
        \dar["\mS \times\rho_{A,A}"']
        &
        \mS \times B(\mS ,\Sigma^\star (\mS +A))
        \dar["{\mS \times B(\mS ,\Sigma^\star(\iter a+A))}"]
        \\
        \mS \times B(A,\Sigma^\star (A+A))
        \dar["{\mS \times B(A,\Sigma^\star\nabla)}"']
        \rar["{\mS \times B(\iter a,\,\Sigma^\star(A+A))}"]
        &
        \mS \times B(\mS ,\Sigma^\star (A+A))
        \dar["{\mS \times B(\mS ,\Sigma^\star\nabla)}"]
        \\
        \mS \times B(A,\Sigma^\star A)
        \rar["{\mS \times B(\iter a,\,\Sigma^\star A)}"]
        \dar["{\mS \times B(A,\hat a)}"']
        &
        \mS \times B(\mS ,\Sigma^\star A)
        \dar["{\mS \times B(\mS ,\hat a)}"]
        \\
        \mS \times B(A,A)
        \rar["{\mS \times B(\iter a,\,A)}"]
        &
        \mS \times B(\mS ,A)
      \end{tikzcd}
    \]
    The two lower parts commute trivially. For the upper rectangle we
    postcompose with the projections of the product in the lower
    right-hand corner: the left-hand component is easily seen to
    commute, and for the right-hand component we use the following
    instance of dinaturality for $\rho_{-,A}$:
    \[
      \begin{tikzcd}[column sep = -20, row sep = 15, baseline = (B.base)]
        &
        \Sigma (\mS \times B(\mS ,A)
        \ar{r}{\rho_{\mS ,A}}
        &[6em]
        B(\mS ,\Sigma^\star (\mS +A))
        \ar{dr}[near end]{B(\mS ,\Sigma^\star (\iter a+A))}
        \\
        \Sigma (\mS \times{B(A,A)})
        \ar{ru}[near start]{\Sigma (\mS \times{B(\iter a,A)})}
        \ar{rd}[near start,swap]{\Sigma (\iter a\times B(A,A))}
        & & &[-2em]
        B(\mS ,\Sigma^\star (A+A))
        \\
        &
        \Sigma (A\times{B(A,A)})
        \ar{r}{\rho_{A,A}}
        &[2em]
        |[alias = B]|
        B(\Sigma^\star (A+A),A)
        \ar{ur}[near end,swap]{B(\iter a,\Sigma^\star (A+A))}
      \end{tikzcd}
    \]
  \end{enumerate}
\end{proof}
}
Using the universal property of the pullback~\eqref{eq:kernel}, we obtain a morphism
${d\colon \mS \to E}$ such that $p_1 \comp d = \id$ and $p_2 \comp d = \id$. (This only needs the definition of the pullback, not the previous lemma.)
It follows that ${p_1^\star, p_2^\star\colon \Sigma^\star E \parto \mS}$ is a reflexive pair in~$\C$ with common section $\eta_E \comp d$, where~$\eta$ is the unit of the monad~$\Sigma^\star$ and $p_i^\star$ is the free extension of $p_i\colon E\to \mS$. By our assumptions, the coequalizer of $p_1^\star$ and $p_2^\star$ is
preserved by the functor~$\Sigma$. Hence, there exists a $\Sigma$-algebra structure
${\iotaq\c\Sigma(\mSq)\to\mSq}$, obtained using the universal property
of the coequalizer $\Sigma(\iter(\iotaq))$ from the diagram
\begin{equation}\label{eq:it-i-quot-coeq}
\begin{tikzcd}[column sep=5em, row sep=normal]
\Sigma\Sigma^\star E\dar["\iota_E"']
  \ar[r,shift right=.5ex,"\Sigma p_2^\star"']
  \ar[r,shift left=.5ex,"\Sigma p_1^\star"]
  &
  \Sigma(\mS)\rar["\Sigma (\iter\iotaq)"]\dar["\iota"]
  &
  \Sigma(\mSq) \ar[dashed]{d}{\iotaq}
  \\
  \Sigma^\star E
  \ar[r,shift right=.5ex,"p_2^\star"']
  \ar[r,shift left=.5ex,"p_1^\star"] &
  \mS\rar["\iter\iotaq"] &  \mSq.
\end{tikzcd}
\end{equation}
Here $\iota_E$ denotes the $\Sigma$-algebra structure of the free algebra $\Sigmas E$, 
and we already denote the coequalizer of $p_1^\star$ and
$p_2^\star$ by $\iter(\iotaq)$, as commutation of the right-hand side
identifies it as the unique $\Sigma$-algebra morphism induced by $\iotaq$.
\begin{lemma}\label{lem:mSq-coalg}
Under the conditions of \Cref{th:main}, there exists a coalgebra structure
  $\varsigma\c\mSq\to B(\mSq,\mSq)$ making the triangle below commute,
  where $\iotaq^\clubsuit=(\iotaq)^\clubsuit$:
  \[
    \begin{tikzcd}[row sep = 2ex]
      & \mS \ar{dl}[swap]{\iter \iotaq} \ar{dr}{\iotaq^\clubsuit} & \\
      \mSq \ar{rr}{\varsigma} & & B(\mSq,\mSq) 
    \end{tikzcd}
  \]
\end{lemma}
\begin{proof}
By definition of $\iter(\iotaq)$ as a coequalizer of $p_1^\star$ and $p_2^\star$,
it suffices to show that $\iotaq^\clubsuit$ also coequalizes $p_1^\star$ and $p_2^\star$,
which we strengthen to $\brks{\iter\iotaq,\, \iotaq^\clubsuit}\comp p_1^\star=\brks{\iter\iotaq, \iotaq^\clubsuit}\comp p_2^\star$. Since $\brks{\iter\iotaq,\, \iotaq^\clubsuit}\comp p_1^\star$ and $\brks{\iter\iotaq,\, \iotaq^\clubsuit}\comp p_2^\star$
are $\Sigma$-algebra morphisms, see diagram \eqref{diag:w-a-clubs-2}, whose domain is the free $\Sigma$-algebra $\Sigma^\star E$,
it suffices to show that the desired equation holds when precomposed with $\eta_E\colon E \to \Sigma^\star E$.
Thus, it remains to show that $\brks{\iter\iotaq,\, \iotaq^\clubsuit}\comp p_1=\brks{\iter\iotaq,\, \iotaq^\clubsuit}\comp p_2$,
which, in turn, reduces to $\iotaq^\clubsuit\comp p_1 = \iotaq^\clubsuit\comp p_2$.
We have
\begin{flalign*}
&&& \iotaq^\clubsuit\comp p_1=\iotaq^\clubsuit\comp p_2\\
&&\iff\; & B(\iter\iotaq,\id)\comp \iotaq^\clubsuit\comp p_1= B(\iter\iotaq,\id)\comp \iotaq^\clubsuit\comp p_2&\text{$B(\iter\iotaq,\id)$ is mono}\\
&&\iff\; & \BmSf{\iter\iotaq}\comp \iota^\clubsuit\comp p_1= \BmSf{\iter\iotaq}\comp\iota^\clubsuit\comp p_2.& \text{\Cref{lem:law_comm}}
\end{flalign*}
Let us denote the coequalizer of $p_1, p_2$ by $q\colon \mS \to
Q$. Since~$ \iter(\iotaq)$ coequalizes $p_1$ and $p_2$, it factorizes
through $q$. It thus suffices to show that
\begin{displaymath}
  \BmSf{q}\comp\iota^\clubsuit\comp p_1
  =
  \BmSf{q}\comp\iota^\clubsuit\comp p_2.
\end{displaymath}
By regularity of the base category~$\gcat$, the unique morphism
$m\colon Q \to Z$ such that $\coit(\iota^\clubsuit) = m \cdot q$ is monic. Since $\BmSf{\argument}$ 
preserves monomorphisms, it suffices to show that
\begin{displaymath}
  \BmSf{\coit\iota^\clubsuit} \comp\iota^\clubsuit\comp p_1
  =
  \BmSf{\coit\iota^\clubsuit} \comp\iota^\clubsuit\comp p_2.
\end{displaymath}
Note that
$\BmSf{\coit\iota^\clubsuit} \comp\iota^\clubsuit = \zeta\comp
\coit\iota^\clubsuit$ since $\coit(\iota^\clubsuit)$ is a coalgebra
morphism from $(\mS, \iota^\clubsuit)$ to $(Z, \zeta)$.  Hence the
above equation follows from
$\coit\iota^\clubsuit\comp p_1 = \coit\iota^\clubsuit\comp p_2$, which
holds by~\eqref{eq:kernel}.
\end{proof}

\noindent These preparations in hand, we can proceed with the proof of the main result.%
\begin{proof}[Proof of~\Cref{th:main}]
Using the coalgebra $\varsigma\c\mSq\to B(\mSq,\mSq)$ from \Cref{lem:mSq-coalg}
together with \Cref{lem:law_comm}, the upper
rectangular cell of the following diagram commutes:
\begin{equation}\label{eq:cong-argument}
\begin{tikzcd}[column sep=2em, row sep=4ex]
\mS  
  \rar[rr,"\iota^\clubsuit"]
  \dar["\iter\iotaq"']
  \ar[dr, "\iotaq^\clubsuit", near end]
  \ar[shiftarr = {xshift=-30}]{dd}[swap]{\coit\iota^\clubsuit} 
& &[6ex]
B(\mS, \mS)
  \dar["{B(\id, \iter\iotaq)}"]
  \ar[shiftarr = {xshift=45}]{dd}{B(\id,\coit\iota^\clubsuit)}
\\
\mSq
  \ar[r,"\varsigma"]
  \ar{d}[swap]{m}
& B(\mSq,\mSq)
  \rar["{B(\iter\iotaq,\id)}"] 
&
B(\mS,  \mSq)
  \ar{d}{B(\id, m)}
\\
  Z
  \ar[rr,"\zeta"] 
& &
  B(\mS, Z)
\end{tikzcd}
\end{equation}
By finality of $(Z,\zeta)$, we also have a morphism $m$ such
that the lower rectangular cell commutes. Therefore,
$\coit(\iota^\clubsuit) = m \comp \iter\iotaq$ by uniqueness of
$\coit(\iota^\clubsuit)$. From this we derive the desired result as
follows. First, we obtain a $\Sigma^\star$-algebra structure
$e\c\Sigma^\star E\to E$ such that $p_1\comp e = p_1^\star$ and
$p_2\comp e = p_2^\star$ by the universal property of~$E$ as the
pullback~\eqref{eq:kernel}, using that
\[\coit(\iota^\clubsuit)\comp p_1^\star = m\comp(\iter\iotaq)\comp
p_1^\star= m\comp(\iter\iotaq)\comp p_2^\star =
\coit(\iota^\clubsuit)\comp p_2^\star\]
by~\eqref{eq:it-i-quot-coeq}. This entails that the pair $p_1,p_2$
has the same coequalizer as $p_1^\star,p_2^\star$, i.e.\ $\iter(\iotaq)$. Indeed,
$\iter(\iotaq)$ coequalizes $p_1 = p_1^\star\comp\eta_E$ and $p_2 = p_2^\star\comp\eta_E$
by definition, and every morphism that coequalizes $p_1$ and $p_2$ must coequalize
$p_1\comp e = p_1^\star$ and $p_2\comp e = p_2^\star$, and hence factor uniquely through $\iter(\iotaq)$. 
Since $\iter(\iotaq)$ is the coequalizer of $p_1,p_2$ in $\C$, and moreover $\iter(\iotaq)$ is a $\Sigma$-algebra morphism by~\eqref{eq:it-i-quot-coeq}, we can use \Cref{rem:proof-strategy} to conclude the proof.
\end{proof}
\begin{example}\label{exa:xcl}
  The $\xCL$ calculus (\Cref{sec:unary-ski}) satisfies the
  assumptions of \Cref{th:main}: $\Set$ is a regular category, every
  polynomial functor $\Sigma$ preserves reflexive coequalizers (see
  \Cref{rem:reflexive}), and the behaviour functor $B(X,Y)=Y+Y^X$
  maps surjections to injections in the contravariant argument and
  preserves injections in the covariant one. Consequently,
  compositionality of $\xCL$ (\Cref{prop:skicong1}) is
  an instance of \Cref{th:main}. More generally, the compositionality of $\HO$ specifications (\Cref{prop:cong-ho}) follows from \Cref{th:main}.
\end{example} 
\begin{example}
  The nondeterministic $\xCL$ calculus (\Cref{sec:nd-ski})
  is handled analogously; just observe that the finite power set
  functor $\mypowfin$ preserves both surjections and injections. Thus,
  compositionality (\Cref{prop:skicong2}) again follows from \Cref{th:main}.
\end{example}

\subsection{Higher-order bialgebras}\label{sec:bialg}
We conclude this section with a bialgebraic perspective on higher-order GSOS laws. First, we isolate the underlying higher-order notion of coalgebra:

\begin{definition}[Higher-Order Coalgebra] A \emph{higher-order coalgebra} for a mixed variance bifunctor $B\colon \C^\opp\times \C\to \C$ is a pair $(C,c)$ of an object $C\in \C$ and a morphism $c\colon C\to B(C,C)$. A \emph{morphism} from $(C,c)$ to a higher-order coalgebra $(C',c')$ is a morphism $h\colon C\to C'$ of $\C$ such that the following diagram commutes:
\begin{equation*}
\begin{tikzcd}[column sep=5ex, row sep=normal]
C
  \ar[rr,"c"]
  \dar["h"'] & &[4ex]
B(C,C)
  \dar["{B(\id,h)}"] \\
C'\ar[r,"c'"] 
& 
B(C',C')
  \rar["{B(h,\id)}"] 
& 
B(C,C')
\end{tikzcd}
\end{equation*}
\end{definition}

\begin{proposition}\label{prop:rho-bialg-cat}
  Higher-order coalgebras for $B$ and their morphisms form a category.%
\end{proposition}

\begin{proof}
Clearly $\id_C$ is a morphism from $(C,c)$ to $(C,c)$. Moreover, the composite $h\comp g$ of two morphisms $g\colon (C,c)\to (C',c')$ and $h\colon (C',c')\to (C'',c'')$ of higher-order coalgebras is again a morphism of higher-order coalgebras by the commutative diagram below:
\begin{equation*}
\begin{tikzcd}[column sep=6ex, row sep=normal, baseline = (B.base)]
C
  \ar[rrr,"c"]
  \dar["g"'] & &[4ex] &[4ex]
B(C,C)
  \dar["{B(\id,g)}"] \\
C'
  \ar[rr,"c'"]
  \dar["h"'] & &
B(C',C')
 \ar[r, "{B(g,\id)}"]
 \dar["{B(\id,h)}"'] &
B(C,C')
  \dar["{B(\id,h)}"] \\
C''
  \ar[r,"c''"] 
& 
B(C'',C'')
  \ar[r,"{B(h,\id)}"] 
&
B(C',C'')
  \ar[r,"{B(g,\id)}"] 
&|[alias=B]|
B(C,C'')\tag*{\qedhere}
\end{tikzcd}
\end{equation*}
\noqed
\end{proof}
A higher-order bialgebra for a higher-order GSOS law $\rho$ combines an algebra structure with a higher-order coalgebra structure compatible with $\rho$:

\begin{definition}[Higher-Order Bialgebra]
Given a $\Pt$-pointed higher-order GSOS law $\rho$, a \emph{$\rho$-bialgebra} is a triple $(A,a,c)$ such that $(A,a)$ is a $\Sigma$-algebra and $(A,c)$ is a higher-order $B$-coalgebra making the following diagram commute:
\begin{equation*}
\begin{tikzcd}[column sep=6ex, row sep=normal]
\Sigma A
  \rar["a"]
  \dar["\Sigma\brks{\id,c}"'] &
 A
  \rar["c"] &[5ex]
B( A, A)\\
\Sigma(A\times B(A, A))
  \rar[r,"\rho_{A,A}"]  
&
B(A,\Sigma^\star( A+ A))
  \rar["{B(\id,\Sigmas\nabla)}"]
&
B(A,\Sigma^\star A)
  \uar["{B(\id,\hat a)}"']
\end{tikzcd}
\end{equation*}
A \emph{morphism} from $(A,a,c)$ to a $\rho$-bialgebra $(A',a',c')$ is a morphism $h\colon A\to A'$ of $\C$ that is both a morphism of $\Sigma$-algebras and of higher-order $B$-coalgebras.
\end{definition}
An \emph{initial} (\emph{final}) $\rho$-bialgebra is simply an initial (final)
object of the category of $\rho$-bialgebras and their morphisms. Similar to the first-order case, the initial $\rho$-bialgebra extends the initial $\Sigma$-algebra:
\begin{proposition}\label{prop:initial-bialgebra}
The triple $(\mS,\iota,\iota^\clubsuit)$ is the initial $\rho$-bialgebra. 
\end{proposition}

\begin{proof}
It follows directly by  
definition of~$\iota^\clubsuit$ in \eqref{diag:can-model}, and by observing that $\iter(\ini) = \id$,
that $(\mS,\iota,\iota^\clubsuit)$ is a $\rho$-bialgebra. To prove initiality, suppose that $(A,a,c)$ is a $\rho$-bialgebra. We show that $\iter /a)\c\mS\to A$
is the unique $\rho$-bialgebra morphism from $(\mS,\iota,\iota^\clubsuit)$ to $(A,a,c)$. To show that $\iter(a)$ is a $\rho$-bialgebra morphism, we need to verify that the diagram
\begin{equation*}
\begin{tikzcd}[column sep=10.2ex, row sep=normal]
\mS
  \ar[rr,"\iota^\clubsuit"]
  \dar["\iter a"']
  \ar[dr,"a^\clubsuit"] & &
B(\mS,\mS)
  \dar["{B(\id,\iter a)}"] 
\\
A\ar[r,"c"] 
& 
B(A,A)
  \rar["{B(\iter a,\id)}"]
& 
B(\mS,A)
\end{tikzcd}
\end{equation*}
commutes. The quadrangular cell commutes by~\Cref{lem:law_comm}, and we are left to
show that $c\comp (\iter a) = a^\clubsuit$.
This follows from the fact that $c\comp (\iter a)$ satisfies the
characteristic property of $a^\clubsuit$ given by \eqref{diag:can-model}. Indeed, the diagram
\begin{equation*}
\begin{tikzcd}[column sep=4ex, row sep=normal, scale cd=.8]
\Sigma(\mS)
  \dar["\Sigma(\iter a)"']
  \ar[rrr,"\iota"] 
  \ar[shiftarr = {xshift=-45}, swap]{dd}{\Sigma\brks{\iter a,\,c\comp\iter a}}
& &[2ex] &[1ex]
\mS
  \dar["\iter a"]
  \ar[shiftarr = {xshift=28}]{dd}{c\comp \iter a}
\\
\Sigma A
  \dar["\Sigma\brks{\id,\,c}"']
  \ar[rrr,"a"] 
& & &
A
  \dar["c"]
\\
\Sigma(A\times {B(A,A)})
  \rar["{\rho_{A,A}}"] 
&
B(A,\Sigma^\star A)
  \rar["{B(\id,\Sigmas\nabla)}"] 
&
B(A,\Sigma^\star (A+A))
  \rar["{B(\id,\hat{a})}"] 
&
B(A,A)
\end{tikzcd}
\end{equation*}
commutes: the top cell commutes by definition of $\iter (a)$, and the bottom one 
commutes by the assumption that $(A,a,c)$ is a $\rho$-bialgebra.

Uniqueness of the bialgebra morphism $\iter(a)\c\mS\to A$ is by initiality of $\mS$ as a $\Sigma$-algebra, since every bialgebra morphism is, by definition, in particular a $\Sigma$-algebra morphism.
\end{proof}
In first-order abstract GSOS, a final bialgebra for a GSOS law can be derived from the final coalgebra $\nu B$. In the higher-order setting, the intended semantic domain for a higher-order GSOS law $\rho$ is $\nu B(\mS,-)$, the final coalgebra for the endofunctor $B(\mS,-)$. However, this object generally does not extend to a final $\rho$-bialgebra. In fact, a final $\rho$-bialgebra usually fails to exist, even for simple `deterministic' behaviour bifunctors:
\begin{example}\label{ex:final-bialgebra}
Consider the bifunctor $B(X,Y)=2^X\c\Set^\opp\times \Set \to \Set$, the empty
signature $\Sigma=\emptyset$, and the unique ($\emptyset$-pointed) higher-order GSOS law $\rho$ of $\Sigma$ over $B$. A $\rho$-bialgebra is just a map $z\c Z\to 2^Z$, and a morphism from a $\rho$-bialgebra $(W,w)$ to $(Z,z)$ is a map $h\c W\to Z$ making the diagram
\begin{equation}\label{eq:bialg-morph-ex}  
\begin{tikzcd}
W\ar{d}[swap]{h} \ar{rr}{w} & & 2^W \ar[equals]{d}\\
Z \ar{r}{z} & 2^Z \ar{r}{2^h} & 2^W
\end{tikzcd}
\end{equation}
commute. We claim that no final $\rho$-bialgebra exists, despite the endofunctor $B(\mS,-)=2^{\mu\Sigma}\cong 1$ having a final coalgebra. Suppose for a contradiction that $(Z,z)$ is a final $\rho$-bialgebra. Choose an arbitrary $\rho$-bialgebra $(W,w)$ such that $\under{W}>\under{Z}$ and $w\c W\to 2^W$ is injective. Then no map $h\c W\to Z$ makes \eqref{eq:bialg-morph-ex} commute, since $w$ is injective but $h$ is not. 
\end{example}
On the positive side, we have an algebra structure
$\iotaq\c\Sigma(\mSq)\to\mSq$ by \Cref{th:main} and a
coalgebra structure $\varsigma\c \mSq\to B(\mSq,\mSq)$ by \Cref{lem:mSq-coalg}, and these combine to a $\rho$-bialgebra:
\begin{proposition}\label{prop:sim-bialg}
  In the setting of \Cref{th:main}, the triple $(\mSq,\iotaq,\varsigma)$ is a $\rho$-bialgebra.
\end{proposition}

\begin{proof}
The outside of the diagram
\begin{equation*}
\begin{tikzcd}[column sep=5ex, row sep=normal, scale cd=.75]
\Sigma(\mS)
  \dar["\Sigma(\iter\iotaq)"']
  \ar[rr,"\iota"] 
  \ar[shiftarr = {xshift=-60}, swap]{dd}{\Sigma\brks{\iter\iotaq,\,\iotaq^\clubsuit}}
& &[3ex] 
\mS
  \dar["\iter\iotaq"]
  \ar[shiftarr = {xshift=40}]{dd}{\iotaq^\clubsuit}
\\
\Sigma(\mSq)
  \dar["\Sigma\brks{\id,\,\varsigma}"']
  \ar[rr,"\iotaq"] 
& & 
\mSq
  \dar["\varsigma"]
\\
\Sigma(\mSq\times {B(\mSq,\mSq)})
  \rar["{\rho_{\mSq,\mSq}}"]
&
B(\mSq,\Sigmas(\mSq+\mSq))
  \rar["{B(\id,\widehat\iotaq\comp\Sigmas\nabla)}"] 
&
B(\mSq,\mSq)
\end{tikzcd}
\end{equation*}
commutes by definition of $\iotaq^\clubsuit$. The side cells commute by~\Cref{lem:mSq-coalg},
and the top middle cell commutes by definition of $\iter(\iotaq)$. Note that $\Sigma(\iter\iotaq)$ is a coequalizer,
since $\iter(\iotaq)$ is a reflexive coequalizer and $\Sigma$ preserves it. Hence $\Sigma(\iter\iotaq)$ is 
epic, and therefore the bottom middle cell commutes, which is the $\rho$-bialgebra
law in question.
\end{proof}
The $\rho$-bialgebra  $(\mSq,\iotaq,\varsigma)$ can thus be regarded as a suitable candidate for a denotational domain in the higher-order setting, despite not being characterized by a universal property.

In conclusion, the above results firmly indicate that bialgebras remain a
meaningful concept in higher-order abstract GSOS, and a potentially useful tool
for deriving congruence results. For instance, in~\citet{UrbatTsampasEtAl23} we
have established a general congruence result with respect to \emph{weak}
(bi)similarity on operational models of higher-order GSOS laws, where the notion
of \emph{lax} higher-order bialgebra figures prominently. Additionally, recent
work~\cite{goncharov2026higherorderbialgebraicdenotationalsemantics} reveals
that it is possible to construct a denotational $\rho$-bialgebra in a manner
that parallels the first-order case.

\section[The Lambda-calculus]{The
  \texorpdfstring{$\boldsymbol{\lambda}$}{$\lambda$}-calculus}
\label{sec:lam}
We now depart from combinatory calculus and move to languages with
variable binding, starting with the all-important (untyped) $\lambda$-calculus.
The $\lambda$-calculus comes in various flavours, such as
\emph{call-by-name} or \emph{call-by-value}, and the respective
operational semantics can be formulated in either big-step or
small-step style. For the purposes of our work, we are going to give a
categorical treatment of the small-step call-by-name and the small-step
call-by-value $\lambda$-calculus. We start with the former, whose operational semantics is
presented in \Cref{fig:lambda}. Here, $p,p',q$ range over possibly open $\lambda$-terms and $[q/x]$ denotes capture-avoiding substitution of the term $q$ for the variable $x$.
\begin{figure}[t]
  \[
    \begin{array}{l@{\qquad}l@{\qquad}l}
      \inference[\texttt{app1}]{\goes{p}{p\pr}} {\goes{p \app q}{p\pr \app q}}
      &
        \inference[\texttt{app2}]{}{\goes{(\lambda x.p) \app q}{p[q/x]}}
    \end{array}
  \]
  \caption{Small-step operational semantics of the call-by-name $\lambda$-calculus.}
  \label{fig:lambda}
\end{figure}

The operational semantics of the call-by-name
$\lambda$-calculus induces a deterministic transition relation $\to$ on the set
of (possibly open) $\lambda$-terms modulo $\alpha$-equivalence. Note that every $\lambda$-term $t$ either \emph{reduces} ($t\to t'$ for some~$t'$) or is in \emph{weak head normal form}, that is, $t$ is a $\lambda$-abstraction $\lambda x.t'$ or of the form $x \app s_{1} \app s_{2} \app \cdots \app
s_{k}$ for a variable $x$ and terms $s_1,\ldots,s_k$ ($k\geq 0$). In particular, a closed term either reduces or is a $\lambda$-abstraction. As usual, we let application associate to the left: $t_1\app t_2\app t_3 \app \cdots \app t_n$ means $(\cdots((t_1\app t_2)\app t_3) \cdots) \app t_n$.

On the side of program equivalences, $\lambda$-calculus semantics can
be roughly divided into three kinds: \emph{applicative
  bisimilarity}~\cite{Abramsky:lazylambda}, \emph{normal form
  bisimilarity}~\cite{DBLP:conf/lics/Lassen05} and \emph{environmental
  bisimilarity}~\cite{DBLP:conf/lics/SangiorgiKS07}. We are looking to
give a coalgebraic account of \emph{strong} versions of applicative bisimilarity, see \Cref{def:strong-app} and \Cref{prop:bisim-vs-appbisim}.

\subsection{The presheaf approach to higher-order languages}
\label{sec:prelims}

\citet{DBLP:conf/lics/FiorePT99} propose the
presheaf category $\vcat$ as a setting for algebraic signatures with
variable binding, such as the $\lambda$-calculus and the
$\pi$-calculus. We review some of the core ideas from their work as
well as follow-up work by \citet{DBLP:conf/lics/FioreT01}.

Let $\fset$ be the category of finite cardinals, a skeleton of the
category of finite sets. The objects of $\fset$ are the
sets $n=\{0,\dots,n-1\}\;(n \in \Nat)$, and morphisms $n\to m$ are
functions. The category $\fset$ has a canonical coproduct structure
\begin{equation}
  \label{coproduct}
  n \xrightarrow{\oname{old}_{n}} n + 1 \xleftarrow{\oname{new}_{n}} 1
\end{equation}
where $\oname{old}_{n}(i)=i$ and $\oname{new}_{n}(0)=n$.
Notice the appropriate naming of the coproduct injections: The idea is
that each object $n \in \fset$ is an untyped context of $n$ free
variables, while morphisms $n \to m$ are variable
\emph{renamings}. When extending a context along
$\oname{old}_{n}(i)=i$, we understand the pre-existing elements
of~$n$ as the ``old'' variables, and the added element
$\oname{new}_{n}(0)$ as the ``new'' variable. The coproduct structure
of $\fset$ gives rise to three fundamental operations on contexts, namely
\emph{exchanging}, \emph{weakening} and \emph{contraction}:
\begin{equation}
  \label{eq:ops}
  \begin{aligned}
    \oname{s} &= [\oname{new}_{1},\oname{old}_{1}] \c 2 \to 2, \\
    \oname{w} &= \oname{old}_{0} \c 0 \to 1, \\
    \oname{c} &= [\id_{1},\id_{1}] \c 2 \to 1.
  \end{aligned}
\end{equation}
We think of a presheaf $X\in \vcat$ as a collection of terms: elements of $X(n)$ are ``$X$-terms'' with
free variables from the set $n=\{0,\ldots,n-1\}$, and for each $r\c n\to m$ the map $X(r)\c X(n)\to X(m)$ sends a term $t\in X(n)$ to the term $X(r)(t)\in X(m)$ obtained by renaming the free variables of $t$ according to $r$.

\begin{example}\label{ex:presheaves}
\begin{enumerate}
\item\label{ex:presheaves-v} The simplest example is the presheaf
  $V\in \vcat$ of variables, defined by
  \begin{equation*}
    V(n) = n \qquad\text{and}\qquad V(r) = r.
  \end{equation*}
  Thus, a $V$-term at stage $n$ is simply a choice of a variable $i \in
  n$.

\item\label{ex:presheaves-sigma} For every algebraic signature
  $\Sigma$, the presheaf $\Sigmas\in \vcat$ of $\Sigma$-terms is given
  by the domain restriction of the free monad on $\Sigma$ to
  $\fset$. Thus $\Sigmas(n)$ is the set of $\Sigma$-terms in variables
  from $n$.
\item\label{ex:presheaves-lambda} The presheaf $\Lambda\in \vcat$ of
  $\lambda$-terms is given by%
  \begin{align*}
    \Lambda(n) &= \text{$\lambda$-terms modulo
      $\alpha$-equivalence with free variables from $n$,}
    \\
    \Lambda(r)(t) & = t[r(0)/0,\ldots,r(n-1)/n-1]\qquad
    \text{for $r\colon n\to m$},
  \end{align*}
where $t[-/-]$ denotes simultaneous substitution in the term $t$. 
\end{enumerate}
\end{example}
The process of substituting terms for variables can be treated at the abstract level of presheaves as follows. For every presheaf $Y \in \vcat$, there is a functor
\[  - \mathbin{\mon} Y \c \vcat \to \vcat \]
given by
\begin{equation}\label{eq:tensor}
 (X \mathbin{\mon} Y)(m) = \int^{n \in \fset} X(n) \product (Y(m))^{n} = \bigl(\coprod_{n\in \fset} X(n)\times (Y(m))^n\bigr)/\approx, 
\end{equation}
where $\approx$ is the equivalence relation generated by all pairs
\[
  (x, y_0,\ldots, y_{n-1}) \approx (x', y_0',\ldots, y_{k-1}')
\]
such that $(x, y_0,\ldots, y_{n-1})\in X(n)\times Y(m)^n$, $(x', y_0',\ldots, y_{k-1}')\in X(k)\times Y(m)^k$ and there
exists $r\c n\to k$ satisfying $x'=X(r)(x)$ and $y_{i}=y'_{r(i)}$ for
$i=0,\ldots, n-1$.  An equivalence class in $(X \mathbin{\mon} Y)(m)$
can be thought of as a term $x\in X(n)$ with $n$ free variables,
together with~$n$ terms $y_0,\ldots,y_{n-1}\in Y(m)$ to be substituted
for them. The above equivalence relation then says that the choice of a term $x
\in X(n)$ and a tuple $y_{0},\dots,y_{n-1}\in Y(m)^n$ should be compatible with renamings:
for instance, if $X = Y = \Lambda$ and $s\in \Lambda(1)\seq \Lambda(2)$, then
for all $t,u \in \Lambda(0)$, the renaming $r\colon 1\to 2$ with $r(0)=1$ witnesses that
$(s \app s,t) \approx (s \app s,u,t)$.

Varying~$Y$,
one obtains the \emph{substitution tensor}
\[\argument \mon \argument \c \vcat \product \vcat \to
\vcat,\] %
which makes $\vcat$ into a (non-symmetric) closed monoidal category with
unit~$V$, the presheaf of variables.  Monoids in $(\vcat, \mon, V)$ can be seen as collections of terms equipped with a
substitution structure.  Closure of $(\vcat, \mon, V)$ is witnessed by the fact that for every $Y\in \vcat$ the functor
$- \mathbin{\mon} Y \c \vcat \to \vcat$ has a right adjoint\footnote[1]{\citet{DBLP:conf/lics/FiorePT99} denote the right adjoint by $\langle Y,-\rangle$; we employ the double bracket notation $\llangle Y,-\rrangle$ instead for distinction with morphisms into products.} given by
\begin{equation*}
  \llangle Y, \argument \rrangle \c \vcat \to \vcat, \qquad \llangle Y,W \rrangle(n) = \int_{m \in \fset} [(Y(m))^{n}, W(m)] = \NT(Y^{n} , W).
\end{equation*}
An element of $\llangle Y,W \rrangle(n)$, viz.\ a natural family
of maps $Y(m)^n\to W(m)$ ($m\in \fset$), is thought of as describing the substitution of $Y$-terms in $m$ variables for the $n$ variables of some fixed ambient term, resulting in  a $W$-term in $m$ variables.

\subsection{Syntax}

Variable binding is captured by the \emph{context extension} endofunctor \[\delta \c \vcat \to \vcat\]
defined on objects by
\begin{equation*}
  \delta X (n) = X(n + 1) \qquad\text{and}\qquad \delta X (h) = X(h + \id_{1})
\end{equation*}
and on morphisms $h \c X \to Y$ by
\begin{equation*}
  (\delta h)_{n} =  \big(X(n + 1) \xra{h_{n + 1}} Y(n + 1)\big).
\end{equation*}
Informally, the elements of $\delta X(n)$ are terms arising by binding the last variable in
an $X$-term with $n+1$ free variables.
The operations $\oname{s},\oname{w},\oname{c}$ on contexts, see \eqref{eq:ops},
give rise to natural transformations
\[
  \oname{swap} \c \delta^{2} \to \delta^{2},
  \qquad
  \oname{up} \c \Id \to \delta \qquad\text{and}\qquad
  \oname{contract} \c \delta^{2} \to \delta
\]
in $\vcat$, which correspond respectively to
the actions of swapping the two ``newest'' variables in a term, weakening and
contraction. Their components are defined by 
\begin{equation}
  \label{eq:ops-2}
  \begin{aligned}
    \oname{swap}_{X,n}
    &=  \big(\,X(n + 2) \xra{X(\id_{n} + \oname{s})} X(n + 2)\,\big), \\
    \oname{up}_{X,n}
    &= \big(\,X(n) \xra{X(\id_{n} + \oname{w})} X(n + 1)\,\big), \\
    \oname{contract}_{X,n}
    &=  \big(\,X(n + 2) \xra{X(\id_{n} + \oname{c})} X(n + 1)\,\big).
  \end{aligned}
\end{equation}

The presheaf $V$ of variables
(\Cref{ex:presheaves}\ref{ex:presheaves-v}) and the endofunctor
$\delta$ are the two main constructs that enable the categorical
modeling of syntax with variable binding. For example, the binding
signature of the $\lambda$-calculus corresponds to the endofunctor
\begin{equation}\label{eq:syn}
  \Sigma \c \vcat \to \vcat,\qquad \Sigma X = V + \delta X + X \product X.
\end{equation}
This is analogous to algebraic signatures determining (polynomial)
endofunctors on $\set$. For~$\Sigma$ as in~\eqref{eq:syn}, the
forgetful functor $\alg{\Sigma} \to \vcat$ has a left adjoint that
takes a presheaf $X \in \vcat$ to the free $\Sigma$-algebra
$\Sigmas X$. In particular, the initial algebra $\mS$ is given by the presheaf
$\Lambda$ of $\lambda$-terms; more precisely:
\begin{proposition}\label{prop:lambda-initial-algebra}
The initial algebra for $\Sigma$ is given by
\[ V+\delta\Lambda+\Lambda\times\Lambda  \xto{[\var,\lambda.(-),\appp]} \Lambda \]
where $\var\colon V\to \Lambda$ is the inclusion of variables, $\lambda.(-)$ sends $t\in \delta\Lambda(n)=\Lambda(n+1)$ to $\lambda n+1.t$, and $\appp$ sends $(t,s)\in \Lambda(n)\times \Lambda(n)$ to $t\, s$.
\end{proposition}
The substitution tensor $\mon$ gives rise to the expected substitution structure on $\lambda$-terms: 
\removeThmBraces
\begin{proposition}
  \label{prop:sub}
  The presheaf $\Lambda = \mS$ of
  $\lambda$-terms admits the structure of a monoid $(\Lambda,
  \mu , \eta)$ in $(\vcat, \mon, V)$ whose unit $\eta\c V\to \Lambda$ is the inclusion of variables and whose multiplication $\mu\colon \Lambda\bullet \Lambda \to \Lambda$ is the uncurried natural transformation $\bar \mu\c  \Lambda\to \llangle \Lambda,\Lambda \rrangle$ given by 
\[ \bar\mu_n \c \Lambda(n)\to \llangle \Lambda,\Lambda \rrangle(n)=\NT(\Lambda^n,\Lambda),\qquad  t\mapsto \lambda \vec{u}\in \Lambda(m)^n. \,t[\vec{u}]. \]
Here, $t[\vec{u}]$ denotes the simultaneous substitution $t[u_0/0,\ldots, u_{n-1}/{n-1}]$. 
\end{proposition}
\resetCurThmBraces

\subsection{Behaviour}
To capture the $\lambda$-calculus in the abstract categorical setting of higher-order GSOS laws developed in \Cref{sec:hogsos}, we consider the behaviour bifunctor 
\begin{equation}
  \label{def:beh}
  B \c (\vcat)^{\opp} \product \vcat \to \vcat,\qquad B(X,Y) = \llangle X,Y \rrangle \product (Y + Y^{X} + 1),
\end{equation}
where $Y^X$ denotes the exponential object in the topos $\vcat$.

\begin{remark}\label{N:exp}
  Exponentials in presheaf categories have a simple explicit description~\cite[Sec.~I.6]{MacLaneMoerdijk92}. The exponential $Y^X$ in $\vcat$
  and its evaluation morphism $\ev\c Y^X \times X \to Y$ are, respectively, given by
  \[
    Y^X(n) = \NT((\argument)^n \times X, Y)
    \qquad
    \text{and}
    \qquad
    \ev_n(f, x) = f_n(\id_n, x) \in Y(n)
  \]
  for every natural transformation $f\c (\argument)^n \times X \to Y$ and $x
  \in X(n)$. In the following we put \[f(x) \,:=\,\ev_n(f,x).\]
\end{remark}
Our intended operational model of the $\lambda$-calculus is a $B(\Lambda,-)$-coalgebra
\begin{equation}\label{eq:lambda-op-model}
\langle \gamma_{1},\gamma_{2}\rangle \c \Lambda \to
\llangle \Lambda,\Lambda \rrangle \product (\Lambda + \Lambda^{\Lambda} + 1)
\end{equation}
on the presheaf of $\lambda$-terms. For each term $t\in \Lambda(n)$, the natural transformation $\gamma_1(t)\c {\Lambda^n\to \Lambda}$ exposes the simultaneous substitution structure, that is, $\gamma_1(t)$ is equal to $\bar \mu(t)$ from \Cref{prop:sub}. Similarly,
 $\gamma_{2}(t)$ is an element of the coproduct
$\Lambda(n) + \Lambda^{\Lambda}(n) + 1$, representing either a reduction step, a
$\lambda$-abstraction seen as a function on terms, or that $t$ is stuck. To apply the higher-order abstract GSOS framework, let us first note that one of its key assumptions holds:
\begin{lemma}\label{lem:final-coalgebra}
  For every $X\in \vcat$ the functor $B(X,-)$ has a final coalgebra.
\end{lemma}

\begin{proof}
  The functor $B(X,-)$ preserves limits of $\omega^\opp$-chains: the
  right adjoints $\llangle X,- \rrangle$ and $(-)^X$ preserve all limits, and
  limits of $\omega^\opp$-chains commute with products and coproducts
  in $\Set$ and thus in $\vcat$ (using that limits and colimits in
  presheaf categories are formed pointwise). Therefore, dually
  to the classic result by \citet{adamek74}, the functor $B(X,-)$ has a final coalgebra
  computed as the limit of the final $\omega^\opp$-chain
  \[
    1\leftarrow B(X,1) \leftarrow B(X,B(X,1)) \leftarrow B(X,B(X,B(X,1)))  \leftarrow \cdots \tag*{$\qed$}
  \]
\def\qed{}
\end{proof}
Since $B(X,-)$ preserves pullbacks, and hence weakly preserves pullbacks, behavioural
equivalence \eqref{eq:kernel} on
$B(X,-)$-coalgebras coincides with coalgebraic
bisimilarity (see~\cite{DBLP:journals/tcs/Rutten00}). Recall that a
\emph{bisimulation} between $B(X,-)$-coalgebras $W \to B(X,W)$ and
$Z \to B(X,Z)$ is a relation between $W$ and $Z$, i.e.\ a sub-presheaf $R\seq W\times Z$, that can be equipped
with a coalgebra structure $r\colon R\to B(X,R)$ such that the projections $p_1\colon R\to W$ and $p_2\colon R\to Z$ are $B(X,-)$-coalgebra
morphisms. 
\[
\begin{tikzcd}[column sep=5em]
W \ar{d}[swap]{w} & \ar{l}[swap]{p_1} R \ar{r}{p_2} \ar[dashed]{d}{r} & Z \ar{d}{z} \\
B(X,W) & B(X,R) \ar{l}[swap]{B(X,p_1)} \ar{r}{B(X,p_2)} & B(X,Z)
\end{tikzcd}
\]
The next proposition gives an elementary
characterization of bisimulations.

\begin{proposition}
  \label{prop:bisim}
  Given $X\in \vcat$ and two $B(X,\argument)$-coalgebras 
\[\langle c_{1},c_{2}\rangle \c W \to \llangle X,W\rrangle \times (W+W^X+1)\quad\text{and}\quad \langle d_{1},d_{2}\rangle \c Z \to \llangle X,Z\rrangle\times(Z+Z^X+1),\] a family of relations
  $R(n)\seq W(n)\times Z(n)$, $n\in \fset$, is a bisimulation if and
  only if for all $n\in \fset$ and $(w,z)\in R(n)$ the following conditions hold (omitting 
  subscripts of components of the natural transformations
  $c_i$, $d_i$):
  \begin{enumerate}
  \item ($W(r)(w), Z(r)(z)) \in R(m)$ for all $r \c n \to m$;
  \item\label{bisim:2} $(c_{1}(w)(\vec{u})
    , d_{1}(z)(\vec{u}))\in R(m)$ for all $m \in \fset$ and $\vec{u} \in X(m)^{n}$;
  \item\label{bisim:3} $c_{2}(w) = w\pr \in W(n) \implies d_{2}(z) = z\pr\in Z(n) \wedge (w\pr, z\pr)\in R(n)$;
  \item $c_{2}(w) = f \in W^{X}(n) \implies d_{2}(z) = g\in Z^X(n) \wedge \forall e \in
    X(n).\,(f(e), g(e))\in R(n)$;
  \item $c_{2}(w) = \ast \implies d_{2}(z) = \ast$;
  \item $d_{2}(z) = z\pr \in Z(n) \implies c_{2}(w) = w\pr\in W(n) \wedge (w\pr,z\pr)\in R(n)$;
  \item $d_{2}(z) = g \in Z^{X}(n) \implies c_{2}(w) = f\in W^X(n) \wedge
    \forall e \in X(n).~(f(e), g(e))\in R(n)$;
  \item \label{bisim:8} $d_{2}(z) = \ast \implies c_{2}(w) = \ast$.
  \end{enumerate}
\end{proposition}
\noindent Before proving the proposition, let us elaborate on the conditions (1)--(8). Condition (1) states that $B(X,-)$-bisimulations are compatible with
the renaming of free variables: given a renaming $r \c n \to m$, the renamed
terms $W(r)(w)$ and $Z(r)(z)$ are related by $R(m)$. Similarly, condition (2)
states that $B(X,-)$-bisimulations are compatible with substitutions: given a
substitution $\vec{u} \in X(m)^{n}$, the resulting terms $c_{1}(w)(\vec{u})$ and
$d_{1}(z)(\vec{u})$ are related by $R(m)$. Conditions (6)--(8) are symmetric to
(3)--(5); in fact, since $B(X,-)$-coalgebras are
deterministic transition systems, conditions (6)--(8) are implied by (3)--(5) and hence could be dropped. We opted to state (6)--(8) explicitly, as
these conditions become relevant in nondeterministic extensions of the
$\lambda$-calculus.

\begin{proof}[Proof of~\Cref{prop:bisim}]
For the $\Longrightarrow$ direction, suppose that $R$ is a bisimulation, i.e.\ there exists a coalgebra structure $\langle r_1,r_2\rangle$ on $R$ making the diagram below commute, where $p_1,p_2$ are the projections: 
\[
\begin{tikzcd}[column sep=45, row sep=3em, scale cd=.85]
W \ar{d}[swap]{\langle c_1,c_2\rangle} & R \ar{l}[swap]{p_1} \ar{r}{p_2} \ar{d}{\langle r_1,r_2\rangle} & Z \ar{d}{\langle d_1,d_2\rangle}  \\
\llangle X,W\rrangle\times (W+W^X+1) &  \ar{l}[swap,yshift=.5em]{\llangle X,p_1\rrangle\times (p_1+p_1^X+1) } \llangle X,R\rrangle\times(R+R^X+1) \ar{r}[yshift=.5em]{\llangle X,p_2\rrangle\times (p_2+p_2^X+1) }  & \llangle X,Z\rrangle \times (Z+Z^X+1)
\end{tikzcd}
\]
Then (1) holds because $R$ is a sub-presheaf of $W\times Z$, and (2), (3) and (5) are immediate from the above diagram. Concerning (4), let  $w \mathbin{R(n)} z$ and suppose that $c_2(w)=f\in W^X(n)$. Then the above diagram implies that $r_2(w,z)=h\in R^X(n)$ and $d_2(z)=g\in Z^X(n)$. Moreover, for all $e\in X(n)$, using infix notation for the binary relation $R(n)$, we have that
\[ f(e) = \ev(f,e) = p_1(\ev(h,e)) ~~\mathbin{R(n)}~~ p_2(\ev(h,e)) = \ev(g,e) = g(e) \]
where the second and the penultimate equality follow via naturality of $\ev$.

For the $\Longleftarrow$ direction, suppose that $R(n)\seq W(n)\times Z(n)$, $n\in\fset$, is a family of relations satisfying (1)--(8) for all $w\mathbin{R(n)} z$. Condition (1) asserts that $R$ is a sub-presheaf of $W\times Z$; thus it remains to define a coalgebra structure $\langle r_1,r_2\rangle$ on $R$ making the diagram above commute. It suffices to define the components 
\[\langle r_{1,n},r_{2,n}\rangle\c  R(n)\to \llangle X,R\rrangle(n)\times (R(n)+ R^X(n)+1)\] 
and prove that the diagram commutes pointwise at every $n\in \fset$;
the naturality of $r_1,r_2$ then follows since the two lower
horizontal maps in the diagram are jointly monomorphic. Indeed, this
is easy to see using the following diagram, where $f\colon n \to m$ is
a morphism of $\fset$ and we abbreviate $F = B(X,-)$,
$r = \langle r_1, r_2\rangle$ and $c = \langle c_1, c_2\rangle$:
\[
  \begin{tikzcd}
    W(n)
    \ar{ddd}[swap]{Wf}
    \ar{rrr}{c_n}
    &&&
    FW(n)
    \ar{ddd}{FW(f)}
    \\
    &
    R(n)
    \ar{d}[swap]{R(f)}
    \ar{lu}{p_{1,n}}
    \ar{r}{r_n}
    &
    FR(n)
    \ar{d}{FR(f)}
    \ar{ru}[swap]{Fp_{1,n}}
    \\
    &
    R(m)
    \ar{ld}[swap]{p_{1,m}}
    \ar{r}{r_m}
    &
    FR(m)
    \ar{rd}{Fp_{1,m}}
    \\
    W(m)
    \ar{rrr}{c_m}
    &&&
    FW(m)
  \end{tikzcd}
\]
Since its outside and the upper, lower, right and left-hand parts
commute due to the naturality of $c$ and $p_1$ so does the desired
inner square when postcomposed by $Fp_{1,m}$. A similar diagram using
$(Z, d = \langle d_1,d_2\rangle)$ and $p_2$ in lieu of $(W, c)$ and $p_{1}$
shows that the desired square commutes when postcomposed by $Fp_{2,m}$. So
since $Fp_{1,m}, Fp_{2,m}$ is jointly monic, the desired square commutes.

We define 
\[r_{1,n}\c R(n)\to \llangle X,R\rrangle(n)=\NT(X^n,R)\qquad\text{by}\qquad
 r_{1,n}(w,z) = \langle c_{1,n}(w), d_{1,n}(z)\rangle.\]
Condition (2) shows that this map is well-typed and that it makes the first component of the diagram commute. 

 To define $r_{2,n}$, using extensivity of the presheaf topos $\vcat$ we express $R$ as a coproduct  $R=R_0+R_1+R_2$ of the sub-presheaves given by
\begin{align*}
R_0(n) & = \{\, (w,z)\in R(n) : c_2(w)\in W(n),\, d_2(z)\in Z(n) \,\}, \\
R_1(n) & = \{\, (w,z)\in R(n) : c_2(w)\in W^X(n),\, d_2(z)\in Z^X(n) \,\}, \\
R_2(n) & = \{\, (w,z)\in R(n) : c_2(w)=\ast,\, d_2(z)=\ast \,\}.
\end{align*}
Thus, it suffices to define $r_{2,n}\c R_0(n)+R_1(n)+R_2(n)\to R(n)+R^X(n)+1$ separately for each summand of its domain. Given $(w,z)\in R_0(n)$, we put
\[ r_{2,n}(w,z) = (c_2(w),d_2(z))\in R(n). \]
By condition~(3), this is well-typed and makes the second component of the diagram (with domain $R$ restricted to $R_0$) commute. Similarly, for $(w,z)\in R_2(n)$ we put 
\[ r_{2,n}(w,z)=\ast; \]
the second component of the diagram (with domain $R$ restricted to $R_2$) then commutes by condition~(5). Finally, for $(w,z)\in R_1(n)$ we put
\[ r_{2,n}(w,z) = \curry\, h (w,z) \in R^X(n)\]
where $h\colon R_1\times X\to R$ is the natural transformation whose component at $m\in \fset$ is given by
\[ h_m((w',z'),e) = (c_2(w')(e), d_2(z')(e)) \in R(m).  \]
Condition~(4) asserts that $h_m$ is well-typed and that the second
component of the diagram (with domain $R$ restricted to $R_1$)
commutes.
\end{proof}

\subsection{Semantics}

As explained above, in our intended operational model
$\langle \gamma_{1},\gamma_{2}\rangle \c \Lambda \to B(\Lambda,\Lambda)$ for the $\lambda$-calculus, the component $\gamma_1$ should be the transpose of the monoid multiplication
$\mu \c \Lambda \bullet \Lambda \to \Lambda$ from \Cref{prop:sub} under
the adjunction
$\argument \mon \Lambda \dashv \llangle \Lambda , \argument
\rrangle$. As an interesting technical subtlety, for this model to be induced by a
higher-order GSOS law $\rho=(\rho_{X,Y})$ of some sort, the argument $X$ is required to be 
equipped with a \emph{point} $\var \c V \to X$. The importance
of points for defining substitution was first identified by
\cite{DBLP:conf/lics/FiorePT99} (see also \cite{DBLP:conf/lics/Fiore08}) and is
worth recalling from its original source.

Fiore et al.~argued that, given an endofunctor $F \c \vcat \to \vcat$,
in order to define a substitution structure
$F^{\star}V \bullet F^{\star}V \to F^{\star}V$ on the free $F$-algebra
over $V$, it is necessary for $F$ to be \emph{tensorially strong}, in
that there is a natural transformation
$\strength_{X,Y} \c FX \bullet Y \to F(X \bullet Y)$ satisfying the expected
coherence laws~\cite[§3]{DBLP:conf/lics/FiorePT99}. For the
special case of~$F$ being the context extension endofunctor $\delta$,
this
 requires the presheaf $Y$
to be equipped with a \emph{point} $\var \c V \to Y$: the strength map
$\strength_{X,Y} \c \delta X \bullet Y \to \delta(X \bullet Y)$ is
given at $m\in\fset$ by%
\[
  [t \in X(n+1),\vec{u} \in Y(m)^{n}] \quad\xmapsto{\strength_{X,Y}}\quad
  [t,(\oname{up}_{Y,m}(\vec{u}),\var_{m+1}(\oname{new}_{m}))\in Y(m+1)^{n+1}];
\]
where $[-]$ are equivalence classes of the equivalence relation $\approx$ appearing in the definition \eqref{eq:tensor} of $\bullet$. Intuitively, given
a substitution of length $n$ on a term with $n+1$ free variables, a
fresh variable in $Y$ should be used to (sensibly) produce a
substitution of length $n+1$.  This situation is relevant in the
context of higher-order GSOS laws of binding signatures over $B(X,Y)$
where, e.g. in the case of the $\lambda$-calculus, one is asked to
define a map of the form (factoring out the unnecessary parts)
\begin{equation}
  \label{eq:rhodelta}
  \rho_{1} \c \delta\llangle X,Y \rrangle \to \llangle X,\delta Y \rrangle.
\end{equation}
Writing $\overline{\strength}$ for the transpose of $\strength$ under $\argument
\mon X \dashv \llangle X , \argument \rrangle$, we obtain $\rho_{1}$ as the composite
\begin{equation*}
  \begin{tikzcd}
    \delta\llangle X,Y \rrangle
    \arrow[rr,"\overline{\strength}_{\llangle X,Y \rrangle,X}"]
    & & \llangle X, \delta(\llangle X,Y \rrangle \bullet X)\rrangle
    \arrow[rr,"\llangle X {,} \delta (\varepsilon) \rrangle"]
    & & \llangle X, \delta Y \rrangle,
  \end{tikzcd}
\end{equation*}
where $\varepsilon\c \llangle X,Y \rrangle \bullet X\to Y$ is the evaluation morphism for the hom-object  $\llangle X,
Y\rrangle$. In elementary terms, the map $\rho_1$ takes a natural transformation $f\colon X^{n+1}\to Y$ to the natural transformation $\rho_1(f)\c X^n\to \delta Y$ given by
\begin{equation}
  \label{eq:substhelp}
  \vec{u}
  \in X(m)^{n} \quad \mapsto \quad f_{m+1}(\oname{up}_{X,m}(\vec{u}),\var_{m+1}(\oname{new}_m)) \in Y(m+1).
\end{equation}
Thus $\rho_{1}$ represents capture-avoiding simultaneous substitution in which the freshest
variable is bound, hence it should not be substituted. 

At the same time, a higher-order GSOS law for the $\lambda$-calculus
needs to turn a $\lambda$-abstraction into a function on
potentially open terms precisely by only substituting the bound
variable. This implies that we need a natural transformation of the form
\begin{equation}
  \label{eq:rholam2}
  \rho_{2} \c \delta\llangle X,Y \rrangle \to Y^{X}.
\end{equation}
Again, we make use of the point $\var\c V \to X$ to produce $\rho_{2}$:
\begin{equation*}
  \begin{tikzcd}
    \delta\llangle X,Y \rrangle
    \arrow[r,"\cong"]
    & \llangle X,Y^{X} \rrangle
    \arrow[rr,"\llangle \var {,} Y^{X} \rrangle"]
    & & \llangle V, Y^{X}\rrangle
    \arrow[r,"\cong"]
    & Y^{X}.
  \end{tikzcd}
\end{equation*}
Here, the first isomorphism is given at $n\in \fset$ by
\[ \delta\llangle X,Y\rrangle(n) = \NT(X^{n+1},Y) \cong \NT(X^n,Y^X) = \llangle
  X,Y^X\rrangle (n). \]
Thus, in elementary terms, the adjoint transpose $\overline{\rho_{2}}$ of
$\rho_{2}$ acts as follows:
\[ \overline{\rho_2}(f) = \lambda e \in X(n).\,f_{n}(\var_n(0),\ldots, \var_n(n-1),e)\qquad\text{for $f\c X^{n+1}\to Y$.} \]
With these preparations at hand, we are now ready to define the small-step operational semantics of the
call-by-name $\lambda$-calculus in terms of a
$V$-pointed higher-order GSOS law of the syntax endofunctor $\Sigma X = V+\delta X + X \product X$ over the behaviour bifunctor
$B(X,Y)=\llangle X,Y\rrangle \times (Y+Y^X+1)$. A law of this type is given by a family of presheaf maps
\[
\begin{tikzcd}
  V+\delta(X\times \llangle X,Y\rrangle \product (Y + Y^{X} + 1) ) + (X\times \llangle X,Y \rrangle \product (Y + Y^{X} + 1))^2 \ar{d}{\rho_{X,Y}}  \\
 \llangle X,\Sigmas(X+Y)  \rrangle \product (\Sigmas(X+Y) + (\Sigmas(X+Y))^{X} + 1)
\end{tikzcd}
\]
dinatural in $(X,\var_X)\in V/\vcat$ and natural in $Y\in \vcat$. We let $\rho_{X,Y,n}$ denote the component of $\rho_{X,Y}$ at $n\in \fset$. For the definition of $\rho$ we set up some notation:
\begin{notation}
  We write
  \[\lambda.(-)\c \delta\Sigmas\to \Sigmas \qquad\text{and}\qquad
    \appp\c \Sigmas\times \Sigmas\to \Sigmas \] for the natural
  transformations whose components come from the $\Sigma$-algebra structure
  on free $\Sigma$-algebras; here $\appp$ denotes application. In the following we will consider free
  algebras of the form $\Sigmas(X+Y)$. For simplicity, we usually keep inclusion maps implicit: Given $t_1,t_2\in X(n)$ and
  $t_1'\in Y(n)$ we write $t_1\app t_2$
  for $\appp([\eta\comp \inl(t_1)], [\eta\comp \inl(t_2)])$, and similarly
  $t_1\app t_1'$ for $\appp([\eta\comp\inl(t_1)], [\eta\comp \inr(t_1')])$
  etc., where $\inl$ and $\inr$ are the coproduct injections and $\eta\c \Id\to\Sigmas$ is the unit of the free monad $\Sigmas$.
\end{notation}
\begin{notation}
  Let \[\pi\colon V\to \llangle X,\Sigmas(X+Y)\rrangle \] be the adjoint
  transpose of
  \[
    V\bullet X \xto{\cong} X \xto{\inl} X+Y \xto{\eta}
    \Sigmas(X+Y).
  \]
  Thus, for $v\in V(n)=n$, the natural transformation
  $\pi(v)(n)\colon X^n\to \Sigmas(X+Y)$ is the $v$-th projection
  $X^n\to X$ followed by $\eta\comp \inl$. Further, recall 
  that $j \c \Pt/\vcat \to \vcat$ denotes the forgetful functor. 
\end{notation}

\begin{definition}[$V$-pointed higher-order GSOS law for the call-by-name
  $\lambda$-calculus]
  \label{def:lamgsos}
  \[ \rho^{\cn}_{X,Y}  \c \Sigma(jX \times B(jX,Y))\to  B(jX, \Sigma^\star (jX+Y)) \]  
  is given by 
  \begin{align*}
   & \rho^{\cn}_{X,Y,n}(tr) = \texttt{case}~tr~\texttt{of} &&& \\
    & v \in V(n)
    & \mapsto \quad
    & \pi(v),*
    & \\
    & \mathsf{\lambda}.(t, f,\_)
    & \mapsto
      \quad
    & \llangle X, \lambda.(-)\comp \eta \comp \inr \rrangle (\rho_{1}(f)),
      (\eta \comp \inr)^{X}(\rho_{2}(f))
    & \\
    & (t_{1}, g, t_{1}\pr) \app (t_{2}, h,\_)
    & \mapsto
      \quad
    & \lambda m,\,\vec{u} \in X(m)^{n}.\, ( g_{m}(\vec{u}) \app h_{m}(\vec{u}) ),t_{1}\pr \app t_{2}
    & \\
    & (t_{1}, g, k) \app (t_{2}, h,\_)
    & \mapsto
      \quad
    & \lambda m,\,\vec{u} \in X(m)^{n}.\, (g_{m}(\vec{u}) \app h_{m}(\vec{u})),\eta \comp \inr \comp k(t_{2})
    & \\
    & (t_{1}, g, *) \app (t_{2}, h,\_)
    & \mapsto
      \quad
    & \lambda m,\,\vec{u} \in X(m)^{n}.\, (g_{m}(\vec{u}) \app h_{m}(\vec{u})),* &
\end{align*}
where $t\in \delta X(n)$, $f\in \delta\llangle X,Y\rrangle(n)$, $g,h\in \llangle X,Y\rrangle(n)$, $k\in Y^X(n)$, $t_1, t_2\in X(n)$
and $t_1'\in Y(n)$. 
\end{definition}

\begin{rem}
We have omitted brackets around the pairs on the right. The last three clauses refer to the application operator $p\app q= \appp(p,q)$ and could also be written as, e.g.,
\[\appp((t_{1}, g, t_{1}\pr), (t_{2}, h,\_))
     \quad\mapsto\quad
     \lambda m,\,\vec{u} \in X(m)^{n}.\,  \appp( g_{m}(\vec{u}) , h_{m}(\vec{u})) , \appp(t_{1}\pr, t_{2}).  \] 
\end{rem}
In \Cref{fig:cn} the definition of $\rho^\cn$ is rephrased in the
style of inference rules. Here, $[\ldots]$ corresponds to the first 
component of $B(jX,Y)$, and $\to$, $\xto{t}$, $\nrightarrow$ correspond to the 
slots in the second component of $B(jX,Y)$. For instance, the rule \texttt{lam} expresses that for every $t\in \delta X(n)$, $f\in \delta\llangle X,Y\rrangle$ and $e\in X(n)$, putting $\vec{u} = (\var_n(0),\ldots, \var_n(n-1),e)$ and $t'=f(\vec{u})$, the second component of $\rho_{X,Y,n}(\lambda.(t,f,\_))$ lies in $Y^X(n)$ and satisfies $\rho_{X,Y,n}(\lambda.(t,f,\_))(e)=t'$. This matches precisely the corresponding clause of \Cref{def:lamgsos}.

 \begin{figure}[t]
   \[
     \begin{array}{l@{\qquad}l}
       \inference[\texttt{var}]{}{v \nrightarrow}
       & \inference[\texttt{lam}]{
       \vec{u} = (\var_n(0),\ldots, \var_n(n-1),e\in X(n))
       & t[\vec{u}] = t'}{\goesv{\lambda.t}{t\pr}{e}}
     \end{array}
   \]
   \[
     \begin{array}{l@{\qquad}l@{\qquad}l}
       \inference[\texttt{app1}]{\goes{t_{1}}{t_{1}\pr}}{\goes{t_{1} \app t_{2}}{t_{1}\pr \app t_{2}}}
       & \inference[\texttt{app2}]{t_{1} \xrightarrow{t_{2}} t_{1}\pr}{\goes{t_{1} \app t_{2}}{t_{1}\pr}}
       &  \inference[\texttt{app3}]{t_{1}\nrightarrow}{t_{1}\app t_{2} \nrightarrow}
     \end{array}
   \]
  
   \[
     \begin{array}{l@{\qquad}l}
       \inference[\texttt{varSub}]{}{v[\vec{u}] = \vec{u}(v)}
       &
         \inference[\texttt{lamSub}]{\vec{w} =
         (\oname{up}_{X,m}(\vec{u}),\var_{m+1}(\oname{new}_m)) & t[\vec{w}]         
         =t^{\prime\prime}}{(\lambda. t)[\vec{u}]=
         \lambda. t^{\prime\prime}}
     \end{array}
   \]
  
   \[
     \inference[\texttt{appSub}]{t_{1}[\vec{u}]=
       t_{1}\pr & t_{2}[\vec{u}]=t_{2}\pr}{(t_{1} \app
       t_{2})[\vec{u}]=t_{1}\pr \app t_{2}\pr}
   \]
   \caption{Law $\rho^{\cn}$ in the form of inference rules.}
   \label{fig:cn}
 \end{figure}

\begin{rem}\label{rem:op-model-lambda}
Instantiating \Cref{def:operational-model}, the operational model of the higher-order GSOS law $\rho^\cn$ is the $B(\Lambda,-)$-coalgebra 
\begin{equation}\label{eq:op-model-lambda}
  \iota^\clubsuit=\langle \gamma_{1},\gamma_{2}\rangle \c \Lambda \to
  \llangle \Lambda,\Lambda \rrangle \times (\Lambda + \Lambda^\Lambda+1),
\end{equation}
that is uniquely determined by the following commutative diagram:
\[
\begin{tikzcd}[column sep=63, scale cd=.85]
V+\delta\Lambda+\Lambda^2 \ar{ddd}{ \iota=[\var,\lambda.(-),\appp]} \ar{r}[yshift=.5em]{\id+\delta\langle \id,\gamma_1,\gamma_2\rangle + \langle \id,\gamma_1,\gamma_2\rangle^2} & V+\delta(\Lambda\times\llangle \Lambda,\Lambda\rrangle \times(\Lambda+\Lambda^\Lambda+1)) + ( \Lambda\times\llangle\Lambda,\Lambda\rrangle \times (\Lambda+\Lambda^\Lambda+1))^2  \ar{d}{\rho_{\Lambda,\Lambda}^{\cn}} \\
&  \llangle \Lambda,\Sigmas(\Lambda+\Lambda)\rrangle \times (\Sigmas(\Lambda+\Lambda) + (\Sigmas(\Lambda+\Lambda))^\Lambda +1)  \ar{d}{\llangle \id,\Sigmas\nabla\rrangle \times (\Sigmas\nabla + (\Sigmas\nabla)^\Lambda +\id) }  \\
&  \llangle \Lambda,\Sigmas\Lambda\rrangle \times (\Sigmas\Lambda + (\Sigmas\Lambda)^\Lambda +1) \ar{d}{ \llangle \id,\hat\iota\rrangle \times (\hat\iota + {\hat\iota}^\Lambda +\id) } \\
\Lambda \ar{r}{ \langle\gamma_1,\gamma_2\rangle} &  \llangle \Lambda,\Lambda\rrangle \times (\Lambda+ \Lambda^\Lambda +1) 
 \end{tikzcd}
\]
\end{rem}
The following two propositions assert that the coalgebra $\langle \gamma_1,\gamma_2\rangle$
coincides with the intended operational model described in
\eqref{eq:lambda-op-model}, that is, its first component
exposes the substitution structure of $\lambda$-terms and its second
component corresponds to the transition system $\to$ on $\lambda$-terms
derived from the operational semantics in \Cref{fig:lambda}.

\begin{proposition}\label{prop:gamma1}
For every $m,n\in\fset$, $t\in \Lambda(n)$ and $\vec{u}\in \Lambda(m)^n$, we have
\[ \gamma_1(t)(\vec{u}) = t[\vec{u}]. \]
\end{proposition}

\begin{proof}
We proceed by induction on the structure of $t$.
  \begin{itemize}
  \item For $t=v \in V(n)$,  
\[\gamma_1(v)(\vec{u})=\pi(v)(\vec{u}) = u_v = v[\vec{u}];\]
the first equality follows from the definition of $\gamma_1$
(\Cref{rem:op-model-lambda}), the second one from the definition of $\pi$, and
the third one from the definition of substitution.
  \item For $t = \lambda x.t\pr$ (where $x=n$ and $t'\in \Lambda(n+1)$),  
  \[\gamma_{1}(t)(\vec{u}) =
    \lambda m.(\gamma_{1}(t\pr)(\oname{up}_{\Lambda,m}(\vec{u}),m)) =  \lambda m.(t\pr[\oname{up}_{\Lambda,m}(\vec{u}),m]) = t[\vec{u}]; \]
  the first equality uses the definition of $\gamma_1$ in terms of $\rho_{1}$
  \eqref{eq:substhelp}, the second one follows by
induction and the third one by the definition of substitution for the case of
$\lambda$-abstraction.
  \item For $t = t_{1} \app t_{2}$,
\[ \gamma_{1}(t)(\vec{u}) = \gamma_{1}(t_{1})(\vec{u}) \app
    \gamma_{1}(t_{2})(\vec{u}) = t_{1}[\vec{u}] \app t_{2}[\vec{u}] = t[\vec{u}]; \]
 the first equality follows from the definition of $\gamma_1$, the second one by induction, and the third one by the definition of substitution.\hfill$\qed$\par\addvspace{6pt}
  \end{itemize}\def\qed{}
\end{proof}
\begin{proposition}\label{prop:gamma2}
For every $n\in\fset$ and $t\in \Lambda(n)$ the following statements hold:
\begin{enumerate}
\item If $\gamma_2(t) \in \Lambda(n)$ then $t\to \gamma_2(t)$;
\item If $\gamma_2(t)\in \Lambda^\Lambda(n)$ then $t=\lambda x.t'$ for some $t'$, and $\gamma_2(t)(e)=t'[e/x]$ for every $e\in \Lambda(n)$;
\item If $\gamma_2(t)=\ast$ then $t$ is stuck, i.e.\ $t=x\app s_1\app\cdots\app s_k$ for $x\in V(n)$, $k\geq 0$ and $s_1,\ldots, s_k\in \Lambda(n)$.
\end{enumerate}
\end{proposition}

\begin{remark}
Note that partial converses of the above statements are implied:
\begin{enumerate}
\item If $t$ reduces, then $\gamma_2(t)\in \Lambda(n)$.
\item If $t=\lambda x.t'$, then $\gamma_2(t)\in \Lambda^\Lambda(n)$.
\item If $t$ is stuck, then $\gamma_2(t)=\ast$
\end{enumerate} 
For instance, if $t$ reduces, then it can neither hold that $\gamma_2(t)\in \Lambda^\Lambda(n)$ or $\gamma_2(t)=\ast$ by \Cref{prop:gamma2}(2),(3), and therefore $\gamma_2(t)\in \Lambda(n)$. Similarly for the other cases.
\end{remark}

\begin{proof}
We proceed by induction on the structure of $t$:
   \begin{itemize}
  \item For $t=v \in V(n)$, 
we have $\gamma_2(t)=\ast$ by the definition of $\gamma_2$ (\Cref{rem:op-model-lambda}); hence case (3) applies, and $t$ is stuck as claimed.
\item For $t = \lambda x.t\pr$, we have
  $\gamma_2(t)\in \Lambda^\Lambda(n)$ using the definition of
  $\rho^{\cn}$. Then case (2) applies, and for every $e\in \Lambda(n)$ we have
\[ \gamma_2(t)(e) = \rho_2(\gamma_1(t'))(e) = \gamma_1(t')(0,\ldots,n-1,e) = t'[e/x] \]
as claimed, where the first two equalities use the definition of $\gamma_2$ and $\rho_2$, respectively, and the third one follows from \Cref{prop:gamma1}.
  \item For $t = t_{1} \app t_{2}$, we distinguish three cases:
\begin{itemize}
\item If $t_1=x\in V(n)$, then $\gamma_2(t)=\ast$ and $t=x\app t_2$. Hence case (3) applies, and $t$ is stuck as claimed.
\item If $t_1=\lambda x.t_1'$, we have $t\to t':= t_1'[t_2/x]$. Then case (1) applies, and 
\[ \gamma_2(t) = \gamma_2(t_1)(t_2) = t_1'[t_2/x] = t' \]
as claimed, where the first equality uses the definition of $\gamma_2$ and the second one follows by induction.
\item If $t_1=t_{1,1}\app t_{1,2}$, then either case (1) or (3) applies to $t_1$. If case (1) applies to $t_1$, i.e.\ $\gamma_2(t_1)\in \Lambda(n)$, we know by induction that $t_1\to \gamma_2(t_1)$. By definition of $\to$ and $\gamma_2$, this implies $\gamma_2(t)\in \Lambda(n)$ and
\[ t \to \gamma_2(t_1)\app t_2 = \gamma_2(t), \]
proving (1) for $t$. If case (3) applies to $t_1$, i.e.\ $\gamma_2(t_1)=\ast$, we know by induction that $t_1=x\app s_1\app\cdots\app s_k$ for some $x\in V(n)$ and $s_1,\ldots, s_k\in \Lambda(n)$. Then also $\gamma_2(t)=\ast$ and $t=x\app s_1\app\cdots\app s_k \app t_2$, proving (3) for $t$.  \hfill$\qed$\par\addvspace{6pt}
\end{itemize}
  \end{itemize}\def\qed{}
\end{proof}
Let ${\sim^{\Lambda}} \seq \Lambda \times \Lambda$ be the bisimilarity
relation on the coalgebra \eqref{eq:op-model-lambda}.
It turns out that $\sim^\Lambda$ matches the open extension of strong applicative bisimilarity, cf.\ \cite{Abramsky:lazylambda}.
\begin{definition}\label{def:strong-app}
\emph{Strong applicative bisimilarity} is the greatest relation $\sim^\ap_0 \,\subseteq\, \Lambda(0) \times \Lambda(0)$ such that for 
  $t_{1} \sim^\ap_0 t_{2}$ the following conditions hold:
  \begin{align*}
    t_{1} \to t_{1}\pr
    & \;\implies\; \exists 
    t_{2}\pr.~t_{2} \to t_{2}\pr \wedge t_{1}\pr \sim^\ap_0 t_{2}\pr;
    \tag*{(A1)} \\
    t_{1} = \lambda x.t_1'
    &\;\implies\; \exists t_2'.\,t_{2} = \lambda x.t_2' \wedge \forall e \in
    \Lambda(0).~t_1'[e/x]\sim^\ap_0t_2'[e/x];
    \tag*{(A2)} \\
    t_{2} \to t_{2}\pr
    &\;\implies\; \exists t_{1}\pr.~t_{1} \to t_{1}\pr \wedge t_{1}\pr \sim^\ap_0 t_{2}\pr;
    \tag*{(A3)} \\
    t_{2} = \lambda x.t_2'
    &\;\implies\; \exists t_1'.\,t_{1} = \lambda x.t_1' \wedge \forall e \in
    \Lambda(0).~t_1'[e/x]\sim^\ap_0t_2'[e/x].\tag*{(A4)}
  \end{align*}
The \emph{open extension} of strong applicative bisimilarity is the relation $\sim^{\ap}\,\seq \Lambda\times \Lambda$ where $\sim^{\ap}_n\,\seq \Lambda(n)\times \Lambda(n)$ for $n>0$ is given by
\[ t_1 \sim^{\ap}_n t_2 \qquad \text{iff}\qquad t_1[\vec{u}] \sim^{\ap}_0 t_2[\vec{u}]\quad \text{for every $\vec{u}\in \Lambda(0)^n$}.\]
\end{definition}

\begin{proposition}\label{prop:bisim-vs-appbisim}
Coalgebraic bisimilarity coincides with the open extension of strong applicative bisimilarity:
\[ {\sim^\Lambda}\; =\, {\sim^\ap}. \]
\end{proposition}

\begin{proof}
By \Cref{prop:bisim,prop:gamma1,prop:gamma2}, bisimilarity is the greatest relation $\sim^\Lambda\,\seq \Lambda\times\Lambda$ such that for every $n\in \fset$ and $t_1 \sim^\Lambda_n t_2$ the following conditions hold:
\begin{enumerate}
\item $t_1[r(0),\ldots,r(n-1)] \sim^\Lambda_m t_2[r(0),\ldots,r(n-1)]$ for all $r \c n \to m$;
\item $t_1[\vec{u}] \sim^\Lambda_m t_2[\vec{u}]$ for all $m\in \fset$ and $\vec{u}\in \Lambda(m)^n$;
\item $t_1\to t_1' \implies \exists t_2'.~ t_2\to t_2' \wedge t_1' \sim^\Lambda_n t_2'$;
\item $t_1=\lambda x.t_1' \implies \exists t_2'.t_2=\lambda x.t_2' \wedge \forall e\in \Lambda(n).  t_1'[e/x] \sim^\Lambda_n t_2'[e/x]$;
\item  $\exists x\in V(n), s_1,\ldots,s_k\in \Lambda(n).\, t_1=x\app s_1\app\cdots\app s_k \Longrightarrow \exists y\in V(n), s_1',\ldots,s_m'\in \Lambda(n).\, t_2=y\app s_1'\app\cdots\app s_m'$;
\item $t_2\to t_2' \implies \exists t_1'.~ t_1\to t_1' \wedge t_1'\sim^\Lambda_n t_2'$;
\item $t_2=\lambda x.t_2' \implies \exists t_1'.t_1=\lambda x.t_1' \wedge \forall e\in \Lambda(n).  t_1'[e/x] \sim^\Lambda_n t_2'[e/x]$;
\item  $\exists y\in V(n), s_1',\ldots,s_m'\in \Lambda(n).\, t_2=y\app s_1'\app\cdots\app s_m' \Longrightarrow  \exists x\in V(n), s_1,\ldots,s_k\in \Lambda(n).\, t_1=x\app s_1\app\cdots\app s_k$. 
\end{enumerate}
Note that condition~(1) is redundant, as it follows from~(2) by putting $\vec{u}=\var_m\comp r$.

\medskip\noindent
\emph{Proof of $\sim^\Lambda \,\seq\, \sim^\ap$.}  Note first that $\sim^\Lambda_0\,\seq\, \Lambda(0)\times \Lambda(0)$ is a strong applicative bisimulation: the above conditions (3), (4), (6), (7) for $n=0$ correspond precisely to (A1)--(A4) with $\sim^\ap_0$ replaced by $\sim^\Lambda_0$. It follows that $\sim^\Lambda_0\,\seq\, \sim^\ap_0$ because $\sim^\ap_0$ is the greatest strong applicative bisimulation. Moreover, for $n>0$ and $t_1\sim^\Lambda_n t_2$, we have
\[ t_1[\vec{u}] \sim^\Lambda_0 t_2[\vec{u}]\quad \text{for every $\vec{u}\in \Lambda(0)^n$}\]
by condition (2), whence 
\[ t_1[\vec{u}] \sim^\ap_0 t_2[\vec{u}]\quad \text{for every $\vec{u}\in \Lambda(0)^n$}\]
because $\sim^\Lambda_0\,\seq\, \sim^\ap_0$, and so $t_1\sim^\ap_n t_2$. This proves $\sim^\Lambda_n \,\seq\, \sim^\ap_n$ for $n>0$ and thus $\sim^\Lambda\,\seq\, \sim^\ap$ overall.

\medskip\noindent
\emph{Proof of $\sim^\ap \,\seq\, \sim^\Lambda$.} Since $\sim^\Lambda$ is the greatest bisimulation, it suffices to show that $\sim^\ap$ is a bisimulation. Thus suppose that $n\in \fset$ and $t_1\sim^\ap_n t_2$; we need to verify the above conditions (2)--(8) with $\sim^\Lambda_n$ replaced by $\sim^\ap_n$. Let us first consider the case $n=0$:
\begin{enumerate}\addtocounter{enumi}{1}
\item Since $t_1$ and $t_2$ are closed terms, this condition simply states that $t_1\sim^\ap_m t_2$ for every $m>0$. This holds by definition of $\sim^\ap_m$ because $t_1[\vec{u}]=t_1 \sim^\ap_0 t_2=t_2[\vec{u}]$ for every $\vec{u}\in \Lambda(0)^m$. 
\item holds by (A1).
\item holds by (A2). 
\item holds vacuously because $t_1$ is a closed term. 
\item holds by (A3).
\item holds by (A4). 
\item holds vacuously because $t_2$ is a closed term. 
\end{enumerate} 
Now suppose that $n>0$: 
\begin{enumerate}
\item[(2)] Let $\vec{u}=(u_0,\ldots,u_{n-1})\in \Lambda(m)^n$. If $m=0$ we have $t_1[\vec{u}]\sim^\ap_0 t_2[\vec{u}]$ by definition of $\sim^\ap_n$. If $m>0$ and   $\vec{v}\in \Lambda(0)^m$ we have
\[ t_1[\vec{u}][\vec{v}] = t_1[u_0[\vec{v}],\ldots, u_{n-1}[\vec{v}]] \sim^\ap_0 t_2[u_0[\vec{v}],\ldots, u_{n-1}[\vec{v}]] = t_2[\vec{u}][\vec{v}], \]
whence $t_1[\vec{u}]\sim^\ap_m t_2[\vec{u}]$.   
\item[(3)] Suppose that $t_1\to t_1'$. We only need to prove that $t_2$ reduces, that is, $t_2\to t_2'$ for some $t_2'\in \Lambda(n)$. Then, for every $\vec{u}\in \Lambda(0)^n$ we have  $t_1[\vec{u}]\sim^\ap_0 t_2[\vec{u}]$ by definition of $\sim^\ap_n$, and $t_1[\vec{u}]\to t_1'[\vec{u}]$ and $t_2[\vec{u}]\to t_2'[\vec{u}]$ because reductions respect substitution. Therefore $t_1'[\vec{u}] \sim^\ap_0 t_2'[\vec{u}]$ by (A1), which proves $t_1'\sim^\ap_n t_2'$ by definition of $\sim^\ap_n$.

To prove that $t_2$ reduces, suppose towards a contradiction that $t_2$ does not reduce. There are two possible cases:

\medskip\noindent \underline{\emph{Case 1:}} $t_2$ is a $\lambda$-abstraction.\\
Since the term $t_1$ reduces, it is neither a variable nor a $\lambda$-abstraction. Therefore, for arbitrary $\vec{u}\in \Lambda(0)^n$, the term $t_1[\vec{u}]$ is not a $\lambda$-abstraction, whereas the term $t_2[\vec{u}]$ is a $\lambda$-abstraction. Thus $t_1[\vec{u}]\not\sim^\ap_0 t_2[\vec{u}]$ and therefore $t_1\not\sim^\ap_n t_2$, a contradiction.

\medskip\noindent \underline{\emph{Case 2:}} $t_2=x\app s_1\app\cdots\app s_k$ for some $x\in V(n)$ and $s_1,\ldots, s_k\in \Lambda(n)$, $k\geq 0$. \\
Given $\lambda$-terms $s,t$ and $m\geq 0$, we write $s\to^m t$ if $s$ reduces to $t$ in exactly $m$ steps; in particular, $s\to^0 t$ if $s=t$. We shall prove that there exists $\vec{u}\in \Lambda(0)^n$ such that \[t_1[\vec{u}]\to^m \tilde{t_1}\qquad \text{and}\qquad t_2[\vec{u}]\to^m\tilde{t_2} \qquad \text{for some $m\geq 0$ and $\tilde{t_1},\tilde{t_2}\in \Lambda(0)$},\] where exactly one of the terms $\tilde{t_1}$ and $\tilde{t_2}$ is a $\lambda$-abstraction. Then $\tilde{t_1}\not\sim^\ap_0\tilde{t_2}$ by (A1) and (A2), whence $t_1[\vec{u}]\not\sim^\ap_0 t_2[\vec{u}]$ by $m$-fold application of (A1), and so $t_1\not\sim^\ap_n t_2$, a contradiction.

In order to construct $\vec{u}\in \Lambda(0)^n$ with the desired property, we consider several subcases:
  
\medskip\noindent \underline{\emph{Case 2.1:}} $t_1\to^k \ol{t_1}$ for some $\ol{t_1}$.

\medskip\noindent \underline{\emph{Case 2.1.1:}} $\ol{t_1}$ is a $\lambda$-abstraction. \\
Choose $\vec{u}$ such that $u_x\to u_x$ (e.g. $u_x=(\lambda y.y\app y)\app (\lambda y.y\app y)$). Then 
$t_2[\vec{u}] \to^k t_2[\vec{u}]$ and $t_2[\vec{u}]$ is not a~$\lambda$-abstraction, while $t_1[\vec{u}]\to^k \ol{t_1}[\vec{u}]$ and $\ol{t_1}[\vec{u}]$ is a $\lambda$-abstraction.

\medskip\noindent \underline{\emph{Case 2.1.2:}} $\ol{t_1}$ is an application $\ol{t_{1,1}}\app \ol{t_{1,2}}$.\\
Choose $\vec{u}$ such that $u_x=\lambda x_1.\lambda x_2.\ldots \lambda x_k.\lambda y.y$. Then $t_2[\vec{u}]\to^k \lambda y.y$, while $t_1[\vec{u}]\to^k \ol{t_1}[\vec{u}]$ and $\ol{t_1}[\vec{u}]$ is not a~$\lambda$-abstraction.

\medskip\noindent \underline{\emph{Case 2.1.3:}} $\ol{t_1}=x$. \\
Choose $\vec{u}$ such that $u_x=\lambda x_1.\lambda x_2.\ldots \lambda x_k.t$ where $t$ is an arbitrary closed term that is not a~$\lambda$-abstraction. Note that $u_x$ is a $\lambda$-abstraction: Since $t_1$ reduces, we have $t_1\neq x = \ol{t_1}$ and thus necessarily $k>0$. Thus $t_1[\vec{u}]\to^k\ol{t_1}[\vec{u}]=u_x$ and $u_x$ is a $\lambda$-abstraction, while $t_2[\vec{u}]\to^k t$ and $t$ is not a $\lambda$-abstraction.

\medskip\noindent \underline{\emph{Case 2.1.4:}} $\ol{t_1}=y$ for some variable $y\neq x$.\\
Choose $\vec{u}$ such that $u_x=\lambda x_1.\lambda x_2.\ldots \lambda x_k.\lambda y.y$ and $u_y$ is not a $\lambda$-abstraction. Then $t_2[\vec{u}]\to^k \lambda y.y$, while $t_1[\vec{u}]\to^k \ol{t_1}[\vec{u}]=u_y$ and $u_y$ is not a $\lambda$-abstraction.

\medskip\noindent \underline{\emph{Case 2.2:}} $t_1\to^m \ol{t_1}$ for some $m<k$ such that $\ol{t_1}$ does not reduce. (Note that $m\neq 0$ because $t_1\to t_1'$.)

\medskip\noindent \underline{\emph{Case 2.2.1:}} $\ol{t_1}$ is a $\lambda$-abstraction. \\
Choose $\vec{u}$ such that $u_x\to u_x$. Then $t_1[\vec{u}]\to^m\ol{t_1}[\vec{u}]$ and $\ol{t_1}[\vec{u}]$ is a $\lambda$-abstraction, while \mbox{$t_2[\vec{u}]\xrightarrow{m} t_2[\vec{u}]$} and $t_2[\vec{u}]$ is not a $\lambda$-abstraction.

\medskip\noindent \underline{\emph{Case 2.2.2:}} $\ol{t_1} = y\app s_1'\ldots s_l'$ for some variable $y\neq x$ and terms $s_1',\ldots, s_l'$, $l\geq 0$.\\
Choose $\vec{u}$ such that $u_x\to u_x$ and $u_y=\lambda x_1.\lambda x_2.\ldots \lambda x_l.\lambda y.y$. Then $t_1[\vec{u}]\to^m\ol{t_1}[\vec{u}] \to^l \lambda y.y$ while ${t_2}[\vec{u}]\to^{m+1} {t_2}[\vec{u}]$ and ${t_2}[\vec{u}]$ is not a $\lambda$-abstraction.

\medskip\noindent \underline{\emph{Case 2.2.3:}} $\ol{t_1} = x\app s_1'\app\cdots\app s_l'$ for some $l> k-m$ and terms $s_1',\ldots, s_l'$.\\
Choose $\vec{u}$ such that $u_x=\lambda x_1.\lambda x_2.\cdots \lambda x_k.\lambda y.y$. Then $t_2[\vec{u}]\to^k \lambda y.y$, while \[t_1[\vec{u}]\to^m \ol{t_1}[\vec{u}] \to^{k-m} (\lambda x_{k-m+1}.\cdots \lambda x_k.\lambda y.y) \app s_{k-m+1}'[\vec{u}]\app\cdots\app 
  s_{l}'[\vec{u}]\] and $(\lambda x_{k-m+1}.\cdots \lambda x_k.\lambda y.y) \app s_{k-m+1}'[\vec{u}] \app\cdots\app 
  s_{l}'[\vec{u}]$ is not a $\lambda$-abstraction.

\medskip\noindent \underline{\emph{Case 2.2.4:}} $\ol{t_1} = x\app s_1'\app\cdots\app s_l'$ for some $l\leq k-m$ and terms $s_1',\ldots, s_l'$.\\
Choose $\vec{u}$ such that $u_x=\lambda x_1.\lambda x_2.\ldots \lambda x_k.t$ where $t$ is an arbitrary closed term that is not a $\lambda$-abstraction. Then \[t_2[\vec{u}]\to^{m+l} (\lambda x_{m+l+1}.\cdots \lambda x_k.t) \app s_{m+l+1}[\vec{u}] \app\cdots\app
  s_{k}[\vec{u}]\] and  $(\lambda x_{m+l+1}\cdots \lambda x_k.t) \app s_{m+l+1}[\vec{u}] \app\cdots\app 
  s_{k}[\vec{u}]$ is not a $\lambda$-abstraction (for $l=k-m$, this is just the term~$t$), while
  \[
    t_1[\vec{u}] \to^m \ol{t_1}[\vec{u}] xto^{l} \lambda
    x_{l+1}.\cdots\lambda x_k.t
  \]
  and  $\lambda x_{l+1}.\cdots\lambda x_k.t$ is a $\lambda$-abstraction
  since $l<k$. (Recall that $m\neq 0$.)

\item[(6)] holds by symmetry to (3). 
\item[(4)] Suppose that $t_1=\lambda x.t_1'$. Then $t_2$ does not reduce (otherwise $t_1$ reduces by (6), a contradiction). Moreover, $t_2$ cannot be of the form $y\app s_1'\app\cdots\app s_l'$ where $y$ is variable and $s_1',\ldots,s_l'$ are terms. In fact, suppose the  contrary, and choose $\vec{u}\in \Lambda(0)^n$ such that $u_y$ is not a $\lambda$-abstraction. Then $t_1[\vec{u}] \not\sim^\ap_0 t_2[\vec{u}]$ since $t_1[\vec{u}]$ is a $\lambda$-abstraction and $t_2[\vec{u}]$ is not, contradicting $t_1\sim^\ap t_2$.

Thus $t_2=\lambda x.t_2'$ for $x=n$ and $t_2'\in \Lambda(n+1)$. Moreover, for every  $e\in \Lambda(n)$ and $\vec{u}\in \Lambda(0)^n$ we have
\begin{align*} 
&~ t_1'[e/x][\vec{u}] = t_1'[\vec{u},e[\vec{u}]] = t_1'[\vec{u},x][e[\vec{u}]/x] \\ 
\sim^\ap_0 &~ t_2'[\vec{u},x][e[\vec{u}]/x] = t_2'[\vec{u},e[\vec{u}]] = t_2'[e/x][\vec{u}]  
\end{align*}
using (A2) and that $t_1[\vec{u}]\sim^\ap_0 t_2[\vec{u}]$ by definition of $\sim^\ap_n$. This proves $t_1'[e/x] \sim^\ap_n t_2'[e/x]$.
\item[(7)] holds by symmetry to (4).
\item[(5)] Suppose that $t_1=x\app s_1\app\cdots\app s_k$. Then $t_2$ does not reduce by (6) and is not a $\lambda$-abstraction by (7), so it must be of the form  $t_2=y\app s_1'\app\cdots\app s_m'$.
\item[(8)] holds by symmetry to (5).  \hfill$\qed$\par\addvspace{6pt}
\end{enumerate} \def\qed{}
\end{proof}
The above proposition and our general compositionality result (\Cref{th:main}) imply:

\begin{corollary}
  \label{cor:cong}
 The open extension $\sim^\ap$ of strong applicative bisimilarity is a congruence.
\end{corollary}

\begin{proof}
We only need to verify that our present setting satisfies the conditions of \Cref{th:main}:
\begin{enumerate} 
\item The presheaf category $\vcat$ is regular, being a topos.
\item The functor $\Sigma X = \Pt+\delta X + X\times X$ preserves
  reflexive coequalizers. In fact, $\delta$ is a left adjoint (with right adjoint $\llangle V+1,-\rrangle$, see~\cite{DBLP:conf/lics/FiorePT99}) and thus
  preserves all colimits. Moreover, reflexive
  coequalizers commute with finite products and coproducts in~$\Set$,
  hence also in $\vcat$ since limits and colimits are formed
  pointwise.
  
\item The functor $B(X,Y)=\llangle X,Y\rrangle\times (Y+Y^X+1)$ preserves monos. To see this, note first that monos in $\vcat$ are the componentwise injective natural transformations and thus stable under products and coproducts; hence it suffices to show that the functors  $(X,Y)\mapsto Y^X$ and $(X,Y) \mapsto \llangle X,Y \rrangle$ preserve monos.
The first functor preserves monos in the covariant component because $(-)^X\c \vcat\to \vcat$ is a right adjoint, and in the contravariant component because $Y^{(-)}\c (\vcat)^\opp\to \vcat$ is a right adjoint (with left adjoint $(Y^{(-)})^\opp\c \vcat\to (\vcat)^\opp$). 
The second functor preserves monos in the covariant component because $\llangle X,-\rrangle\c \vcat\to \vcat$ is a right adjoint. To see that it preserves monos in the contravariant component, suppose that $f\c X'\to X$ is an epimorphism in $\vcat$. Then $\llangle f,Y\rrangle\c \llangle X,Y\rrangle\to \llangle X',Y\rrangle$ is the natural transformation with components
\[ \llangle f,Y\rrangle_n\c \NT(X^n,Y)\to\NT((X')^n,Y), \qquad g\mapsto  g\comp f^n.\]
This map is clearly monic because $f$ is epic. Thus $\llangle f,Y\rrangle$ is monic in $\vcat$.  \hfill$\qed$\par\addvspace{6pt}
\end{enumerate} \def\qed{}
\end{proof}
\subsection{Call-by-value evaluation}\label{sec:cbv}

Much analogously to the call-by-name
$\lambda$-calculus, we can implement the call-by-value $\lambda$-calculus
(\Cref{fig:cbv}) in higher-order abstract GSOS.
\begin{figure}[t!]
  \[
    \begin{array}{l@{\quad}l@{\quad}l}
     \inference[\texttt{app1}]{ \quad q = \lambda x.\_}
        {\goes{(\lambda x.p) \app q}{p[q/x]}}
      & \inference[\texttt{app2}]{\goes{p}{p\pr}}{\goes{p \app q}{p\pr \app
        q}}
      & \inference[\texttt{app3}]{\goes{q}{q\pr}}{\goes{(\lambda x.p) \app q}{(\lambda x.p) \app q\pr}}
    \end{array}
  \]
  \caption{Small-step operational semantics of the call-by-value $\lambda$-calculus.}
  \label{fig:cbv}
\end{figure}%
The corresponding higher-order GSOS law differs from the one in \Cref{def:lamgsos}
only in the case of application on closed terms.
\begin{definition}[$V$-pointed higher-order GSOS law of the call-by-value $\lambda$-calculus]\label{def:lam-gsos-cbv}
  \[   \rho^{\cv}_{X,Y} \c
 \Sigma(jX \times B(jX,Y))
  \to 
     B(jX, \Sigma^\star (jX+Y))  \]
    is given by
  \begin{align*}
    & \rho^{\cv}_{X,Y,n}(tr) =  \texttt{case}~tr~\texttt{of} &&& \\
    & v \in V(n)
    & \mapsto \quad
    & \pi(v),* & \\
    & \mathsf{\lambda}.(t, f,\_)
    & \mapsto
      \quad
    & \llangle X, \lambda.(-)\comp \eta \comp \inr \rrangle (\rho_{1}(f)),
      (\eta \comp \inr)^{X}(\rho_{2}(f))
    & \\
    & (t_{1}, g, t_{1}\pr) \app (t_{2}, h,\_)
    & \mapsto
      \quad
    & \lambda m,\,\vec{u} \in X(m)^{n}.\, ( g_{m}(\vec{u}) \app h_{m}(\vec{u}) ),t_{1}\pr \app t_{2}
    & \\
    & (t_{1}, g, k) \app (t_{2}, h,t_{2}\pr)
    & \mapsto
      \quad
    & \lambda m,\,\vec{u} \in X(m)^{n}.\, (g_{m}(\vec{u}) \app h_{m}(\vec{u})),t_{1} \app t_{2}\pr
    & \\
    & (t_{1}, g, k) \app (t_{2}, h,k_{2})
    & \mapsto
      \quad
    & \lambda m,\, \vec{u} \in X(m)^{n}.\, (g_{m}(\vec{u}) \app h_{m}(\vec{u})),\eta \comp \inr \comp k(t_{2})
    & \\
    & (t_{1}, g, k) \app (t_{2}, h,\_)
    & \mapsto
      \quad
    & \lambda m,\,\vec{u} \in X(m)^{n}.\, (g_{m}(\vec{u}) \app h_{m}(\vec{u})),\eta \comp \inr \comp k(t_{2})
    & \\
    & (t_{1}, g, *) \app (t_{2}, h,\_)
    & \mapsto
      \quad
    & \lambda m,\,\vec{u} \in X(m)^{n}.\, (g_{m}(\vec{u}) \app h_{m}(\vec{u})),*
    &
  \end{align*}
  where $t\in \delta X(n)$, $f\in \delta\llangle X,Y\rrangle(n)$, $g,h\in
  \llangle X,Y\rrangle(n)$, $k,k_{2}\in
  Y^X(n)$, $t_1, t_2\in X(n)$ and $t_1',t_{2}'\in Y(n)$. Again
  brackets around the pairs on the right are omitted, and the last five clauses refer to the application operator $p\app q = \appp(p,q)$.
\end{definition}
Applying \Cref{th:main} to the call-by-value $\lambda$-calculus shows
that coalgebraic bisimilarity, as expressed in \Cref{prop:bisim}, is a
congruence. Note however that unlike the case for the call-by-name
$\lambda$-calculus, coalgebraic bisimilarity does not correspond to a
strong version of call-by-value applicative
bisimilarity (see e.g.~\cite{pitts_2011}): The former relates terms if they
exhibit the same behaviour when applied to arbitrary closed terms,
while the latter considers only application to \emph{values}. Capturing call-by-value applicative
bisimilarity in the coalgebraic framework is left as an open problem; see also \Cref{sec:concl}.

\subsection{Typed \texorpdfstring{$\boldsymbol{\lambda}$}{𝜆}-calculi}
We have shown in the present section how to implement untyped $\lambda$-calculi in the higher-order abstract GSOS framework, using presheaf models. Higher-order languages with types can be treated in a very similar manner, as demonstrated in \cite{gmstu24} for the case of simple type systems, and in \cite{gmtu24lics} for recursive types. In a nutshell, one moves from the base category $\Set^\fset$ to the category $(\Set^{\fset/\mathsf{Ty}})^{\mathsf{Ty}}$, where $\mathsf{Ty}$ is the set of types (regarded as a discrete category) and $\fset/\mathsf{Ty}$ is the comma category of \emph{typed variable contexts}, i.e.\ pairs $(n,\Gamma)$ of a finite cardinal $n$ and a function $\Gamma\colon n\to \mathsf{Ty}$. Informally, a presheaf $X\in (\Set^{\fset/\mathsf{Ty}})^{\mathsf{Ty}}$ associates to every $\tau\in \mathsf{Ty}$ and $\Gamma\in \fset/\mathsf{Ty}$ a set $X_\tau(\Gamma)$ of terms of type $\tau$ in context $\Gamma$. Using this categorical setup, it is not difficult to devise higher-order GSOS laws for typed $\lambda$-calculi, which are essentially type-indexed versions of \Cref{def:lamgsos} and \Cref{def:lam-gsos-cbv} above. We refer the reader to \emph{op.~cit.} for more details.

\section{Conclusions, further developments, and future work}
\label{sec:concl}

We have introduced the notion of \emph{higher-order GSOS law},
effectively transferring the principles behind the
bialgebraic framework by~\cite{DBLP:conf/lics/TuriP97} to higher-order
languages. We have demonstrated that, under mild assumptions, strong coalgebraic
bisimilarity in systems specified by higher-order GSOS laws is a congruence, a result
guaranteeing the compositionality of semantics within our abstract
framework. In addition, we have implemented combinatory logics as
well as the call-by-name $\lambda$-calculus as
higher-order GSOS laws in suitable categories.

Our compositionality result for strong coalgebraic bisimilarity illustrates that
the higher-order abstract GSOS framework can not only \emph{model} the
operational semantics of higher-order languages, but more importantly also
provides the means to \emph{reason} about such languages at a high level of
abstraction. In recent work, we have further substantiated this point by lifting
several key operational techniques to the categorical generality of our
framework. In \cite{UrbatTsampasEtAl23} we introduce a generalization of
\emph{Howe's method}, see
\cite{DBLP:conf/lics/Howe89,DBLP:journals/iandc/Howe96},
and apply it to derive a general congruence result for (weak) applicative
similarity. In \cite{gmstu24} we study unary logical predicates along with
induction up-to techniques to reason about them efficiently, with proofs of
strong normalization as a key application. Finally, in \cite{gmtu24lics} we
develop a theory of step-indexed logical relations as a sound proof method for
contextual equivalence. We note that all these operational techniques are
usually introduced in an ad hoc manner and need to be carefully adapted for
each individual language. The more principled approach based on higher-order
abstract GSOS provides a clean conceptual separation between their non-trivial,
language-dependent core and their generic, language-independent aspects.

Substantial progress has also been achieved towards supporting different
paradigms of operational semantics. Recall that \Cref{sec:cbv} hints that the naive
treatment of call-by-value languages in the present work is not entirely
satisfactory, as the ensuing coalgebraic notion of bisimilarity does not match
standard applicative bisimilarity. In~\cite{10.1145/3704871}, we
resolve this issue by moving from $\Set$-valued presheaves to $\Set^2$-valued
presheaves, with one sort for values and one for non-value terms. In recent
work~\cite{goncharov2025bialgebraicreasoningstatefullanguages}, a similar
method is used to apply our framework to stateful languages. By doing so,
we also reconcile first-order abstract GSOS with stateful languages,
(a well-known problem, see e.g.
\cite{DBLP:journals/entcs/Abou-SalehP11} and \cite{DBLP:conf/fscd/0001MS0U22}).
Moreover,
\citet{10.1145/3776697} extends the theory of higher-order abstract GSOS to
reason about behavioural conformances in languages with quantitative features.

Last but not least,
in~\cite{goncharov2026higherorderbialgebraicdenotationalsemantics} we provide a
higher-order bialgebraic account of denotational semantics in the style of Turi
and Plotkin, completing the
bialgebraic picture of \Cref{sec:bialg}.

Let us conclude with outlining some directions for future research.
A powerful technique for
compositionality results for effectful languages is given by \emph{environmental
bisimulations}, see~\cite{DBLP:conf/lics/SangiorgiKS07}, which we aim to understand
from the perspective of our categorical approach.
Another goal of interest is to extend the
notion of a \emph{morphism of distributive
  laws}, see~\cite{DBLP:journals/entcs/Watanabe02} and \cite{DBLP:conf/calco/KlinN15},
 to higher-order GSOS laws. As a potential application, this would enable modeling compilers
of higher-order languages that preserve semantic properties across
compilation. In first-order abstract GSOS, this idea has been previously
explored by \cite{DBLP:conf/cmcs/0001NDP20} and
\cite{DBLP:conf/aplas/AbateBT21}. Finally, we aim to develop a fibrational
theory of logical relations in higher-order abstract GSOS, with potential
applications to parametricity~\cite{DBLP:conf/ifip/Reynolds83} and dependent types.

\medskip\noindent\textbf{Acknowledgement.}                          %
  Stelios Tsampas wishes to thank Andreas Nuyts and Christian Williams for the
  numerous and fruitful discussions.

\medskip\noindent\textbf{Funding.} 
The authors acknowledge the following support:
\begin{itemize}
\item Sergey Goncharov is supported by the Deutsche Forschungsgemeinschaft (DFG, German
  Research Foundation) -- project number 527481841.
\item Stefan Milius and Lutz Schröder are supported by the Deutsche Forschungsgemeinschaft (DFG, German
  Research Foundation) -- project number 517924115.
\item  Henning Urbat is supported by the Deutsche Forschungsgemeinschaft (DFG, German
  Research Foundation) -- project number 470467389.
\item Stelios Tsampas is supported by the Deutsche Forschungsgemeinschaft (DFG, German
  Research Foundation) -- project number 419850228 and 527481841.
\end{itemize}

\bibliographystyle{jfplike}
\bibliography{mainBiblio}

\end{document}

%% file: preamble.tex
\usepackage{booktabs}   %
\usepackage{subcaption} %

\usepackage{amsmath}    %

\usepackage{amsthm}
\usepackage{amssymb}    %
\usepackage{pifont}
\usepackage{amsfonts}
\usepackage{graphicx}   %
\usepackage{wrapfig}    %
\usepackage{hyperref}   %
\usepackage[colorinlistoftodos]{todonotes} %
\usepackage{listings}   %
\usepackage{framed}     %
\usepackage[altpo]{backnaur}
\usepackage{semantic}   %
\usepackage{stmaryrd}  %
\usepackage{verbatim}   %
\usepackage{etoolbox}   %
\usepackage{cancel}       %
\usepackage{adjustbox}  %
\usepackage{cutwin}     %
\usepackage{academicons} %
\usepackage[edges]{forest}
\usepackage{microtype}
\usepackage{multirow}
\usepackage{mathtools}
\usepackage{enumerate}
\usepackage{refcount}
\setrefcountdefault{1}
\usetikzlibrary{arrows,arrows.meta} %

\usepackage{tikz}
\usepackage{tikz-cd}    %
\usepackage{tikz-qtree}

\tikzcdset{scale cd/.style={every label/.append style={scale=#1},
    cells={nodes={scale=#1}}}}

\usepackage{centernot} %
\usepackage{thmtools}

\newcommand{\ownthmSpaceAbove}{5pt}
\newcommand{\ownthmSpaceBelow}{5pt}
\newcommand{\resetCurThmBraces}{%
\gdef\curThmBraceOpen{(}%
\gdef\curThmBraceClose{)}}
\resetCurThmBraces
\newcommand{\removeThmBraces}{%
\gdef\curThmBraceOpen{}%
\gdef\curThmBraceClose{}}

\declaretheoremstyle[
    spaceabove=\ownthmSpaceAbove,
    spacebelow=\ownthmSpaceBelow,
    headpunct=.,
    postheadspace=.5em,
    notebraces={\curThmBraceOpen}{\curThmBraceClose},
    postheadhook={\resetCurThmBraces},
]{definition}
\declaretheoremstyle[
    style=definition,
    bodyfont=\itshape,
    notebraces={\curThmBraceOpen}{\curThmBraceClose},
    postheadhook={\resetCurThmBraces},
]{theorem}

    \theoremstyle{theorem}
    \newtheorem{theorem}{Theorem}[section]
    \newtheorem{proposition}[theorem]{Proposition}
    \newtheorem{corollary}[theorem]{Corollary}
    \newtheorem{lemma}[theorem]{Lemma}
    \theoremstyle{definition}
    \newtheorem{example}[theorem]{Example}
    \theoremstyle{definition}
    \newtheorem{definition}[theorem]{Definition}
    \newtheorem{remark}[theorem]{Remark}
\newcommand{\pr}{^\prime}

\newcommand{\product}{\times}

\newcommand{\Set}{\mathbf{Set}}

\newcommand{\xCL}{\textbf{xCL}\xspace}

\newcommand{\Id}{\mathrm{Id}}
\newcommand{\id}{\mathrm{id}}

\newcommand{\strength}{\mathsf{st}}

\newif\ifedit

\edittrue

\ifedit
\newcommand{\ST}[1]{\textcolor{purple}{ST: #1}}
\newcommand{\STin}[1]{\todo[color=purple!30,inline]{Stelios: #1}}
\newcommand{\AN}[1]{\textcolor{brown}{AN: #1}}
\newcommand{\ANin}[1]{\todo[color=yellow!30,inline]{Andreas: #1}}
\newcommand{\DD}[1]{\textcolor{brown}{DD: #1}}
\newcommand{\DDin}[1]{\todo[color=yellow!30,inline]{Dominique: #1}}
\newcommand{\HU}[1]{\textcolor{brown}{DD: #1}}
\newcommand{\HUin}[1]{\todo[color=yellow!30,inline]{Henning: #1}}
\newcommand{\SM}[1]{\textcolor{brown}{DD: #1}}
\newcommand{\SMin}[1]{\todo[color=yellow!30,inline]{Stefan: #1}}
\newcommand{\LS}[1]{\textcolor{brown}{DD: #1}}
\newcommand{\LSin}[1]{\todo[color=yellow!30,inline]{Lutz: #1}}
\else
\newcommand{\ST}[1]{}
\newcommand{\STin}[1]{}
\newcommand{\AN}[1]{}
\newcommand{\ANin}[1]{}
\newcommand{\DD}[1]{}
\newcommand{\DDin}[1]{}
\newcommand{\HU}[1]{}
\newcommand{\HUin}[1]{}
\newcommand{\SM}[1]{}
\newcommand{\SMin}[1]{}
\newcommand{\LS}[1]{}
\newcommand{\LSin}[1]{}
\fi

\newcommand{\arrow}[2]{%
  $\begin{tikzcd}[ampersand replacement=\&] #1 \arrow[r, shift left=0.3ex] \& #2 \end{tikzcd}$}

\newcommand{\goes}[2]{\ensuremath{#1 \rightarrow #2}}

\newcommand{\goesv}[3]{\ensuremath{#1 \xrightarrow{~#3~} #2}}

\newcommand{\set}{\mathbf{Set}}

\newcommand{\Nat}{\mathsf{Nat}}

\newcommand{\coit}{\mathop{\oname{coit}}}

\makeatletter
\DeclareFontFamily{OMX}{MnSymbolE}{}
\DeclareSymbolFont{MnLargeSymbols}{OMX}{MnSymbolE}{m}{n}
\SetSymbolFont{MnLargeSymbols}{bold}{OMX}{MnSymbolE}{b}{n}
\DeclareFontShape{OMX}{MnSymbolE}{m}{n}{
  <-6>  MnSymbolE5
  <6-7>  MnSymbolE6
  <7-8>  MnSymbolE7
  <8-9>  MnSymbolE8
  <9-10> MnSymbolE9
  <10-12> MnSymbolE10
  <12->   MnSymbolE12
}{}
\DeclareFontShape{OMX}{MnSymbolE}{b}{n}{
  <-6>  MnSymbolE-Bold5
  <6-7>  MnSymbolE-Bold6
  <7-8>  MnSymbolE-Bold7
  <8-9>  MnSymbolE-Bold8
  <9-10> MnSymbolE-Bold9
  <10-12> MnSymbolE-Bold10
  <12->   MnSymbolE-Bold12
}{}

\let\llangle\@undefined
\let\rrangle\@undefined
\DeclareMathDelimiter{\llangle}{\mathopen}%
{MnLargeSymbols}{'164}{MnLargeSymbols}{'164}
\DeclareMathDelimiter{\rrangle}{\mathclose}%
{MnLargeSymbols}{'171}{MnLargeSymbols}{'171}
\makeatother

\usepackage{bm} %

\providecommand{\catname}{\mathbf} 
\providecommand{\clsname}{\mathcal}
\providecommand{\oname}[1]{{\operatorname{\mathsf{#1}}}}

\def\defcatname#1{\expandafter\def\csname B#1\endcsname{\catname{#1}}}
\def\defcatnames#1{\ifx#1\defcatnames\else\defcatname#1\expandafter\defcatnames\fi}
\defcatnames ABCDEFGHIJKLMNOPQRSTUVWXYZ\defcatnames

\def\defclsname#1{\expandafter\def\csname C#1\endcsname{\clsname{#1}}}
\def\defclsnames#1{\ifx#1\defclsnames\else\defclsname#1\expandafter\defclsnames\fi}
\defclsnames ABCDEFGHIJKLMNOPQRSTUVWXYZ\defclsnames

\def\defbbname#1{\expandafter\def\csname BB#1\endcsname{{\bm{\mathsf{#1}}}}}
\def\defbbnames#1{\ifx#1\defbbnames\else\defbbname#1\expandafter\defbbnames\fi}
\defbbnames ABCDEFGHIJKLMNOPQRSTUVWXYZ\defbbnames

\def\Set{\catname{Set}}

\providecommand{\argument}{\operatorname{-\!-}}

\DeclareOldFontCommand{\bf}{\normalfont\bfseries}{\mathbf}
\providecommand{\Id}{\operatorname{Id}}

\providecommand{\id}{\mathsf{id}}

\providecommand{\comp}{\mathbin{\circ}}

\providecommand{\xto}[1]{\,\xrightarrow{#1}\,}

\providecommand{\dar}{\kern-1.2pt\operatorname{\downarrow}}	
\providecommand{\uar}{\kern-1.2pt\operatorname{\uparrow}}	

\providecommand{\fst}{\oname{fst}}
\providecommand{\snd}{\oname{snd}}
\providecommand{\pr}{\oname{pr}}

\providecommand{\brks}[1]{\langle #1\rangle}

\providecommand{\inl}{\oname{inl}}
\providecommand{\inr}{\oname{inr}}

\DeclareSymbolFont{Symbols}{OMS}{cmsy}{m}{n}
\DeclareMathSymbol{\iobj}{\mathord}{Symbols}{"3B}
\providecommand{\curry}{\oname{curry}}

\providecommand{\ev}{\oname{ev}}

\usepackage{stmaryrd}

\providecommand{\pacman}[1]{}					                     %

\newcommand{\undefine}[1]{\let #1\relax}					                       %

\providecommand{\noqed}{\def\qed{}}				                 %

\providecommand{\mone}{{\text{\kern.5pt\rmfamily-}\mathsf{\kern-.5pt1}}}

\providecommand{\smin}{\smallsetminus}

\makeatletter
\def\mfix#1{\oname{#1}\@ifnextchar\bgroup\@mfix{}}	       %
\def\@mfix#1{#1\@ifnextchar\bgroup\mfix{}}			           %
\makeatother

\providecommand{\case}[3]{\mfix{case}{\mathbin{}#1}{of}{#2}{\kern-1pt;}{\mathbin{}#3}}

\usepackage{dsfont}

\DeclareMathSymbol{\mathinvertedexclamationmark}{\mathord}{operators}{'074}
\DeclareMathSymbol{\mathexclamationmark}{\mathord}{operators}{'041}
\makeatletter
\newcommand{\raisedmathinvertedexclamationmark}{%
  \mathord{\mathpalette\raised@mathinvertedexclamationmark\relax}%
}
\newcommand{\raised@mathinvertedexclamationmark}[2]{%
  \raisebox{\depth}{$\m@th#1\mathinvertedexclamationmark$}%
}
\makeatother

\newcommand{\Pt}{V}

\newcommand{\Pow}{\mathcal{P}}

\newcommand{\HO}{\mathcal{HO}}

\newcommand{\under}[1]{\lvert#1\rvert}
\newcommand{\SKI}{\mathrm{SKI}}

\newcommand{\ap}{{\mathrm{ap}}}
\newcommand{\var}{\mathsf{var}}
\newcommand{\cn}{{\mathrm{cn}}}
\newcommand{\cv}{{\mathrm{cv}}}

\newcommand{\Sigmas}{\Sigma^{\star}}

\newcommand{\ar}{\mathsf{ar}}
\newcommand{\NT}{\mathrm{Nat}}

\newcommand{\epito}{\twoheadrightarrow}

\newcommand{\seq}{\subseteq}
\newcommand{\ol}{\overline}

\newcommand{\beh}{{\mathsf{beh}}}
\providecommand{\C}{}
\providecommand{\D}{}

\renewcommand{\C}{{\mathbb{C}}}
\renewcommand{\D}{{\mathbb{D}}}

\renewcommand{\id}{{\mathsf{id}}}
\renewcommand{\Nat}{\mathds{N}}
\renewcommand{\NT}{\mathsf{Nat}}

\newcommand{\f}{\oname{f}}

\newcommand{\takeout}[1]{\empty}

\newcommand{\ini}{\iota}
\newcommand{\ter}{\tau}

\DeclareMathOperator{\Coalg}{\mathsf{Coalg}}
\DeclareMathOperator{\Alg}{\mathbf{Alg}}

\renewcommand{\rho}{\varrho}

\newcommand{\opp}{\mathsf{op}}

\newcommand{\mypowfin}{\mathscr{P}_\omega}

\newcommand{\pullbackangle}[2][]{\arrow[phantom,to path={
                     -- ($ (\tikztostart)!1cm!#2:([xshift=8cm]\tikztostart) $)
                        node[anchor=west,pos=0.0,rotate=#2,
                        inner xsep = 0]
                        {\begin{tikzpicture}[minimum
                        height=1mm,baseline=0,#1]
    \draw[-] (0,0) -- (.5em,.5em) -- (0,1em);
                        \end{tikzpicture}}}]{}}

\usepackage{enumitem}
\setlist[enumerate,1]{label=(\arabic*),font=\normalfont,align=left,leftmargin=0pt,labelindent=0pt,listparindent=\parindent,labelwidth=0pt,itemindent=!,topsep=3pt,parsep=0pt,itemsep=3pt,start=1}
\setlist[enumerate,2]{label=(\alph*),font=\normalfont,labelindent=*,leftmargin=*,start=1}
\setlist[itemize]{labelindent=*,leftmargin=*}
\setlist[description]{labelindent=*,leftmargin=*,itemindent=-1 em}

\renewcommand{\comp}{\cdot}
\renewcommand{\c}{\colon}

 \usepackage[inline]{showlabels}
\usepackage{seqsplit}
\usepackage{xstring}
\usepackage{xcolor}
\usepackage{bbding}


\usepackage{graphicx,csquotes}
\hypersetup{
  final=true,
  bookmarks=true,
  colorlinks=true,
  hidelinks,
  linkcolor=black,
  citecolor=black,
  filecolor=black,
  urlcolor=black,
  pdftitle={Higher-Order Bialgebraic Semantics},
  pdfauthor={Sergey Goncharov, Stefan Milius, Lutz Schröder, Stelios Tsampas, Henning Urbat},
  pdfkeywords={Abstract GSOS, Categorical Semantics, Higher-order Operational Semantics},
  pdfduplex={DuplexFlipLongEdge},
}

\tikzstyle{shiftarr}=[
        rounded corners,%
        to path={--([#1]\tikztostart.center)
                     -- ([#1]\tikztotarget.center) \tikztonodes
                     -- (\tikztotarget)},
]

\tikzset{
    commutative diagrams/.cd,
    arrow style=tikz,
    diagrams={>=stealth},
    row sep=large,
    column sep = huge
}

\usetikzlibrary{decorations.pathmorphing}

\usepackage{hypcap}
\usepackage[footnote,marginclue,nomargin]{fixme}
\FXRegisterAuthor{hu}{ahu}{HU} %
\FXRegisterAuthor{sm}{asm}{SM} %
\FXRegisterAuthor{st}{ast}{ST} %
\FXRegisterAuthor{ls}{als}{LS} %
\FXRegisterAuthor{sg}{asg}{SG} %

\usepackage{xspace}

\theoremstyle{definition}
\newtheorem{rem}[theorem]{Remark} %
\DeclareMathOperator{\parto}{\rightrightarrows}
\numberwithin{equation}{section}

\newcommand{\xra}[1]{\xrightarrow{~#1~}}
\renewcommand{\xto}{\xra}
\let\xmpsto=\xmapsto
\renewcommand{\xmapsto}[1]{\xmpsto{~#1~}}

\newcommand{\V}{\mathcal{V}}
\newcommand{\iotaq}{\iota_\sim}

\renewcommand{\Nat}{\mathbb{N}}

\newcommand{\fset}{\mathbb{F}}
\newcommand{\vcat}{\set^{\fset}}
\newcommand{\alg}[1]{\mathbf{Alg}(#1)}

\newcommand{\mon}{\bullet}

\newcommand{\gcat}{\mathbb{C}}

\providecommand{\DiNat}{\mathsf{DiNat}}
\newcommand{\mS}{{\mu\Sigma}}
\newcommand{\mSq}{{\mS_{\sim}}}
\newcommand{\BmS}[1]{B(\mS,#1)}
\newcommand{\BmSf}[1]{B(\id,#1)}
\DeclareMathOperator{\iter}{\oname{it}}
\DeclareMathOperator{\coiter}{\oname{coit}}

\newcommand{\appp}{\mathsf{app}}
\newcommand{\app}{\,}

\newcommand{\skitermu}{\Lambda_{\mathsf{\tiny CL}}}
\newcommand{\finalc}{Z}

\theoremstyle{definition}
\newtheorem{notation}[theorem]{Notation}
\newtheorem{assumptions}[theorem]{Assumptions}